%% file: main.tex
\RequirePackage{pdfmanagement}
\SetKeys[document/metadata]{pdfversion=1.7, pdfstandard=A-3u, lang=en-US}
\RequirePackage{cmap}[2021-02-06]

\documentclass[a4paper,USenglish,autoref,thm-restate,numberwithinsect,pdfa]{lipics-v2021}

\usepackage[T1]{fontenc}[2021-04-29]
\usepackage{csquotes}[2024-04-04]
\usepackage{microtype}[2024-03-29]

\usepackage{cite}
\usepackage{amsmath,amssymb,amsthm,amsfonts,mathtools}
\usepackage{bm}
\usepackage{doi}[2018-09-09]
\usepackage{ebproof}[2021-01-28]
\usepackage{thmtools}[2023-05-04]
\usepackage[svgnames]{xcolor}[2024-09-29]
\usepackage{xparse}[2025-10-09]
\usepackage{xspace}[2014-10-28]
\usepackage{stmaryrd}[1994-03-03]
\usepackage{multirow}
\usepackage{makecell}
\usepackage{booktabs}

\usepackage{gobble-user}[2019-01-04]
\usepackage{hyperref}[2024-11-05]
\usepackage{mathpartir}[2016-02-24]
\usepackage{mathcommand}[2021-06-07]
\usepackage{mleftright}[2019-12-03]
\usepackage{tikz}
\usepackage{newunicodechar}[2018-04-08]
\usepackage{usebib}[2012-03-17]
\usepackage{varwidth}[2009-03-30]
\usepackage{wrapfig}
\usepackage[titleref, user]{zref}[2023-09-14]
\usepackage{zref-clever}[2024-09-30]

\allowdisplaybreaks

\AddToHook{cmd/caption/before}{\ifvmode\else\par\fi\removelastskip}

\NewExpandableDocumentCommand\texthorizbar{}{―}
\DeclareTextFontCommand\strong{\bfseries\mathversion{bold}}

\DeclareMathSymbol\mathexclam{\mathord}{operators}{33}
\NewDocumentCommand\Or{}{\:\mathrel{\text{or}}\:}
\DeclareDocumentCommand\And{}{\:\mathrel{\text{and}}\:}

\newmuskip\punctmuskip
\newcount\punctpenalty
\NewDocumentCommand\DeclareMathPunct{mm}{
  \begingroup
  \lccode`~=`#1
  \lowercase{\endgroup\def~}%
    {\mathpunct{#2}\unskip\unpenalty\penalty\punctpenalty\mskip\punctmuskip}
  \AddToHook{begindocument/end}
    {\mathchardef#2=\mathcode`#1\mathcode`#1=\string'100000\relax}
}
\DeclareMathPunct{,}{\mathcomma}
\DeclareMathPunct{.}{\mathperiod}

\pdfglyphtounicode{tfm:cmex10/parenleftbig}      {0028} 
\pdfglyphtounicode{tfm:cmex10/parenrightbig}     {0029} 
\pdfglyphtounicode{tfm:cmex10/braceleftbig}      {007B} 
\pdfglyphtounicode{tfm:cmex10/bracerightbig}     {007D} 
\pdfglyphtounicode{tfm:cmex10/angbracketleftbig} {27E8} 
\pdfglyphtounicode{tfm:cmex10/angbracketrightbig}{27E9} 
\pdfglyphtounicode{tfm:cmex10/vextendsingle}     {007C} 
\pdfglyphtounicode{tfm:cmex10/parenleftBig}      {0028} 
\pdfglyphtounicode{tfm:cmex10/parenrightBig}     {0029} 
\pdfglyphtounicode{tfm:cmex10/braceleftbigg}     {007B} 
\pdfglyphtounicode{tfm:cmex10/parenleftBigg}     {0028} 
\pdfglyphtounicode{tfm:cmex10/parenrightBigg}    {0029} 
\pdfglyphtounicode{tfm:cmex10/braceleftBigg}     {007B} 
\pdfglyphtounicode{tfm:cmex10/summationtext}     {2211} 
\pdfglyphtounicode{tfm:cmex10/producttext}       {220F} 
\pdfglyphtounicode{tfm:cmmi10/arrowhookleft}     {21AA} 
\pdfglyphtounicode{tfm:cmsy7/mapsto}             {21A6} 
\pdfglyphtounicode{tfm:cmsy10/mapsto}            {21A6} 

\RenewDocumentCommand\mktextquote{m +m m m m m}{#1#2#4#5#3#6}

\ebproofset{
  center=false,
  label separation=\labelsep,
  rule margin=\fboxsep,
  template={%
    \normalfont
    \raisebox
      {\totalheight-\height-\dp\strutbox}
      {$\UseName{m@th}\inserttext$}%
    \mathstrut
  },
}
\AddToHook{env/prooftree/begin}{\let\&=&}

\newcounter{rule}

\zcRefTypeSetup{rule}{
  Name-sg=Rule,
  name-sg=rule,
  Name-pl=Rules,
  name-pl=rules,
  reffont=\mdseries\upshape,
}

\letcs\NewEbproofStatement{__ebproof_new_statement:nnn}
\newdimen\treeaxis
\ebproofnewstyle{caption}{
  proof style=downwards,
  template={%
    \llap{\makebox[\treeaxis][l]
      {\normalfont\small\strong{\inserttext}}\mathstrut}%
  },
  rule margin=\smallskipamount,
  rule style=no rule,
}
\NewEbproofStatement{caption}{>{\TrimSpaces}m e\label}{
  \AddToHookNext{env/prooftree/end}{
    \GetTitleStringExpand{#1}
    \global\cslet{@currentlabel}\GetTitleStringResult
    \global\cslet{@currentlabelname}\GetTitleStringResult
    \global\csdef{@currentcounter}{rule}
    \rewrite{\global\treeaxis=\treemark{axis}\copy\treebox}
    \infer[caption]1{\IfValueT{#2}{\stepcounter{rule}\UseName{ltx@label}{#2}}#1}
  }
}
\RenewDocumentEnvironment{prooftree*}{O{} b}
  {}
  {\[\begin{prooftree}[#1]#2\end{prooftree}\]\ignorespacesafterend}
\NewDocumentCommand\by{m r()}{\infer#1[\zcref[noname]{#2}]}

\pdfglyphtounicode{tfm:eufm10/M}{1D510}

\DeclareFontFamily{U}{fontawesome5}{}
\DeclareFontShape{U}{fontawesome5}{b}{n}{<-> s * [0.8333] fa5free2solid}{}
\DeclareSymbolFont{fa5symbol}{U}{fontawesome5}{b}{n}
\DeclareMathSymbol{\MathFaLockOpen}{\mathord}{fa5symbol}{26}
\pdfglyphtounicode{tfm:fa5free1regular/envelope}{2709}
\pdfglyphtounicode{tfm:fa5free2solid/lock-open}{1F513}
\DeclareDocumentMathCommand\faLockOpen{}{\MathFaLockOpen}

\graphicspath{{images/}}

\def\kpSymbolsC{jkplsyc}
\LoadFontDefinitionFile{U}{jkplsyc}
\DeclareSymbolFont{kpSymbolsC}{U}{\kpSymbolsC}{m}{n}
\SetSymbolFont{kpSymbolsC}{bold}{U}{\kpSymbolsC}{bx}{n}
\pdfglyphtounicode{tfm:\kpSymbolsC/medcircle}{26AC}
\DeclareMathSymbol{\mdsmwhtcircle}{\mathord}{kpSymbolsC}{7}

\input{listings-lambda-calculi.prf}
\NewDocumentCommand\code{}
  {\lstinline[language=lambda-calculi, basicstyle=\ttfamily]}

\DeclareExpandableDocumentCommand\MathparLineskip{}
  {\deflength{\lineskiplimit}{\abovedisplayskip}\relax
  \deflength{\lineskip}{\abovedisplayskip}}
\ExpandArgs{c}
\apptocmd\MathparBindings{\UseName{@restorepar}}{}{}

\DeclarePairedDelimiter{\parens}{\lparen}{\rparen}
\DeclarePairedDelimiter{\brackets}{\lbrack}{\rbrack}
\DeclarePairedDelimiter{\Brackets}{\llbracket}{\rrbracket}

\DeclarePairedDelimiter{\tuple}{\langle}{\rangle}
\DeclarePairedDelimiter\abs\lvert\rvert

\DeclareMathOperator*\pos{pos}
\DeclareMathOperator*\restricted{\upharpoonright}

\DeclareMathOperator*\FC{\mathbf{FC}}
\DeclareMathOperator*\FV{\mathbf{FV}}
\DeclareMathOperator*\DomC{\mathbf{Dom}\sb{\mathrm{C}}}

\microtypesetup{stretch=30, shrink=40}

\AddToHook{begindocument}{\mleftright}

\newunicodechar{Γ}{\Gamma}
\newunicodechar{Δ}{\Delta}
\newunicodechar{Θ}{\Theta}
\newunicodechar{Λ}{\Lambda}
\newunicodechar{Π}{\Pi}
\newunicodechar{Σ}{\Sigma}
\newunicodechar{α}{\alpha}
\newunicodechar{β}{\beta}
\newunicodechar{γ}{\gamma}
\newunicodechar{δ}{\delta}
\newunicodechar{λ}{\lambda}
\newunicodechar{ρ}{\rho}
\newunicodechar{σ}{\sigma}
\newunicodechar{ℒ}{\mathcal{L}}
\newunicodechar{→}{\rightarrow}
\newunicodechar{↦}{\mapsto}
\newunicodechar{↪}{\hookrightarrow}
\newunicodechar{⇒}{\Rightarrow}
\newunicodechar{∀}{\forall}
\newunicodechar{∃}{\exists}
\newunicodechar{∅}{\emptyset}
\newunicodechar{∈}{\in}
\newunicodechar{∉}{\notin}
\newunicodechar{∋}{\ni}
\newunicodechar{∑}{\sum}
\newunicodechar{∖}{\setminus}
\newunicodechar{∗}{\ast}
\newunicodechar{∣}{\mid}
\newunicodechar{∧}{\land}
\newunicodechar{∨}{\lor}
\newunicodechar{∪}{\cup}
\newunicodechar{≔}{\coloneq}
\newunicodechar{≠}{\neq}
\newunicodechar{≡}{\equiv}
\newunicodechar{≼}{\preccurlyeq}
\newunicodechar{≽}{\succcurlyeq}
\newunicodechar{⊆}{\subseteq}
\newunicodechar{⊊}{\subsetneq}
\newunicodechar{⊑}{\sqsubseteq}
\newunicodechar{⊒}{\sqsupseteq}
\newunicodechar{⊢}{\vdash}
\newunicodechar{⊥}{\bot}
\newunicodechar{⊨}{\vDash}
\newunicodechar{⊩}{\Vdash}
\newunicodechar{⊬}{\nvdash}
\newunicodechar{⊭}{\nvDash}
\newunicodechar{⊮}{\nVdash}
\newunicodechar{⊴}{\unlhd}
\newunicodechar{⊵}{\unrhd}
\newunicodechar{⋅}{\cdot}
\newunicodechar{□}{\Box}
\newunicodechar{◦}{\circ}
\newunicodechar{⟨}{\langle}
\newunicodechar{⟩}{\rangle}
\newunicodechar{⨿}{\amalg}
\newunicodechar{⩴}{\Coloneq}
\newunicodechar{⪯}{\preceq}
\newunicodechar{⪰}{\succeq}
\newunicodechar{}{\mathord{\blacktriangleleft}}
\newunicodechar{}{\mathord{\blacktriangleright}}
\newunicodechar{𝔐}{\mathfrak{M}}

\LoadFontDefinitionFile{U}{stmry}
\DeclareFontShape{U}{stmry}{b}{n}{<-> ssub * stmry/m/n}{}
\pdfglyphtounicode{llbracket}{27E6}
\pdfglyphtounicode{rrbracket}{27E7}
\pdfglyphtounicode{llparenthesis}{2987}
\pdfglyphtounicode{rrparenthesis}{2988}

\UseName{clist_map_inline:nn}
  {theorem, lemma, corollary, definition}
  {\csdef{theH#1}{\UseName{the#1}}}

\begingroup
  \theoremstyle{claimstyle}
  \newtheorem*{case}{Case}
  \newtheorem*{subcase}{Subcase}
\endgroup

\usetikzlibrary{
  arrows,
  cd,
  decorations.pathmorphing,
}
\tikzset{
  >=stealth',
  line around/.style={decoration={pre length=#1, post length=#1}},
  stage/.style={
    decorate,
    decoration={snake, segment length=0.8em, amplitude=0.4ex},
    line around=0.3333em,
  },
  stage/.value forbidden,
  Transition/.pic={\draw[->, #1] (0, 0) to +(2em, 0);},
  Scope/.pic={\pic{Transition};},
  Stage/.pic={\pic{Transition={stage, decoration={pre length=0.2em}}};},
}
\tikzcdset{
  arrow style=tikz,
  arrows={semithick},
}

\newbibfield{author}
\bibinput{bibliography}

\NewDocumentCommand\Textcite{}{\textcite}

\ExplSyntaxOn
\regex_const:Nn \c_bmtt_bibtex_author_delimiter_regex { \s+ (?i) (?:and) \s+ }
\int_const:Nn \c_bmtt_max_cite_names_int { 2 }
\cs_new:Nn \bmtt_seq_to_sentence:N
  {
    \seq_use:Nnnn #1
      { ~ and \nobreakspace }
      { , ~ }
      { , ~ and \nobreakspace }
  }
\NewDocumentCommand \citeauthor { m !t. }
  {
    \group_begin:
      \exp_args:NNe \regex_split:NnN
        \c_bmtt_bibtex_author_delimiter_regex
        { \usebibentry { #1 } { author } }
        \l_tmpa_seq
      \seq_set_map:NNn \l_tmpa_seq \l_tmpa_seq { \clist_item:nn { ##1 } { 1 } }
      \int_compare:nTF
        { \seq_count:N \l_tmpa_seq > \c_bmtt_max_cite_names_int }
        {
          \seq_item:Nn \l_tmpa_seq { 1 } \nobreakspace
          et \tex_penalty:D \tex_exhyphenpenalty:D \c_space_tl
          al . \IfBooleanF{#2}{\@}
        }
        { \bmtt_seq_to_sentence:N \l_tmpa_seq }
    \group_end:
  }
\NewDocumentCommand \textcite { m }
  {
    \group_begin:
      \seq_set_from_clist:Nn \l_tmpa_seq { #1 }
      \seq_set_map:NNn \l_tmpa_seq \l_tmpa_seq
        {
          \citeauthor { ##1 } \nobreakspace \cite { ##1 }
        }
      \bmtt_seq_to_sentence:N \l_tmpa_seq
    \group_end:
  }
\ExplSyntaxOff

\UseName{zref@setdefault}{\textcolor{red}{\textbf{??}}}
\ztitlerefsetup{expand=true, cleanup={\let\ref\zref}}
\zcsetup{
  hyperref,
  nameinlink=false,
  lastsep={,\space and\nobreakspace},
  rangesep={\textendash},
  rangetopair=false,
}
\UseName{zref@newprop}{claim}{}
\zcRefTypeSetup{item}{
  Name-sg=,
  name-sg=,
  Name-pl=,
  name-pl=,
  refbounds={, \lparen, \rparen, },
}

\NewDocumentCommand\zcnameref{O{}}{\zcref*[nocap, noref, noabbrev, #1]}
\NewDocumentCommand\Zcref{s O{}}{\IfBooleanTF{#1}{\zcref*}{\zcref}[cap, #2]}
\NewDocumentCommand\Zcsubref{O{}}{\zcref[ref=subref, noname, #1]}

\NewExpandableDocumentCommand\logic{m}{\mathbf{#1}}
\def\logicS4{\logic{S4}}
\def\IS4{\logic{IS4}}
\def\CS4{\logic{CS4}}
\NewExpandableDocumentCommand\LTL{}{\logic{LTL}}
\NewExpandableDocumentCommand\PLL{}{\logic{PLL}}
\NewExpandableDocumentCommand\CML{}{\logic{CML}}
\NewExpandableDocumentCommand\BML{}{\logic{BML}}
\NewExpandableDocumentCommand\theLogic{}{\BML}

\NewDocumentCommand\Scope{}{\tikz[baseline=-axis_height] \pic{Scope};\xspace}
\NewDocumentCommand\Stage{}{\tikz[baseline=-axis_height] \pic{Stage};\xspace}
\NewDocumentCommand\tp{}{
  \mkern1mu\nonscript\mkern1mu
  \boldsymbol{\mathexclam}
  \mkern1mu\nonscript\mkern1mu
}
\NewDocumentCommand\bounded{m m}{\binder{#1} \within #2}
\NewDocumentCommand\cls{!>{\SplitArgument1{i>}}d()}{\bounded#1}
\NewDocumentCommand\var{>{\SplitArgument1:}r() !e@}
  {\binder{\firstoftwo#1} \has@{#2} \secondoftwo#1}
\NewDocumentCommand\poly{!d()}{∀\IfValueT{#1}{\cls(#1).}}
\NewDocumentCommand\open{!d()}
  {\IfValueT{#1}{\cramped{\sp{\mkern2mu\cls(#1)}}}}
\NewDocumentCommand\hasType{!e^}{
  \nobreak \mskip 5mu minus 4mu \mathord{:}
  \IfValueT{#1}{\sp{#1\mkern-5mu}}
  \mskip 6mu plus 1mu minus 3mu \relax
}
\NewDocumentCommand\has{!e@}{\IfValueTF{#1}{\hasType^{\binder{#1}}}{\hasType}}
\NewDocumentCommand\given{O{\UseName{delimsize}}}{
  \mathrel{}
  \nolinebreak \mathclose{}
  \nolinebreak #1 \vert
  \nolinebreak \mathopen{}
  \nolinebreak \mathrel{}
}
\newmuskip\bmttglue
\NewDocumentCommand\quo{>{\SplitArgument1{i>}}r() m}
  {\mathbf{quo}\nobreak
    \{^{\bounded #1} \penalty\punctpenalty \mskip\bmttglue #2 \}}
\NewDocumentCommand\unq{e@ m}
  {\mathbf{unq}\nobreak
    \{^{#1} \penalty\punctpenalty \mskip\bmttglue #2 \}}
\NewDocumentCommand\assign{m m m m}
  {\IfValueT{#2}{\binder{#2} \coloneq #4,} \binder{#1} \coloneq #3}
\NewDocumentCommand\subst{>{\SplitArgument1@}m >{\SplitArgument1@}m}{\assign#1#2}
\NewDocumentCommand\w{>{\SplitArgument1/}r[]}{\brackets{\subst#1}}

\NewDocumentCommand\predicate{m !d[]}{\mathcal{#1}\IfValueT{#2}{\Brackets{#2}}}

\NewExpandableDocumentCommand\Parens{m}
  {\lparen\mkern-\medmuskip\lparen #1 \rparen\mkern-\medmuskip\rparen}
\NewDocumentCommand\toF{r||}{\abs*{#1}}
\NewDocumentCommand\toB{r||}{#1^{⪰\,\tp}}

\NewDocumentEnvironment{parened}{b}
  {}
  {
    \left\lparen
      \varwidth\textwidth
        $\begin{multlined}#1\end{multlined}$
      \endvarwidth
    \right\rparen
  }
\NewDocumentCommand\NewTheoremEnvironment{m}{
  \NewDocumentEnvironment{#1*}
    {>{\TrimSpaces}O{} >{\TrimSpaces}e\label +b}
    {
      \ifhmode\unskip\unskip\fi
      \IfValueF{##2}{\undefined}%
      \UseName{zref@localaddprops}{main}{claim}%
      \UseName{zref@setcurrent}{claim}{\RestateClaim{##2}}%
      \begin{#1}[
        name={##1},
        label={##2},
        restate={
          [
            name={%
              \IfBlankF{##1}{##1\space\lparen}%
              On\space\zcpageref{##2}%
              \IfBlankF{##1}{\rparen}%
            }
          ]bmtt_restate_##2%
        }
      ]
        ##3
      \end{#1}
    }{\ignorespacesafterend}
}
\NewTheoremEnvironment{theorem}
\NewTheoremEnvironment{lemma}
\newtheorem*{notation*}{Notation}

\makeatletter
\NewDocumentEnvironment{inparaenum}{}
  {%
    \global\advance\@enumdepth\@ne
    \ExpandArgs{e}\usecounter{enum\@roman{\the\@enumdepth}}%
    \def\@itemlabel{\UseName{label\@listctr}}%
    \def\@item[##1]{\refstepcounter\@listctr\makelabel{##1}\ignorespaces}%
    \def\makelabel##1{##1\nobreakspace}%
    \ignorespaces
  }
  {\global\advance\@enumdepth\m@ne\ignorespacesafterend}
\makeatother

\NewDocumentEnvironment{conditions}{}
  {\left\lbrace\varwidth\linewidth\list
    {\strut\nolinebreak
      \UseName{labelitem\romannumeral\inteval{\UseName{@itemdepth}+1}}}
    {\leftmargin1.5em\labelwidth1.5em\parsep0pt\itemsep\smallskipamount}}
  {\unskip\endlist\endvarwidth\right.}

\NewDocumentCommand\RestateClaim{m}{%
  \ifx\csname bmtt_restate_#1\endcsname\relax
    \todo{#1}%
  \else
    \UseName{bmtt_restate_#1}*%
  \fi
  \ignorespaces
}

\NewDocumentCommand\todo{m}{\textcolor{red}{#1}}

\NewDocumentCommand\lamcirc{}{\ensuremath{\lambda^{\mdsmwhtcircle}}\xspace}
\NewDocumentCommand\lambox{}{\ensuremath{\lambda^{\Box}}\xspace}
\NewDocumentCommand\lamalpha{}{\ensuremath{\lambda^{\alpha}}\xspace}
\NewDocumentCommand\lamtri{}{\ensuremath{\lambda^{\triangleright}}\xspace}

\NewExpandableDocumentCommand\bml{}{$\BML$\xspace}

\NewDocumentCommand\ctext{m}{\raise0.2ex\hbox{\textcircled{\scriptsize{#1}}}}

\input{bmtt.tex}

\title{Bounded Modal Logic} 
\subtitle{Explicit Scope Dependencies in Multi-Stage Programming}

\author{Yuito Murase}{Kyoto University, Japan}{murase@fos.kuis.kyoto-u.ac.jp}{https://orcid.org/0000-0001-6038-6249}{Supported by JSPS KAKENHI Grant Number JP24KJ1363.}

\author{Akinori Maniwa}{Unaffiliated}{akinorimaniwa@outlook.com}{https://orcid.org/0009-0004-4872-4785}{}

\authorrunning{Y. Murase and A. Maniwa} 

\Copyright{Yuito Murase and Akinori Maniwa} 

\ccsdesc[500]{Theory of computation~Modal and temporal logics}

\keywords{%
  modal logic,
  modal lambda calculus,
  multi-stage programming%
} 

\category{} 

\relatedversion{} 

\nolinenumbers 

\EventEditors{John Q. Open and Joan R. Access}
\EventNoEds{2}
\EventLongTitle{42nd Conference on Very Important Topics (CVIT 2016)}
\EventShortTitle{CVIT 2016}
\EventAcronym{CVIT}
\EventYear{2016}
\EventDate{December 24--27, 2016}
\EventLocation{Little Whinging, United Kingdom}
\EventLogo{}
\SeriesVolume{42}
\ArticleNo{23}

\begin{document}

\maketitle

\begin{abstract}
It is widely known that proof systems for modal logic can be interpreted as type systems for multi-stage programming (MSP).  However, existing modal-logical foundations for MSP do not fully account for staged programs with complex scoping structures.  For example, a modal account of cross-stage persistence, in which free variables in generated code may refer to runtime bindings, has not yet been fully established.

This paper presents \emph{Bounded Modal Logic} (\bml), a constructive modal
logic with modalities bounded by names for scopes and first-order-style
quantification over those names. This makes scope dependencies of code
fragments explicit, thereby enabling reasoning about staged programs with
nontrivial scoping behavior, including cross-stage persistence.

We present a natural-deduction system and a Kripke semantics for \bml, and prove soundness and completeness. Its Curry--Howard calculus embeds the standard S4- and LTL-based staging calculi while also typing code that uses cross-stage persistence and is later executed by \texttt{run}. We further establish standard metatheoretic properties and a staged semantics supporting stage-by-stage execution.
\end{abstract}

\section{Introduction}\label{sec:intro}
\paragraph*{Modal Foundations for MetaML-style Multi-Stage Programming}
Multi-stage programming (MSP) constructs and executes code fragments at runtime; MetaML~\cite{journals/tcs/TahaS00,conf/popl/TahaN03} is a canonical calculus for this style of programming.
\begin{figure}[t]
  \vspace{-1.0\baselineskip}
  \begin{subfigure}{0.5\textwidth}
    \centering
    \includegraphics[page=1]{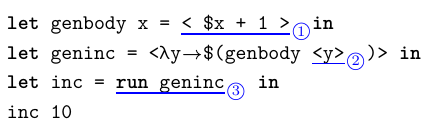}
    \captionsetup{skip=2pt}
    \caption{A staged program with quasi-quotation and \code|run|.}
    \label{fig:metaml-example}
  \end{subfigure}\hfill
  \begin{subfigure}{0.45\textwidth}
    \centering
    \includegraphics[page=2]{images/metaml-example.pdf}
    \captionsetup{skip=2pt}
    \caption{A staged program with cross-stage persistence.}
  \label{fig:metaml-example-csp}
  \end{subfigure}
  \caption{MetaML-style multi-stage programming.}
  \vspace{-0.75\baselineskip}
\end{figure}
\Zcref{fig:metaml-example} illustrates the core constructs of MetaML-style MSP.
In \ctext{1}, quotation \code|<dots>| forms a code fragment, whereas splicing \code|~dots| inserts generated code into an enclosing quotation.
Quotations can also produce code that is open with respect to a future-stage context: in \ctext{2}, \code|<y>| refers to the parameter of the generated function.
Finally, generated code can be executed at runtime using \code|run|, as in \ctext{3}.
Together, these constructs support concise functional programs that manipulate code as data and are used in MetaOCaml~\cite{journals/scp/Kiselyov26}, Template Haskell~\cite{conf/haskell/SheardJ02}, and Scala 3~\cite{conf/gpce/StuckiBO21}.

Under the Curry--Howard correspondence, MetaML-style MSP admits standard modal accounts based on S4~\cite{journals/jacm/DaviesP01} and linear-time temporal logic (LTL)~\cite{journals/jacm/Davies17}. Both provide quasi-quotation (\ctext{1}), but characterize different capabilities: S4 accounts for closed code generation together with \code|run| (\ctext{3}), whereas LTL accounts for open future-stage code (\ctext{2}). Modal foundations intended to accommodate both capabilities have therefore been studied extensively~\cite{conf/esop/MoggiTBS99,conf/ppdp/YuseI06,journals/corr/abs-1010-3806}; we review this line of work in \Zcref{sec:related}. However, reconciling the S4 and LTL accounts does not by itself yield a modal characterization of the full range of MetaML-style staging.

The program in \Zcref{fig:metaml-example-csp} illustrates the missing ingredient. Although \code|inc| is bound at runtime (\ctext{4}), it occurs inside a quotation (\ctext{5}) and remains tied to that lexical binding when the generated code is executed by \code|run| (\ctext{6}). This is a form of \emph{cross-stage persistence} (CSP)~\cite{journals/tcs/TahaS00,conf/flops/HanadaI14,journals/pacmpl/XieWNY23}. The quoted code is thus simultaneously dependent on an enclosing runtime scope and executable by \code|run|. This combination cannot be characterized merely as the coexistence of the standard LTL capability of constructing open future-stage code and the standard S4 capability of executing closed code. What is needed is a single modal logic, equipped with a proof theory and Kripke semantics, that accounts for this interaction.

\paragraph*{Birelational Kripke Structures in Staged Programs}
We make this issue precise in terms of the scope and stage relationships
required by a staged program.  As illustrated in
\Zcref{fig:syntacticstructure}, we represent these requirements by a
birelational Kripke frame.
Its nodes are scopes, and its two accessibility relations are described below.
\begin{figure}[b]
  \vspace{-1.0em}
  \begin{subfigure}{0.35\textwidth}
    \raisebox{-\totalheight}{\includegraphics[page=1]{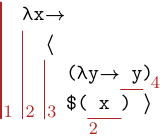}}\hfill
    \raisebox{-\totalheight}{\includegraphics[page=1]{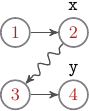}}
    \caption{Kripke structure with S4~\cite{journals/jacm/DaviesP01}}\label{fig:syntacticstructure:Kripke-style}
  \end{subfigure}\hfill
  \begin{subfigure}{0.6\textwidth}
    \raisebox{-\totalheight}{\includegraphics[page=2]{lambda-terms.pdf}}\hfill
    \raisebox{-\totalheight}{\includegraphics[page=2]{syntactic-structures.pdf}}
    \caption{Kripke structure with LTL~\cite{journals/jacm/Davies17}.}\label{fig:syntacticstructure:lamcirc}
  \end{subfigure}

  \centering
  \medskip\noindent
  \begin{subfigure}{0.65\textwidth}
    \centering
    \raisebox{-\totalheight/2}{\includegraphics[page=3]{lambda-terms.pdf}}\qquad
    \raisebox{-\totalheight/2}{\includegraphics[page=3]{syntactic-structures.pdf}}
    \caption{Kripke structure of \Zcref{fig:metaml-example-csp}}\label{fig:syntacticstructure:metaml}
  \end{subfigure}
  \captionsetup{skip=2pt}
  \caption{Staged programs and their Kripke structures.}
  \label{fig:syntacticstructure}
\end{figure}
\begin{itemize}
  \item A \emph{scope-nesting} edge \Scope permits an inner scope to use
  variables available in an outer scope. For example, a variable binding that
  introduces a scope gives rise to such an edge.
  \item A \emph{stage-transition} edge \Stage introduces a stage boundary,
  allowing splices in an inner scope to escape to an outer scope. A quotation
  gives rise to such an edge.
\end{itemize}

The effect of quotations differs in \Zcref{fig:syntacticstructure:Kripke-style,fig:syntacticstructure:lamcirc}, which reflects the difference between S4 and LTL.
In \Zcref{fig:syntacticstructure:Kripke-style}, the quote introduces a new scope 3 with only a stage transition from 2.
This behavior corresponds to the S4 interpretation of MetaML, where a quote always moves to the top-level scope, where no variable can be used.

On the other hand, in \Zcref{fig:syntacticstructure:lamcirc}, the inner quote
introduces scope 5, with a stage-transition edge from scope 2 and a
scope-nesting edge from scope 4. The latter accounts for the occurrence of
\code|y| in scope 5. More generally, LTL keeps scopes at different stages
strictly separated. A scope introduced by a quotation receives a
stage-transition edge from the source scope and a scope-nesting edge from the
innermost scope at the target stage. The latter allows the quotation to use
variables bound outside it at the same stage.

Such behavior reflects birelational Kripke semantics for S4 and LTL~\cite{simpson1994proof,journals/iandc/KojimaI11,alechina+2001categorical}, when we regard scope nesting as intuitionistic accessibility and stage transition as modal accessibility.

\Zcref{fig:syntacticstructure:metaml} illustrates the Kripke structure of the
staged program in \Zcref{fig:metaml-example-csp}. In this structure, the
quotation introduces the new scope 3 that is connected to scope 2 by both a
stage transition and scope nesting; the scope-nesting relation is necessary to
account for the occurrence of \code|inc|. This structure is not captured by
either S4 or LTL: S4 lacks the required scope-nesting edge, while LTL supplies
such an edge only from within the target stage, not from the source of the
stage transition itself.

This comparison identifies the gap in the existing modal accounts: the relevant
programs involve an interaction between scope nesting and stage transition that
is not represented by either S4 or LTL. We are thus led to the following
question: \emph{what kind of modal logic can account for such scope-and-stage
structures in MetaML-style staged programs?}

\paragraph*{Our Proposal: Bounded Modal Logic}
This paper answers this question by proposing \emph{Bounded Modal Logic}
(\bml), a constructive modal logic for reasoning about birelational Kripke
structures of staged programs. To capture the interaction between scope
nesting and stage transition, \bml makes scopes nameable in the object language
and uses such names as bounds on modalities. Its formula language is given by:
\[
  \bmttnt{A}, \bmttnt{B} \Coloneqq
  \bmttnt{p}
  \mid \bmttnt{A}  \mathbin{\rightarrow}  \bmttnt{B}
  \mid  \Box^{\mathord{\succeq}  \gamma }  \bmttnt{A} 
  \mid  \forall  \gamma_{{\mathrm{1}}} \within  \gamma_{{\mathrm{2}}} . \bmttnt{A} 
\]
Here, $\gamma$, $\gamma_{{\mathrm{1}}}$, and $\gamma_{{\mathrm{2}}}$ range over names for elements of
the underlying Kripke structure, regarded here as variable scopes. We call
such names \emph{classifiers}. The relation $\gamma_{{\mathrm{1}}} \, \preceq \, \gamma_{{\mathrm{2}}}$ denotes scope
nesting: the scope named by $\gamma_{{\mathrm{2}}}$ is nested inside the
scope named by $\gamma_{{\mathrm{1}}}$. We also have a special classifier $ \exclam $ that
represents the global scope.

The central connective is the \emph{bounded modality} $ \Box^{\mathord{\succeq}  \gamma }  \bmttnt{A} $.
The symbol $\succeq \gamma$ indicates a lower bound with respect to the
scope-nesting order. Thus, $ \Box^{\mathord{\succeq}  \gamma }  \bmttnt{A} $ refines an ordinary modal type by
requiring the target of the stage transition to lie above the scope named by
$\gamma$. Intuitively, it describes code of type $\bmttnt{A}$ whose generated
stage may still use the bindings available in that scope.
The following diagram illustrates this semantic pattern.\par
{\centering
  \vspace{0.5\baselineskip}
  \includegraphics{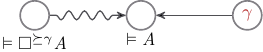}
\par}
\noindent{}Thus, the annotation $\gamma$ records the source scope on which the code may
depend; for example, \code|< inc 10 >| in
\Zcref{fig:syntacticstructure:metaml} has type $ \Box^{\mathord{\succeq}  \delta }  \mathtt{int} $, where $\delta$
names the scope of \code|inc|.

In addition, \emph{polymorphic classifier quantifier} $ \forall  \gamma_{{\mathrm{1}}} \within  \gamma_{{\mathrm{2}}} . \bmttnt{A} $ introduces a quantification over $\gamma_{{\mathrm{1}}}$ with lower bound $\gamma_{{\mathrm{2}}}$ with regard to scope nesting. Such a quantifier is useful to generalize functions with classifiers: $ \forall  \gamma_{{\mathrm{1}}} \within  \gamma_{{\mathrm{2}}} . \bmttsym{(}   \Box^{\mathord{\succeq}  \gamma_{{\mathrm{1}}} }  \bmttnt{A}   \mathbin{\rightarrow}   \Box^{\mathord{\succeq}  \gamma_{{\mathrm{1}}} }  \bmttnt{A}   \bmttsym{)} $ can be interpreted as a type for a function over code fragments with arbitrary scope $\gamma_{{\mathrm{1}}}$, which is nested within $\gamma_{{\mathrm{2}}}$.
Thus, \bml has an aspect of first-order predicate logic whose domain is a birelational Kripke structure. Note that atomic propositions in \bml are only propositional variables, and relations between classifiers do not form propositions on their own, unlike ordinary predicate logic. As we shall see later, those assertions will be treated as first-class judgments in our proof system.

\paragraph*{Contributions and Organization}
\Zcref{sec:structure, sec:logic:nd, sec:logic:semantics, sec:logic:completeness}
develop \bml-structures, a natural-deduction system, and Kripke semantics, and
establish soundness and completeness.  \Zcref{sec:calc, sec:staging} develop
the corresponding Curry--Howard calculus and staged reduction semantics,
together with their metatheory.  \Zcref{sec:s4andltl} establishes a logical
correspondence with $\CS4$ and type-preserving embeddings of the corresponding
S4- and LTL-based staging calculi.

\section{Syntactic Structure of Staged Programs}\label{sec:structure}
Before introducing \bml, we formally define a birelational Kripke structure, called a \bml-structure, that captures the structure of staged programs.
\begin{definition}[\bml-Structure]
  A\/ \bml-structure is a quintuple\/ $ \langle D ,  \preceq ,  \sqsubseteq ,  V ,   \exclam  \rangle $ where
  \begin{itemize}
  \item $ \langle D , \preceq \rangle $ is a preordered set with
  \item a root\/ $ \exclam $, the least element of $D$ with respect to $ \preceq $;
  \item $ \sqsubseteq $ is a preorder on $D$ with \emph{stability}\footnote{%
        The term `stable' is
        borrowed from \Textcite{stell+2016bi-intuitionistic}.%
      } condition $\mathord{ \mathrel{(\mathord{ \preceq }) } } \subseteq \mathord{ \mathrel{(\mathord{ \sqsubseteq }) } }$, or equivalently: left-stability $\mathord{ \mathrel{(\mathord{  \mathrel{\mathord{ \preceq } \mathbin{;} \mathord{ \sqsubseteq } }  }) } } \subseteq \mathord{ \mathrel{(\mathord{ \sqsubseteq }) } }$ and right-stability $\mathord{ \mathrel{(\mathord{  \mathrel{\mathord{ \sqsubseteq } \mathbin{;} \mathord{ \preceq } }  }) } } \subseteq \mathord{ \mathrel{(\mathord{ \sqsubseteq }) } }$;
    \item $V$ assigns each atom~$p$ to an upward-closed subset of~$D$.
  \end{itemize}
\end{definition}

We regard a \bml-structure as capturing the essential syntactic structure of staged programs, where the intuitionistic relation $ \preceq $ denotes scope nesting and the modal relation $ \sqsubseteq $ denotes stage transition. $D$ is a set of scopes or locations in programs. $ \exclam $ represents the global scope, which behaves as a bottom element with regard to $ \preceq $. We impose the stability condition for the modal relation to match the proof system introduced in~\Zcref{sec:logic:nd}.

Before defining the proof system of \bml, we discuss how \bml-structures arise in Kripke/Fitch-style proof systems~\cite{journals/jacm/DaviesP01,conf/fossacs/Clouston18,journals/pacmpl/ValliappanRC22,journals/pacmpl/GratzerSB19,conf/esop/MuraseNI23}, which provide a basis for S4-based modal calculi with quasi-quotation. This introduces several key notions that we will use in the next section. We recall a Kripke/Fitch-style proof system for a constructive variant of $\logicS4$ from the literature. We derive propositions via the judgment $ {\Gamma } \vdash_{\logicS4}  A $, where a context ${\Gamma }$ has the following structure.
\[ {\Gamma }, {\Delta } \Coloneqq  \varepsilon  \mid {\Gamma }  \bmttsym{,}  A \mid {\Gamma }  \bmttsym{,}  \mathord{\blacktriangleright} \]

\Textcite{journals/jacm/DaviesP01} mention that $ \mathord{\blacktriangleright} $ represents modal transition, and propositions between each $ \mathord{\blacktriangleright} $ represent assumptions holding at a specific world.\footnote{To be precise, \Textcite{journals/jacm/DaviesP01} introduced a stack structure over contexts to represent modal transitions between contexts. In our definition, we can regard $ \mathord{\blacktriangleright} $ to delimit contexts. With regard to Fitch-style proof systems~\cite{conf/fossacs/Clouston18,journals/pacmpl/ValliappanRC22,journals/pacmpl/GratzerSB19,conf/esop/MuraseNI23}, we can identify $ \mathord{\blacktriangleright} $ with $\faLockOpen$ in their definition of contexts.} We elaborate this idea to make a correspondence between a context and a \bml-structure. For example, a context $\bmttnt{p}  \bmttsym{,}  \mathord{\blacktriangleright}  \bmttsym{,}  \bmttnt{p}  \bmttsym{,}  \bmttnt{q}  \bmttsym{,}  \mathord{\blacktriangleright}$ corresponds to the \bml-structure below:
\begin{center}
  \includegraphics{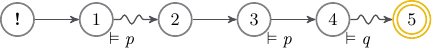}
\end{center}
$\Scope$ and $\Stage$ represent $ \preceq $ and $ \sqsubseteq $, respectively. Here, we can regard a context as carrying three kinds of information:
\begin{enumerate}
\item A \bml-structure itself;
\item Assumptions holding at each element of the \bml-structure; and
\item The current position in the \bml-structure where deduction is performed (drawn as a yellow double circle).
\end{enumerate}

In particular, we use the notion of the current position of a \bml-structure; we simply call it a \emph{position} in the rest of the paper. Then, each item in a context works in the following way:
\begin{itemize}
\item An empty context $ \varepsilon $ corresponds to a \bml-structure with a single element $ \exclam $. No assumption holds there, and the position is $ \exclam $.
\item When adding an assumption $A$ to ${\Gamma }$, it introduces a new element with an assumption $A$. The position moves to the new element, and it introduces $ \preceq $ along with the movement.
\item When adding $ \mathord{\blacktriangleright} $ to ${\Gamma }$, it introduces a new element, and the position moves to it. It also introduces $ \sqsubseteq $ along with the movement.
\end{itemize}

We can informally understand derivation rules with respect to \bml-structure of contexts, regarding them as Kripke models for intuitionistic modal logic~\cite{simpson1994proof,alechina+2001categorical}. We explain three selected rules from the viewpoint of \bml-structure:
\vspace{-0.75\baselineskip}
\begin{linenomath}
  \begin{mathpar}
    \begin{prooftree}
      \caption{Hyp}\label{rule:ks-derive-hyp}
      \hypo{\mathord{\blacktriangleright} \, \notin \, {\Gamma }_{{\mathrm{2}}}}
      \infer1{ {\Gamma }_{{\mathrm{1}}}  \bmttsym{,}  A  \bmttsym{,}  {\Gamma }_{{\mathrm{2}}} \vdash_{\logicS4}  A }
    \end{prooftree}
    \and
    \begin{prooftree}
      \caption{$ \Box $-I}\label{rule:ks-derive-nec-i}
      \hypo{ {\Gamma }  \bmttsym{,}  \mathord{\blacktriangleright} \vdash_{\logicS4}  A }
      \infer1{ {\Gamma } \vdash_{\logicS4}   \Box A  }
    \end{prooftree}
    \and
    \begin{prooftree}
      \caption{$ \Box $-E}\label{rule:ks-derive-nec-e}
      \hypo{ {\Gamma }_{{\mathrm{1}}} \vdash_{\logicS4}   \Box A  }
      \infer1{ {\Gamma }_{{\mathrm{1}}}  \bmttsym{,}  {\Gamma }_{{\mathrm{2}}} \vdash_{\logicS4}  A }
    \end{prooftree}
  \end{mathpar}
\end{linenomath}
\begin{description}
\item[\ref{rule:ks-derive-hyp}:] As ${\Gamma }_{{\mathrm{2}}}$ does not include $ \mathord{\blacktriangleright} $, the relations introduced by ${\Gamma }_{{\mathrm{2}}}$ are all $ \preceq $. As $ \preceq $ is reflexive and transitive, we use persistency to conclude that $A$ holds at the position of ${\Gamma }_{{\mathrm{1}}}  \bmttsym{,}  A  \bmttsym{,}  {\Gamma }_{{\mathrm{2}}}$.
\item[\ref{rule:ks-derive-nec-i}:] The premise $ {\Gamma }  \bmttsym{,}  \mathord{\blacktriangleright} \vdash_{\logicS4}  A $ states that $A$ holds at the element that is reachable by $ \sqsubseteq $ from the position of ${\Gamma }$. Hence, we can conclude $ \Box A $ at the position of ${\Gamma }$.
\item[\ref{rule:ks-derive-nec-e}:] As the relations introduced by ${\Gamma }_{{\mathrm{2}}}$ include both $ \preceq $ and $ \sqsubseteq $, we get a single transition with $ \sqsubseteq $ using the stability condition and reflexivity/transitivity of $ \sqsubseteq $. Hence, we can safely conclude $A$ at the position of ${\Gamma }_{{\mathrm{1}}}  \bmttsym{,}  {\Gamma }_{{\mathrm{2}}}$.
\end{description}

When proof terms are assigned to these rules, the structural information
carried by contexts is reflected in the syntax of modal lambda terms:
hypotheses correspond to variable bindings, while $ \mathord{\blacktriangleright} $ records the
stage transition induced by quotation. This makes it natural to relate modal
lambda terms to \bml-structures. In \Zcref{sec:logic:nd,sec:calc}, we develop
a richer context structure and a corresponding lambda calculus for \bml.

\section{Natural Deduction}\label{sec:logic:nd}
Having introduced the underlying \bml-structure as a conceptual basis, we now formalize reasoning over it by presenting a natural-deduction proof system.
This proof system is designed to internalize the intuitions observed in Kripke/Fitch-style proof systems.
First, we annotate each item of a context with a \emph{classifier}, following the tradition of \emph{labelled proof systems}~\cite{book/Gabbay96,negri2005proof,journals/entcs/ReedP09}. This allows us to refer to elements in a \bml-structure.
\[ \Gamma, \Delta \Coloneqq  \varepsilon  \mid \Gamma, \bmttnt{A}^{\gamma} \mid \Gamma, \mathord{\blacktriangleright} ^{\gamma} \]
We also use an \emph{initial} classifier $ \exclam $ to refer to an empty context.
We then extend the structure of contexts as follows (For clarity, we color classifiers if they are newly introduced):
\[ \Gamma, \Delta \Coloneqq  \varepsilon  \mid \Gamma  \bmttsym{,}   { \bmttnt{A} }^{\binder{ \gamma } }  \mid \Gamma  \bmttsym{,}   \mathord{\blacktriangleright} ^{\binder{ \gamma_{{\mathrm{1}}} } \within \gamma_{{\mathrm{2}}} }  \mid \Gamma  \bmttsym{,}   \mathord{\blacktriangleleft} ^{ \gamma }  \mid \Gamma  \bmttsym{,}   \binder{ \gamma_{{\mathrm{1}}} } \within \gamma_{{\mathrm{2}}}  \]
Empty contexts and hypotheses behave as before.
$ \mathord{\blacktriangleright} ^{\binder{ \gamma_{{\mathrm{2}}} } \within \gamma_{{\mathrm{1}}} } $ introduces an intuitionistic transition from $\gamma_{{\mathrm{1}}}$ to $\gamma_{{\mathrm{2}}}$, in addition to a modal transition.
$ \mathord{\blacktriangleleft} ^{ \gamma_{{\mathrm{1}}} } $ is a new item that only moves position to $\gamma_{{\mathrm{1}}}$, which requires that there is a modal transition in the opposite direction of the movement.
$\gamma_{{\mathrm{2}}} \,  \succeq  \, \gamma_{{\mathrm{1}}}$ is also a new item that introduces a new element $\gamma_{{\mathrm{2}}}$, and an intuitionistic transition from $\gamma_{{\mathrm{1}}}$ to $\gamma_{{\mathrm{2}}}$.
We summarize the behavior of each item in \Zcref{fig:contextsummary}. \Zcref{fig:intuition} provides concrete examples of contexts and their corresponding \bml-structures.
\begin{table}[bp]
  \centering
  \caption{Summary of the context syntax: what happens when $\Gamma$ is extended with each item. Colored classifiers indicate that they are introduced by the items.}%
  \label{fig:contextsummary}
  \begin{threeparttable}
    \begin{tabular}{c c c c c c}
      \toprule
      \multirowcell{2}[-\jot/2]{\textbf{Item}}
      &
        \multirowcell{2}[-\jot/2]{\textbf{Moves position?}\\\textbf{(to)}}
      &
        \multicolumn{3}{c}{\textbf{Assumptions added}}
      &
        \multirowcell{2}[-\jot/2]{\textbf{Consumed}\\\textbf{by}}
      \\
      \cmidrule(lr){3-5}
      &
      &
        \textbf{($ \preceq $)}
      &
        \textbf{($ \sqsubseteq $)}
      &
        \textbf{Proposition} & \\
      \midrule
      $ { \bmttnt{A} }^{\binder{ \gamma } } $
      & Yes\;($\gamma$)
      & $ \mathrm{pos} ( \Gamma )  \, \preceq \, \gamma$
      & \texthorizbar
        & $\bmttnt{A}$ holds at $\gamma$
          & (\ref{rule:derive-to-i})
      \\
      $ \mathord{\blacktriangleright} ^{\binder{ \gamma_{{\mathrm{2}}} } \within \gamma_{{\mathrm{1}}} } $
      & Yes\;($\gamma_{{\mathrm{2}}}$)
      & $\gamma_{{\mathrm{1}}} \, \preceq \, \gamma_{{\mathrm{2}}}$
      & $ \mathrm{pos} ( \Gamma )  \, \sqsubseteq \, \gamma_{{\mathrm{2}}}$
        & \texthorizbar
          & (\ref{rule:derive-bm-i})
      \\
      $ \mathord{\blacktriangleleft} ^{ \gamma } $
      & Yes\;($\gamma$)
      & \texthorizbar
      & \texthorizbar\rlap{\tnote{1}}
        & \texthorizbar
          & (\ref{rule:derive-bm-e})
      \\
      $\gamma_{{\mathrm{2}}} \,  \succeq  \, \gamma_{{\mathrm{1}}}$
      & No
      & $\gamma_{{\mathrm{1}}} \, \preceq \, \gamma_{{\mathrm{2}}}$
      & \texthorizbar
        & \texthorizbar
          & (\ref{rule:derive-polycls-i})
      \\
      \bottomrule
    \end{tabular}
    \begin{tablenotes}
      \item[1] $\gamma \, \sqsubseteq \,  \mathrm{pos} ( \Gamma ) $ is not added as an assumption, but is required to be deduced from $\Gamma$.
    \end{tablenotes}
  \end{threeparttable}
\end{table}
\begin{figure}[bpt]
  \centering
  \NewDocumentCommand\subfig{O{} m}
    {\subcaptionbox{#2}[#1\textwidth]
      {\includegraphics[page=\inteval{\value{subfigure}+1}]
        {context-structures.pdf}}}%
  \hfil
  \subfig[0.2]{$ \varepsilon $.\label{fig:intuition:a}}\hfil
  \subfig[0.3]{$ { \bmttnt{p} }^{\binder{ \gamma_{{\mathrm{1}}} } }   \bmttsym{,}   { \bmttnt{q} }^{\binder{ \gamma_{{\mathrm{2}}} } } $.\label{fig:intuition:b}}\hfil
  \par\bigskip
  \hfil
  \subfig[0.3]{$ { \bmttnt{p} }^{\binder{ \gamma_{{\mathrm{1}}} } }   \bmttsym{,}   { \bmttnt{q} }^{\binder{ \gamma_{{\mathrm{2}}} } }   \bmttsym{,}   \mathord{\blacktriangleright} ^{\binder{ \gamma_{{\mathrm{3}}} } \within  \mathord{\boldsymbol{!} }  } $.\label{fig:intuition:c}}\hfil
  \subfig[0.3]{$ { \bmttnt{p} }^{\binder{ \gamma_{{\mathrm{1}}} } }   \bmttsym{,}   { \bmttnt{q} }^{\binder{ \gamma_{{\mathrm{2}}} } }   \bmttsym{,}   \mathord{\blacktriangleright} ^{\binder{ \gamma_{{\mathrm{3}}} } \within  \mathord{\boldsymbol{!} }  }   \bmttsym{,}   \mathord{\blacktriangleleft} ^{ \gamma_{{\mathrm{1}}} } $.\label{fig:intuition:d}}\hfil
  \subfig[0.36]{$ { \bmttnt{p} }^{\binder{ \gamma_{{\mathrm{1}}} } }   \bmttsym{,}   { \bmttnt{q} }^{\binder{ \gamma_{{\mathrm{2}}} } }   \bmttsym{,}   \mathord{\blacktriangleright} ^{\binder{ \gamma_{{\mathrm{3}}} } \within  \mathord{\boldsymbol{!} }  }   \bmttsym{,}   \mathord{\blacktriangleleft} ^{ \gamma_{{\mathrm{1}}} }   \bmttsym{,}   \binder{ \gamma_{{\mathrm{4}}} } \within \gamma_{{\mathrm{3}}} $.\label{fig:intuition:e}}\hfil
  \caption{\bml contexts and their corresponding \bml-structures. \Scope and \Stage indicate intuitionistic and modal transitions, respectively. Each yellow double circle indicates the current position of each context.}\label{fig:intuition}
  \vspace{-0.75\baselineskip}
\end{figure}

We proceed with formal definitions. First we define $ \mathrm{pos} ( \Gamma ) $, the position of $\Gamma$.
\begin{definition}[Position of Contexts]\samepage
  \begin{align*}
     \mathrm{pos} (  \varepsilon  )  &=  \exclam  &   \mathrm{pos} ( \Gamma  \bmttsym{,}   { \bmttnt{A} }^{\binder{ \gamma } }  )  &= \gamma &  \mathrm{pos} ( \Gamma  \bmttsym{,}   \mathord{\blacktriangleright} ^{\binder{ \gamma_{{\mathrm{1}}} } \within \gamma_{{\mathrm{2}}} }  )  &= \gamma_{{\mathrm{1}}}\\
     \mathrm{pos} ( \Gamma  \bmttsym{,}   \mathord{\blacktriangleleft} ^{ \gamma }  )  &= \gamma &  \mathrm{pos} ( \Gamma  \bmttsym{,}   \binder{ \gamma_{{\mathrm{1}}} } \within \gamma_{{\mathrm{2}}}  )  &=  \mathrm{pos} ( \Gamma ) 
  \end{align*}
\end{definition}
For the sake of space, we introduce a shorthand notation for positions.
\begin{notation*}
 We write\/ $ \Gamma ^{\position{ \gamma } } $ to represent\/ $\Gamma$ with its position $\gamma$. When we write a context with multiple meta-variables like\/ $ \Gamma_{{\mathrm{1}}} ^{\position{ \gamma_{{\mathrm{1}}} } }   \bmttsym{,}   \Gamma_{{\mathrm{2}}} ^{\position{ \gamma_{{\mathrm{2}}} } } $, then it means that\/ $ \mathrm{pos} ( \Gamma_{{\mathrm{1}}} )  = \gamma_{{\mathrm{1}}}$ and\/ $ \mathrm{pos} ( \Gamma_{{\mathrm{1}}}  \bmttsym{,}  \Gamma_{{\mathrm{2}}} )   \bmttsym{=}  \gamma_{{\mathrm{2}}}$ hold. Note that it does not mean\/ $ \mathrm{pos} ( \Gamma_{{\mathrm{2}}} )   \bmttsym{=}  \gamma_{{\mathrm{2}}}$ because\/ $\Gamma_{{\mathrm{2}}}$ can be empty.
\end{notation*}

We write $ \mathbf{Dom}_{\token{C} }( \Gamma ) $ for a set of classifiers defined in $\Gamma$, corresponding to elements in a \bml-structure. Then, $\Gamma  \vdash  \gamma_{{\mathrm{1}}}  \preceq  \gamma_{{\mathrm{2}}}$ and $ \Gamma \vdash \gamma_{{\mathrm{1}}} \sqsubseteq \gamma_{{\mathrm{2}}} $ describe scope nesting and stage transition between classifiers, whose derivation rules can be found in \Zcref{fig:derive-trans}. \ref{rule:itrans-refl} and \ref{rule:gtrans-trans} describe the reflexive and transitive nature of these relations. \ref{rule:mtrans-lift} states that $ \preceq $ can lift to $ \sqsubseteq $, which corresponds to the stability condition in \bml-structure. The rest of the rules introduce $ \preceq $ and $ \sqsubseteq $ based on each item of a context as described in \Zcref{fig:contextsummary}.
\begin{figure}[bt]
  \begin{inparaenum}
    \item
      $\mathord{ \trianglelefteq } \in \{\mathord{ \preceq }, \mathord{ \sqsubseteq }\}$.
    \item
      $\vdash  \Gamma  \hasType \, \bmttkw{ctx}$ is assumed.
  \end{inparaenum}
  \begin{mathpar}
    \begin{prooftree}
      \caption{$ \preceq $-Refl}\label{rule:itrans-refl}
      \hypo{\gamma \, \in \,  \mathbf{Dom}_{\token{C} }( \Gamma ) }
      \infer1{\Gamma  \vdash  \gamma  \preceq  \gamma}
    \end{prooftree}
    \and
    \begin{prooftree}
      \caption{$ \trianglelefteq $-Trans}\label{rule:gtrans-trans}
      \hypo{\Gamma  \vdash  \gamma_{{\mathrm{1}}}  \trianglelefteq  \gamma_{{\mathrm{2}}}}
      \hypo{\Gamma  \vdash  \gamma_{{\mathrm{2}}}  \trianglelefteq  \gamma_{{\mathrm{3}}}}
      \infer2{\Gamma  \vdash  \gamma_{{\mathrm{1}}}  \trianglelefteq  \gamma_{{\mathrm{3}}}}
    \end{prooftree}
    \and
    \begin{prooftree}
      \caption{$ \sqsubseteq $-Lift}\label{rule:mtrans-lift}
      \hypo{\Gamma  \vdash  \gamma_{{\mathrm{1}}}  \preceq  \gamma_{{\mathrm{2}}}}
      \infer1{ \Gamma \vdash \gamma_{{\mathrm{1}}} \sqsubseteq \gamma_{{\mathrm{2}}} }
    \end{prooftree}
    \\
    \begin{prooftree}
      \caption{$ \preceq $-Hyp}\label{rule:itrans-hyp}
      \hypo{}
      \infer1{ \Gamma_{{\mathrm{1}}} ^{\position{ \gamma_{{\mathrm{1}}} } }   \bmttsym{,}   { \bmttnt{A} }^{\binder{ \gamma_{{\mathrm{2}}} } }   \bmttsym{,}  \Gamma_{{\mathrm{2}}}  \vdash  \gamma_{{\mathrm{1}}}  \preceq  \gamma_{{\mathrm{2}}}}
    \end{prooftree}
    \and
    \begin{prooftree}
      \caption{$ \preceq $-Cls}\label{rule:itrans-cls}
      \hypo{ \binder{ \gamma_{{\mathrm{1}}} } \within \gamma_{{\mathrm{2}}}  \, \in \, \Gamma}
      \infer1{\Gamma  \vdash  \gamma_{{\mathrm{2}}}  \preceq  \gamma_{{\mathrm{1}}}}
    \end{prooftree}
    \and
    \begin{prooftree}
      \caption{$ \sqsubseteq $-$ \mathord{\blacktriangleright} $}\label{rule:mtrans-open}
      \hypo{}
      \infer1{  \Gamma_{{\mathrm{1}}} ^{\position{ \gamma_{{\mathrm{1}}} } }   \bmttsym{,}   \mathord{\blacktriangleright} ^{\binder{ \gamma_{{\mathrm{2}}} } \within \gamma_{{\mathrm{3}}} }   \bmttsym{,}  \Gamma_{{\mathrm{2}}} \vdash \gamma_{{\mathrm{1}}} \sqsubseteq \gamma_{{\mathrm{2}}} }
    \end{prooftree}
    \and
    \begin{prooftree}
      \caption{$ \preceq $-$ \mathord{\blacktriangleright} $}\label{rule:itrans-open}
      \hypo{ \mathord{\blacktriangleright} ^{\binder{ \gamma_{{\mathrm{1}}} } \within \gamma_{{\mathrm{2}}} }  \, \in \, \Gamma}
      \infer1{\Gamma  \vdash  \gamma_{{\mathrm{2}}}  \preceq  \gamma_{{\mathrm{1}}}}
    \end{prooftree}
  \end{mathpar}
  \caption{Derivation rules for $\Gamma  \vdash  \gamma_{{\mathrm{1}}}  \preceq  \gamma_{{\mathrm{2}}}$ and $ \Gamma \vdash \gamma_{{\mathrm{1}}} \sqsubseteq \gamma_{{\mathrm{2}}} $.}\label{fig:derive-trans}
\end{figure}

$\vdash  \Gamma  \hasType \, \bmttkw{ctx}$ and $\Gamma  \vdash  \bmttnt{A}  \hasType \, \bmttkw{prop}$ state well-formedness of $\Gamma$ and $\bmttnt{A}$, respectively. Most of the derivation rules for these judgments ensure occurrences of classifiers in $\Gamma$ and $\bmttnt{A}$ are well defined, and we omit rules. One exception is \ref{rule:wf-ctx-shut} in \Zcref{fig:derivationrules}, which ensures that the context $ \Gamma ^{\position{ \gamma } }   \bmttsym{,}   \mathord{\blacktriangleleft} ^{ \delta } $ satisfies $  \Gamma ^{\position{ \gamma } }  \vdash \delta \sqsubseteq \gamma $.

The judgment $\Gamma  \vdash  \bmttnt{A}$ asserts truth of $\bmttnt{A}$ under the context $\Gamma$. \Zcref{fig:derivationrules} lists derivation rules.
\begin{figure}[bt]
  \noindent
  \begin{mathpar}[\footnotesize]
    \begin{prooftree}\caption{WF-$ \mathord{\blacktriangleleft} $}\label{rule:wf-ctx-shut}
      \hypo{\vdash   \Gamma ^{\position{ \gamma } }   \hasType \, \bmttkw{ctx}}
      \hypo{  \Gamma ^{\position{ \gamma } }  \vdash \delta \sqsubseteq \gamma }
      \infer2{\vdash   \Gamma ^{\position{ \gamma } }   \bmttsym{,}   \mathord{\blacktriangleleft} ^{ \delta }   \hasType \, \bmttkw{ctx}}
    \end{prooftree}
    \and
    \begin{prooftree}
      \caption{Hyp}\label{rule:derive-hyp}
      \hypo{ { \bmttnt{A} }^{\binder{ \gamma_{{\mathrm{1}}} } }  \, \in \,  \Gamma ^{\position{ \gamma_{{\mathrm{2}}} } } }
      \hypo{ \Gamma ^{\position{ \gamma_{{\mathrm{2}}} } }   \vdash  \gamma_{{\mathrm{1}}}  \preceq  \gamma_{{\mathrm{2}}}}
      \infer2{ \Gamma ^{\position{ \gamma_{{\mathrm{2}}} } }   \vdash  \bmttnt{A}}
    \end{prooftree}
    \and
    \begin{prooftree}
      \caption{$\to$-I}\label{rule:derive-to-i}
      \hypo{\Gamma  \bmttsym{,}   { \bmttnt{A_{{\mathrm{1}}}} }^{\binder{ \gamma } }   \vdash  \bmttnt{A_{{\mathrm{2}}}}}
      \hypo{\gamma \, \notin \,  \token{FC}( \bmttnt{A_{{\mathrm{2}}}} ) }
      \infer2{\Gamma  \vdash  \bmttnt{A_{{\mathrm{1}}}}  \mathbin{\rightarrow}  \bmttnt{A_{{\mathrm{2}}}}}
    \end{prooftree}
    \and
    \begin{prooftree}
      \caption{$\to$-E}\label{rule:derive-to-e}
      \hypo{\Gamma  \vdash  \bmttnt{A_{{\mathrm{1}}}}  \mathbin{\rightarrow}  \bmttnt{A_{{\mathrm{2}}}}}
      \hypo{\Gamma  \vdash  \bmttnt{A_{{\mathrm{1}}}}}
      \infer2{\Gamma  \vdash  \bmttnt{A_{{\mathrm{2}}}}}
    \end{prooftree}
    \and
    \begin{prooftree}
      \caption{$\Box$-I}\label{rule:derive-bm-i}
      \hypo{\Gamma  \bmttsym{,}   \mathord{\blacktriangleright} ^{\binder{ \gamma_{{\mathrm{1}}} } \within \gamma_{{\mathrm{2}}} }   \vdash  \bmttnt{A}}
      \hypo{\gamma_{{\mathrm{1}}} \, \notin \,  \token{FC}( \bmttnt{A} ) }
      \infer2{\Gamma  \vdash   \Box^{\mathord{\succeq}  \gamma_{{\mathrm{2}}} }  \bmttnt{A} }
    \end{prooftree}
    \and
    \begin{prooftree}
      \caption{$\Box$-E}\label{rule:derive-bm-e}
      \hypo{ \Gamma ^{\position{ \gamma_{{\mathrm{1}}} } }   \bmttsym{,}   \mathord{\blacktriangleleft} ^{ \gamma_{{\mathrm{2}}} }   \vdash   \Box^{\mathord{\succeq}  \gamma_{{\mathrm{3}}} }  \bmttnt{A} }
      \hypo{ \Gamma ^{\position{ \gamma_{{\mathrm{1}}} } }   \vdash  \gamma_{{\mathrm{3}}}  \preceq  \gamma_{{\mathrm{1}}}}
      \infer2{ \Gamma ^{\position{ \gamma_{{\mathrm{1}}} } }   \vdash  \bmttnt{A}}
    \end{prooftree}
    \and
    \begin{prooftree}
      \caption{$\forall$-I}\label{rule:derive-polycls-i}
      \hypo{\Gamma  \bmttsym{,}   \binder{ \gamma_{{\mathrm{1}}} } \within \gamma_{{\mathrm{2}}}   \vdash  \bmttnt{A}}
      \infer1{\Gamma  \vdash   \forall  \gamma_{{\mathrm{1}}} \within  \gamma_{{\mathrm{2}}} . \bmttnt{A} }
    \end{prooftree}
    \and
    \begin{prooftree}
      \caption{$\forall$-E}\label{rule:derive-polycls-e}
      \hypo{\Gamma  \vdash   \forall  \gamma_{{\mathrm{1}}} \within  \gamma_{{\mathrm{2}}} . \bmttnt{A} }
      \hypo{\Gamma  \vdash  \gamma_{{\mathrm{2}}}  \preceq  \gamma_{{\mathrm{3}}}}
      \infer2{\Gamma  \vdash   \bmttnt{A} [ \binder{ \gamma_{{\mathrm{1}}} } \coloneqq \gamma_{{\mathrm{3}}}  ] }
    \end{prooftree}
  \end{mathpar}
  \caption{Curated rules for $\vdash  \Gamma  \hasType \, \bmttkw{ctx}$ and $\Gamma  \vdash  \bmttnt{A}$.}\label{fig:derivationrules}
\end{figure}
As discussed in Kripke/Fitch-style proof systems, we can understand these derivation rules via \bml-structure.
\ref{rule:derive-hyp} explicitly requires an intuitionistic transition $\gamma_{{\mathrm{1}}} \, \preceq \, \gamma_{{\mathrm{2}}}$ to use $\bmttnt{A}$ at $\gamma_{{\mathrm{1}}}$ via persistency.
Introduction rules \ref{rule:derive-to-i}, \ref{rule:derive-bm-i} and \ref{rule:derive-polycls-i} state that implication, bounded modality and polymorphic classifier quantifiers correspond to the structure of $ { \bmttnt{A} }^{\binder{ \gamma } } $, $ \mathord{\blacktriangleright} ^{\binder{ \gamma_{{\mathrm{1}}} } \within \gamma_{{\mathrm{2}}} } $ and $\gamma_{{\mathrm{1}}} \,  \succeq  \, \gamma_{{\mathrm{2}}}$, respectively. \ref{rule:derive-bm-e} uses the structure of $ \mathord{\blacktriangleleft} $ to get a modal relation from $\gamma_{{\mathrm{2}}}$ to $\gamma_{{\mathrm{1}}}$. It also requires a condition $\gamma_{{\mathrm{3}}} \, \preceq \, \gamma_{{\mathrm{1}}}$, which is required by the bound. We provide \Zcref{fig:derivationexamples} as examples of derivations.
\begin{figure}[bt]
  \noindent
  \footnotesize
  \[ \Gamma_{{\mathrm{2}}} =  \binder{ \gamma_{{\mathrm{1}}} } \within  \mathord{\boldsymbol{!} }    \bmttsym{,}   { \bmttsym{(}   \forall  \gamma_{{\mathrm{2}}} \within  \gamma_{{\mathrm{1}}} .  \Box^{\mathord{\succeq}  \gamma_{{\mathrm{2}}} }  \bmttnt{A}   \mathbin{\rightarrow}   \Box^{\mathord{\succeq}  \gamma_{{\mathrm{2}}} }  \bmttnt{B}    \bmttsym{)} }^{\binder{ \gamma_{{\mathrm{3}}} } } \qquad \Gamma_{{\mathrm{3}}} =  \mathord{\blacktriangleright} ^{\binder{ \gamma_{{\mathrm{4}}} } \within \gamma_{{\mathrm{1}}} }   \bmttsym{,}   { \bmttnt{A} }^{\binder{ \gamma_{{\mathrm{5}}} } }   \bmttsym{,}   \mathord{\blacktriangleleft} ^{ \gamma_{{\mathrm{3}}} }  \]
  \begin{mathpar}
    \begin{prooftree}
      \infer0[\ref{rule:derive-hyp}]{\Gamma_{{\mathrm{2}}}  \bmttsym{,}  \Gamma_{{\mathrm{3}}}  \vdash   \forall  \gamma_{{\mathrm{2}}} \within  \gamma_{{\mathrm{1}}} . \bmttsym{(}   \Box^{\mathord{\succeq}  \gamma_{{\mathrm{2}}} }  \bmttnt{A}   \mathbin{\rightarrow}   \Box^{\mathord{\succeq}  \gamma_{{\mathrm{2}}} }  \bmttnt{B}   \bmttsym{)} }
      \infer1[\ref{rule:derive-polycls-e}]{\Gamma_{{\mathrm{2}}}  \bmttsym{,}  \Gamma_{{\mathrm{3}}}  \vdash   \Box^{\mathord{\succeq}  \gamma_{{\mathrm{5}}} }  \bmttnt{A}   \mathbin{\rightarrow}   \Box^{\mathord{\succeq}  \gamma_{{\mathrm{5}}} }  \bmttnt{B} }
      \infer0[\ref{rule:derive-hyp}]{\Gamma_{{\mathrm{2}}}  \bmttsym{,}  \Gamma_{{\mathrm{3}}}  \bmttsym{,}   \mathord{\blacktriangleright} ^{\binder{ \gamma_{{\mathrm{6}}} } \within \gamma_{{\mathrm{5}}} }   \vdash  \bmttnt{A}}
      \infer1[\ref{rule:derive-bm-i}]{\Gamma_{{\mathrm{2}}}  \bmttsym{,}  \Gamma_{{\mathrm{3}}}  \vdash   \Box^{\mathord{\succeq}  \gamma_{{\mathrm{5}}} }  \bmttnt{A} }
      \infer2[\ref{rule:derive-to-e}]{\Gamma_{{\mathrm{2}}}  \bmttsym{,}  \Gamma_{{\mathrm{3}}}  \vdash   \Box^{\mathord{\succeq}  \gamma_{{\mathrm{5}}} }  \bmttnt{B} }
      \infer1[\ref{rule:derive-bm-e}]{\Gamma_{{\mathrm{2}}}  \bmttsym{,}   \mathord{\blacktriangleright} ^{\binder{ \gamma_{{\mathrm{4}}} } \within \gamma_{{\mathrm{1}}} }   \bmttsym{,}   { \bmttnt{A} }^{\binder{ \gamma_{{\mathrm{5}}} } }   \vdash  \bmttnt{B}}
      \infer1[\ref{rule:derive-to-i}]{\Gamma_{{\mathrm{2}}}  \bmttsym{,}   \mathord{\blacktriangleright} ^{\binder{ \gamma_{{\mathrm{4}}} } \within \gamma_{{\mathrm{1}}} }   \vdash  \bmttnt{A}  \mathbin{\rightarrow}  \bmttnt{B}}
      \infer1[\ref{rule:derive-bm-i}]{\Gamma_{{\mathrm{2}}}  \vdash   \Box^{\mathord{\succeq}  \gamma_{{\mathrm{1}}} }  \bmttsym{(}  \bmttnt{A}  \mathbin{\rightarrow}  \bmttnt{B}  \bmttsym{)} }
      \infer1[\ref{rule:derive-to-i}]{ \binder{ \gamma_{{\mathrm{1}}} } \within  \mathord{\boldsymbol{!} }    \vdash  \bmttsym{(}   \forall  \gamma_{{\mathrm{2}}} \within  \gamma_{{\mathrm{1}}} .  \Box^{\mathord{\succeq}  \gamma_{{\mathrm{2}}} }  \bmttnt{A}   \mathbin{\rightarrow}   \Box^{\mathord{\succeq}  \gamma_{{\mathrm{2}}} }  \bmttnt{B}    \bmttsym{)}  \mathbin{\rightarrow}   \Box^{\mathord{\succeq}  \gamma_{{\mathrm{1}}} }  \bmttsym{(}  \bmttnt{A}  \mathbin{\rightarrow}  \bmttnt{B}  \bmttsym{)} }
      \infer1[\ref{rule:derive-polycls-i}]{ \varepsilon   \vdash   \forall  \gamma_{{\mathrm{1}}} \within   \mathord{\boldsymbol{!} }  . \bmttsym{(}   \forall  \gamma_{{\mathrm{2}}} \within  \gamma_{{\mathrm{1}}} .  \Box^{\mathord{\succeq}  \gamma_{{\mathrm{2}}} }  \bmttnt{A}   \mathbin{\rightarrow}   \Box^{\mathord{\succeq}  \gamma_{{\mathrm{2}}} }  \bmttnt{B}    \bmttsym{)}  \mathbin{\rightarrow}   \Box^{\mathord{\succeq}  \gamma_{{\mathrm{1}}} }  \bmttsym{(}  \bmttnt{A}  \mathbin{\rightarrow}  \bmttnt{B}  \bmttsym{)}  }
    \end{prooftree}
    \and
    \begin{prooftree}
      \infer0[\ref{rule:derive-hyp}]{ { \bmttsym{(}   \Box^{\mathord{\succeq}   \mathord{\boldsymbol{!} }  }  \bmttnt{A}   \bmttsym{)} }^{\binder{ \gamma_{{\mathrm{1}}} } }   \vdash   \Box^{\mathord{\succeq}   \mathord{\boldsymbol{!} }  }  \bmttnt{A} }
      \infer1[\ref{rule:derive-bm-e}]{ { \bmttsym{(}   \Box^{\mathord{\succeq}   \mathord{\boldsymbol{!} }  }  \bmttnt{A}   \bmttsym{)} }^{\binder{ \gamma_{{\mathrm{1}}} } }   \vdash  \bmttnt{A}}
      \infer1[\ref{rule:derive-to-i}]{ \varepsilon   \vdash   \Box^{\mathord{\succeq}   \mathord{\boldsymbol{!} }  }  \bmttnt{A}   \mathbin{\rightarrow}  \bmttnt{A}}
    \end{prooftree}
  \end{mathpar}
  \caption{Examples of Derivations (omitting derivations w.r.t.\@ $ \preceq $).}
  \label{fig:derivationexamples}
\end{figure}

Finally, we establish that propositions remain true along intuitionistic
transitions.
\begin{lemma*}[Persistency]\label{claim:monotonicity}
  If\/ $ \Gamma_{{\mathrm{1}}} ^{\position{ \gamma_{{\mathrm{1}}} } }   \vdash  \bmttnt{A}$ and\/ $ \Gamma_{{\mathrm{1}}} ^{\position{ \gamma_{{\mathrm{1}}} } }   \bmttsym{,}   \Gamma_{{\mathrm{2}}} ^{\position{ \gamma_{{\mathrm{2}}} } }   \vdash  \gamma_{{\mathrm{1}}}  \preceq  \gamma_{{\mathrm{2}}}$, then\/ $ \Gamma_{{\mathrm{1}}} ^{\position{ \gamma_{{\mathrm{1}}} } }   \bmttsym{,}   \Gamma_{{\mathrm{2}}} ^{\position{ \gamma_{{\mathrm{2}}} } }   \vdash  \bmttnt{A}$.
\end{lemma*}

\section{Kripke Semantics}\label{sec:logic:semantics}
To define Kripke semantics for \bml, we follow the standard construction for intuitionistic first-order logic~\cite{simpson1994proof}, regarding \bml-structure as its domain. Thus, we define a \bml-model as a family of \bml-structures that captures growing domain structure along with $ \preccurlyeq $.
\begin{definition}[\bml-Model]
  A\/ \bml-model is a triple\/ $ \langle W ,   \preccurlyeq  ,   \lbrace  M _{ \bmttnt{w} }  \rbrace_{ \bmttnt{w} \in W }  \rangle $ where
  \begin{itemize}
  \item $ \langle W ,  \preccurlyeq  \rangle $ is a nonempty preordered set;
  \item Each $ M _{ \bmttnt{w} } $ is a \bml-structure\/ $ \langle  D _{ \bmttnt{w} }  ,   \preceq _{ \bmttnt{w} }  ,   \sqsubseteq _{ \bmttnt{w} }  ,   V _{ \bmttnt{w} }  ,    \exclam  _{ \bmttnt{w} }  \rangle $; and
  \item If $\bmttnt{w_{{\mathrm{1}}}} \,  \preccurlyeq  \, \bmttnt{w_{{\mathrm{2}}}}$, then
    \begin{itemize}
    \item $ D _{ \bmttnt{w_{{\mathrm{1}}}} }  \subseteq  D _{ \bmttnt{w_{{\mathrm{2}}}} } $;
    \item $\bmttnt{d_{{\mathrm{1}}}} \,  \preceq _{ \bmttnt{w_{{\mathrm{1}}}} }  \, \bmttnt{d_{{\mathrm{2}}}} \implies \bmttnt{d_{{\mathrm{1}}}} \,  \preceq _{ \bmttnt{w_{{\mathrm{2}}}} }  \, \bmttnt{d_{{\mathrm{2}}}}$;
    \item $\bmttnt{d_{{\mathrm{1}}}} \,  \sqsubseteq _{ \bmttnt{w_{{\mathrm{1}}}} }  \, \bmttnt{d_{{\mathrm{2}}}} \implies \bmttnt{d_{{\mathrm{1}}}} \,  \sqsubseteq _{ \bmttnt{w_{{\mathrm{2}}}} }  \, \bmttnt{d_{{\mathrm{2}}}}$;
    \item $ V _{ \bmttnt{w_{{\mathrm{1}}}} }   \bmttsym{(}  \bmttnt{p}  \bmttsym{)} \subseteq  V _{ \bmttnt{w_{{\mathrm{2}}}} }   \bmttsym{(}  \bmttnt{p}  \bmttsym{)}$; and
    \item $  \exclam  _{ \bmttnt{w_{{\mathrm{1}}}} }  =   \exclam  _{ \bmttnt{w_{{\mathrm{2}}}} } $.
    \end{itemize}
  \end{itemize}
\end{definition}

Given a \bml-model $\mathfrak{M} =  \langle W ,   \preccurlyeq  ,   \lbrace  M _{ \bmttnt{w} }  \rbrace_{ \bmttnt{w} \in W }  \rangle $ and $\bmttnt{w} \, \in \, W$, a \emph{$\bmttnt{w}$-assignment~$\rho$} is a partial map from the set of all classifiers to $ D _{ \bmttnt{w} } $ with $ \exclam  \mapsto   \exclam  _{ \bmttnt{w} } $. For simplicity, we assume that an assignment has sufficient domain of definition for interpretation.
Given a \bml-model $\mathfrak{M} =  \langle W ,   \preccurlyeq  ,   \lbrace  M _{ \bmttnt{w} }  \rbrace_{ \bmttnt{w} \in W }  \rangle $, the \emph{satisfaction} of a formula $\bmttnt{A}$ at $\bmttnt{d} \, \in \,  D _{ \bmttnt{w} } $ with $\rho$ on $\bmttnt{w} \, \in \, W$, written $ \mathfrak{M} ,  \bmttnt{w} ,  \bmttnt{d}    \Vdash  ^{ \rho }  \bmttnt{A} $, is defined as
\begin{alignat*}{2}
  &  \mathfrak{M} ,  \bmttnt{w} ,  \bmttnt{d}    \Vdash  ^{ \rho }  \bmttnt{A}  && \iff \forall \bmttnt{v} \,  \succcurlyeq  \, \bmttnt{w}. \parens[\big]{ \mathfrak{M} ,  \bmttnt{v} ,  \bmttnt{d}    \vDash  ^{ \rho }  \bmttnt{A} } \text{,} \\
  \intertext{where $ \mathfrak{M} ,  \bmttnt{w} ,  \bmttnt{d}    \vDash  ^{ \rho }  \bmttnt{A} $ is defined as follows:}
  &  \mathfrak{M} ,  \bmttnt{w} ,  \bmttnt{d}    \vDash  ^{ \rho }  \bmttnt{p}  && \iff \bmttnt{d} \, \in \,  V _{ \bmttnt{w} }   \bmttsym{(}  \bmttnt{p}  \bmttsym{)} \text{;}\\
  &  \mathfrak{M} ,  \bmttnt{w} ,  \bmttnt{d}    \vDash  ^{ \rho }  \bmttnt{A}  \mathbin{\rightarrow}  \bmttnt{B}  && \iff \forall \bmttnt{e} \,   \succeq  _{ \bmttnt{w} }  \, \bmttnt{d}. \parens[\big]{ \mathfrak{M} ,  \bmttnt{w} ,  \bmttnt{e}    \Vdash  ^{ \rho }  \bmttnt{A}  \implies  \mathfrak{M} ,  \bmttnt{w} ,  \bmttnt{e}    \Vdash  ^{ \rho }  \bmttnt{B} } \text{;} \\
  &  \mathfrak{M} ,  \bmttnt{w} ,  \bmttnt{d}    \vDash  ^{ \rho }   \Box^{\mathord{\succeq}  \gamma }  \bmttnt{A}   && \iff \forall \bmttnt{e} \,  \sqsupseteq _{ \bmttnt{w} }  \, \bmttnt{d}. \parens[\big]{\rho  \bmttsym{(}  \gamma  \bmttsym{)} \,  \preceq _{ \bmttnt{w} }  \, \bmttnt{e} \implies  \mathfrak{M} ,  \bmttnt{w} ,  \bmttnt{e}    \Vdash  ^{ \rho }  \bmttnt{A} } \text{;} \\
  &  \mathfrak{M} ,  \bmttnt{w} ,  \bmttnt{d}    \vDash  ^{ \rho }   \forall  \gamma_{{\mathrm{1}}} \within  \gamma_{{\mathrm{2}}} . \bmttnt{A}   && \iff \forall \bmttnt{e} \,   \succeq  _{ \bmttnt{w} }  \, \rho  \bmttsym{(}  \gamma_{{\mathrm{2}}}  \bmttsym{)}. \parens[\big]{ \mathfrak{M} ,  \bmttnt{w} ,  \bmttnt{d}    \Vdash  ^{  \rho  \cdot [   \gamma_{{\mathrm{1}}} \mapsto \bmttnt{e}   ]  }  \bmttnt{A} } \text{.}
\end{alignat*}
It should be noted that $ \Vdash $ and $ \vDash $ are defined mutually recursively.

As \bml-model captures growth of \bml-structure, it has two intuitionistic relations: $ \preccurlyeq $ in \bml-model and $ \preceq $ in \bml-structure. Hence, persistency for \bml-model is stated in a more elaborate manner.
\begin{lemma*}[Semantic Persistency]\label{claim:semantic-monotonicity}
  Let $\mathfrak{M}$ be a \bml-model and suppose $\bmttnt{w} \,  \preccurlyeq  \, \bmttnt{v}$ and $\bmttnt{d} \,  \preceq _{ \bmttnt{v} }  \, \bmttnt{e}$. If\/ $ \mathfrak{M} ,  \bmttnt{w} ,  \bmttnt{d}    \Vdash  ^{ \rho }  \bmttnt{A} $, then\/ $ \mathfrak{M} ,  \bmttnt{v} ,  \bmttnt{e}    \Vdash  ^{ \rho }  \bmttnt{A} $.
\end{lemma*}
Comparing this to \Zcref{claim:monotonicity}, $ \preccurlyeq $ corresponds to inclusion between contexts. Building upon this intuition, we build a canonical model for \bml in the next section.

As we have seen in this section, the birelational Kripke structure for syntactic structure of programs does not necessarily correspond to Kripke semantics for its logical counterpart. Rather, its Kripke semantics is captured by a growing family of such syntactic structures.

\section{Soundness and Completeness}\label{sec:logic:completeness}
We confirm the correspondence between the Kripke semantics and the proof system. In preparation, we introduce additional definitions for Kripke semantics. Unlike satisfaction of a formula, the accessibility of each transition relation is independent of $\bmttnt{d}$ and interpreted as:
\begin{alignat*}{2}
  &  \mathfrak{M} ,  \bmttnt{w}    \Vdash  ^{ \rho }  \gamma_{{\mathrm{1}}} \, \trianglelefteq \, \gamma_{{\mathrm{2}}}  && \iff \rho  \bmttsym{(}  \gamma_{{\mathrm{1}}}  \bmttsym{)} \, \trianglelefteq \, \rho  \bmttsym{(}  \gamma_{{\mathrm{2}}}  \bmttsym{)}
  \text{.}
  \tag{where $\mathord{ \trianglelefteq } \in \lbrace  \preceq ,  \sqsubseteq  \rbrace$}
\end{alignat*}
The interpretation of a context $\Gamma$ is determined based on $ \mathrm{pos} ( \Gamma ) $ as follows:
\begin{alignat*}{2}
  &  \mathfrak{M} ,  \bmttnt{w}    \Vdash  ^{ \rho }   \varepsilon   && \iff \text{always;} \\
  &  \mathfrak{M} ,  \bmttnt{w}    \Vdash  ^{ \rho }  \Gamma  \bmttsym{,}   { \bmttnt{A} }^{\binder{ \gamma } }   && \iff  \mathfrak{M} ,  \bmttnt{w}    \Vdash  ^{ \rho }  \Gamma ,\enskip \mathfrak{M} ,  \bmttnt{w}    \Vdash  ^{ \rho }   \mathrm{pos} ( \Gamma )  \, \preceq \, \gamma , \And  \mathfrak{M} ,  \bmttnt{w} ,  \rho  \bmttsym{(}  \gamma  \bmttsym{)}    \Vdash  ^{ \rho }  \bmttnt{A}  \text{;} \\
  &  \mathfrak{M} ,  \bmttnt{w}    \Vdash  ^{ \rho }  \Gamma  \bmttsym{,}   \mathord{\blacktriangleright} ^{\binder{ \gamma_{{\mathrm{2}}} } \within \gamma_{{\mathrm{1}}} }   && \iff  \mathfrak{M} ,  \bmttnt{w}    \Vdash  ^{ \rho }  \Gamma ,\enskip \mathfrak{M} ,  \bmttnt{w}    \Vdash  ^{ \rho }   \mathrm{pos} ( \Gamma )  \, \sqsubseteq \, \gamma_{{\mathrm{2}}} , \And  \mathfrak{M} ,  \bmttnt{w}    \Vdash  ^{ \rho }  \gamma_{{\mathrm{1}}} \, \preceq \, \gamma_{{\mathrm{2}}}  \text{;} \\
  &  \mathfrak{M} ,  \bmttnt{w}    \Vdash  ^{ \rho }  \Gamma  \bmttsym{,}   \mathord{\blacktriangleleft} ^{ \gamma }   && \iff  \mathfrak{M} ,  \bmttnt{w}    \Vdash  ^{ \rho }  \Gamma  \text{;} \\
  &  \mathfrak{M} ,  \bmttnt{w}    \Vdash  ^{ \rho }  \Gamma  \bmttsym{,}   \binder{ \gamma_{{\mathrm{2}}} } \within \gamma_{{\mathrm{1}}}   && \iff  \mathfrak{M} ,  \bmttnt{w}    \Vdash  ^{ \rho }  \Gamma  \And  \mathfrak{M} ,  \bmttnt{w}    \Vdash  ^{ \rho }  \gamma_{{\mathrm{1}}} \, \preceq \, \gamma_{{\mathrm{2}}}  \text{.}
\end{alignat*}
The \emph{semantic consequence} is defined accordingly:
\begin{alignat*}{2}
  & \Gamma \,  \Vdash  \, \bmttnt{A} && \iff \forall \mathfrak{M},\bmttnt{w},\rho. \parens[\big]{ \mathfrak{M} ,  \bmttnt{w}    \Vdash  ^{ \rho }  \Gamma  \implies  \mathfrak{M} ,  \bmttnt{w} ,  \rho  \bmttsym{(}   \mathrm{pos} ( \Gamma )   \bmttsym{)}    \Vdash  ^{ \rho }  \bmttnt{A} } \text{;} \\
  & \Gamma \,  \Vdash  \, \gamma_{{\mathrm{1}}} \, \trianglelefteq \, \gamma_{{\mathrm{2}}} && \iff \forall \mathfrak{M},\bmttnt{w},\rho.\parens[\big]{ \mathfrak{M} ,  \bmttnt{w}    \Vdash  ^{ \rho }  \Gamma  \implies  \mathfrak{M} ,  \bmttnt{w}    \Vdash  ^{ \rho }  \gamma_{{\mathrm{1}}} \, \trianglelefteq \, \gamma_{{\mathrm{2}}} } \tag{where $\mathord{ \trianglelefteq } \in \lbrace  \preceq ,  \sqsubseteq  \rbrace$} \text{.}
\end{alignat*}
Now we can state soundness:
\begin{theorem*}[Kripke Soundness]\label{claim:Kripke-soundness}\quitvmode\samepage
  \begin{enumerate}
  \item\label{item:Kripke-soundness:types}
    If\/ $\Gamma  \vdash  \bmttnt{A}$, then\/ $\Gamma \,  \Vdash  \, \bmttnt{A}$.
  \item\label{item:Kripke-soundness:scopes}
    If\/ $\Gamma  \vdash  \gamma_{{\mathrm{1}}}  \preceq  \gamma_{{\mathrm{2}}}$, then\/ $\Gamma \,  \Vdash  \, \gamma_{{\mathrm{1}}} \, \preceq \, \gamma_{{\mathrm{2}}}$.
  \item\label{item:Kripke-soundness:stages}
    If\/ $ \Gamma \vdash \gamma_{{\mathrm{1}}} \sqsubseteq \gamma_{{\mathrm{2}}} $, then\/ $\Gamma \,  \Vdash  \, \gamma_{{\mathrm{1}}} \, \sqsubseteq \, \gamma_{{\mathrm{2}}}$.
  \end{enumerate}
\end{theorem*}

To prove completeness we use a \textit{canonical-model} construction:
\begin{definition}[Canonical Model]\label{definition:canonical-model}
  $ \can{ \mathfrak{M} }  =  \langle  \can{ W }  ,   \can  \preccurlyeq   ,   \lbrace  \can{ M _{ \Gamma } }  \rbrace_{ \Gamma \in  \can{ W }  }  \rangle $ is defined as follows:
  \begin{itemize}
  \item $ \can{ W } $ is the set of all well-formed contexts;
  \item $\Gamma \,  \can  \preccurlyeq   \, \Delta \iff \exists \Gamma'.(\Delta = \Gamma,\Gamma')$;
  \item $ \can{ M _{ \Gamma } }  =  \langle  \can{ D _{ \Gamma } }  ,   \can{ \preceq _{ \Gamma } }  ,   \can{ \sqsubseteq _{ \Gamma } }  ,   \can{ V _{ \Gamma } }  ,   \can{  \exclam  _{ \Gamma } }  \rangle $ where
    \begin{itemize}
    \item $ \can{ D _{ \Gamma } }  =  \mathbf{Dom}_{\token{C} }( \Gamma ) $ with $ \can{  \exclam  _{ \Gamma } }  =  \exclam $;
    \item $\gamma_{{\mathrm{1}}} \,  \can{ \preceq _{ \Gamma } }  \, \gamma_{{\mathrm{2}}} \iff \Gamma  \vdash  \gamma_{{\mathrm{1}}}  \preceq  \gamma_{{\mathrm{2}}}$;
    \item $\gamma_{{\mathrm{1}}} \,  \can{ \sqsubseteq _{ \Gamma } }  \, \gamma_{{\mathrm{2}}} \iff  \Gamma \vdash \gamma_{{\mathrm{1}}} \sqsubseteq \gamma_{{\mathrm{2}}} $;
    \item $\gamma \, \in \,  \can{ V _{ \Gamma } }   \bmttsym{(}  \bmttnt{p}  \bmttsym{)} \iff \Gamma  \bmttsym{,}   \mathord{\blacktriangleleft} ^{  \mathord{\boldsymbol{!} }  }   \vdash   \Box^{\mathord{\succeq}  \gamma }  \bmttnt{p} $.
    \end{itemize}
  \end{itemize}
\end{definition}

\begin{lemma}
  $ \can{ \mathfrak{M} } $ is a\/ \bml-model.
\end{lemma}

This canonical model clarifies how the layered structure of \bml-model corresponds to the structure of contexts in our proof system: each $ \can{ M _{ \Gamma } } $ models the structure formed by \emph{classifiers}, whereas $ \can{ \mathfrak{M} } $ models the structure formed by \emph{contexts}. The relation $ \can{ \preceq _{ \Gamma } } $ in a \bml-structure of the canonical model corresponds to the judgment $\Gamma  \vdash  \gamma_{{\mathrm{1}}}  \preceq  \gamma_{{\mathrm{2}}}$ while the relation $ \preccurlyeq $ corresponds to an order between contexts.

In stating truth lemma, it should be noted that there is a subtle difference between Kripke semantics and proof system. In semantics, the truth of a formula is always defined at each point of a model, whereas in syntax, only its validity at~$ \mathrm{pos} ( \Gamma ) $ can be asserted under $\Gamma$;
in this respect, \bml differs from ordinary labelled proof systems\nobreakspace (cf.\@ e.g.,~\cite{negri2005proof,simpson1994proof}). The restriction, however, does not pose a problem for establishing their correspondence, because $\Gamma  \bmttsym{,}   \mathord{\blacktriangleleft} ^{  \mathord{\boldsymbol{!} }  }   \vdash   \Box^{\mathord{\succeq}  \gamma }  \bmttnt{A} $ can be used instead to represent the validity of~$\bmttnt{A}$ at $\gamma$ under~$\Gamma$.
Here, the bounded modality expresses monotonicity, analogous to the $\logicS4$ modality in the Gödel--McKinsey--Tarski translation.

The \emph{canonical $\Gamma$-assignment} $ \can{ \rho _{ \Gamma } } $ is an assignment that maps each $\gamma \, \in \,  \mathbf{Dom}_{\token{C} }( \Gamma ) $ to itself, and we write $ \Gamma ,  \gamma  \can{  \Vdash  }  \bmttnt{A} $ if $  \can{ \mathfrak{M} }  ,  \Gamma ,  \gamma    \Vdash  ^{  \can{ \rho _{ \Gamma } }  }  \bmttnt{A} $. The \emph{truth lemma} is stated in the following form:

\begin{lemma*}[Truth Lemma]\label{claim:truth-lemma}
  $ \Gamma ,  \gamma  \can{  \Vdash  }  \bmttnt{A}  \iff \Gamma  \bmttsym{,}   \mathord{\blacktriangleleft} ^{  \mathord{\boldsymbol{!} }  }   \vdash   \Box^{\mathord{\succeq}  \gamma }  \bmttnt{A} $.
\end{lemma*}

Finally, we show Kripke completeness:

\begin{theorem*}[Kripke Completeness]\label{claim:Kripke-completeness}\quitvmode\samepage
  \begin{enumerate}
  \item If\/ $\Gamma \,  \Vdash  \, \bmttnt{A}$, then\/ $\Gamma  \vdash  \bmttnt{A}$.
  \item If\/ $\Gamma \,  \Vdash  \, \gamma_{{\mathrm{1}}} \, \preceq \, \gamma_{{\mathrm{2}}}$, then\/ $\Gamma  \vdash  \gamma_{{\mathrm{1}}}  \preceq  \gamma_{{\mathrm{2}}}$.
  \item If\/ $\Gamma \,  \Vdash  \, \gamma_{{\mathrm{1}}} \, \sqsubseteq \, \gamma_{{\mathrm{2}}}$, then\/ $ \Gamma \vdash \gamma_{{\mathrm{1}}} \sqsubseteq \gamma_{{\mathrm{2}}} $.
  \end{enumerate}
\end{theorem*}

\section{Lambda Calculus and Metatheory}\label{sec:calc}
Under the Curry--Howard isomorphism~\cite{book/GirardTL89,book/SorensenU2006}, we can consider a typed lambda calculus that corresponds to our natural-deduction system.
The syntactic categories of our term calculus are defined as follows:

\vspace{0.5em}
\begin{tabular}{llcl}
  \textbf{Types} & $\bmttnt{A}$, $\bmttnt{B}$ & $\Coloneqq$ & $\bmttnt{p} \mid \bmttnt{A}  \mathbin{\rightarrow}  \bmttnt{B} \mid  \Box^{\mathord{\succeq}  \gamma }  \bmttnt{A}  \mid  \forall  \gamma_{{\mathrm{1}}} \within  \gamma_{{\mathrm{2}}} . \bmttnt{A} $ \\
  \textbf{Contexts} & $\Gamma$, $\Delta$ & $\Coloneqq$ & $ \varepsilon  \mid \Gamma  \bmttsym{,}   \binder{ \bmttmv{x} }\has@{ \gamma } \bmttnt{A}  \mid \Gamma  \bmttsym{,}   \tau : \mathord{\blacktriangleright} ^{\binder{ \gamma_{{\mathrm{1}}} } \within \gamma_{{\mathrm{2}}} }  \mid \Gamma  \bmttsym{,}   \mathord{\blacktriangleleft} _{ \bmttnt{T} }^{ \gamma }  \mid \Gamma  \bmttsym{,}   \binder{ \gamma_{{\mathrm{1}}} } \within \gamma_{{\mathrm{2}}} $ \\
  \textbf{MT-witness} & $\bmttnt{T}$ & $\Coloneqq$ & $ \overrightarrow{ \tau } $ \\
  \textbf{Terms} & $\bmttnt{M}$, $\bmttnt{N}$ & $\Coloneqq$ & $\bmttmv{x} \mid  \lambda \binder{ \bmttmv{x} }\has@{ \gamma } \bmttnt{A} \ldotp \bmttnt{M}  \mid  \bmttnt{M_{{\mathrm{1}}}}   \bmttnt{M_{{\mathrm{2}}}}  \mid  \mathbf{quo} (\binder{ \tau })\lbrace^{\binder{ \gamma_{{\mathrm{1}}} } \within \gamma_{{\mathrm{2}}} }  \bmttnt{M} \rbrace  \mid  \mathbf{unq} _{ \bmttnt{T} }\lbrace^{ \gamma } \bmttnt{M} \rbrace $ \\
  && $\mid$ & $ \lambda \binder{ \gamma_{{\mathrm{1}}} }\within  \gamma_{{\mathrm{2}}} . \bmttnt{M}  \mid  \bmttnt{M} \gamma $
\end{tabular}
\vspace{0.5em}

In our calculus, we provide two sorts of proof terms: one for modal transitions and another for proofs of propositions. For modal transitions, we introduce atomic modal transition witness $\tau$ for the hypothetical modal transition introduced by $ \mathord{\blacktriangleright} $. A modal transition witness $\bmttnt{T}$ represents a proof of a modal transition, which is a sequence of atomic modal transition witnesses. We write $ \varepsilon $ for an empty sequence and $ \bmttnt{T_{{\mathrm{1}}}}  +  \bmttnt{T_{{\mathrm{2}}}} $ for the concatenation of two modal transition witnesses. Modal transition judgments are extended to carry modal transition witnesses as $ \Gamma \vdash \bmttnt{T} : \gamma_{{\mathrm{1}}} \sqsubseteq \gamma_{{\mathrm{2}}} $. Its derivation rules are given in \Zcref{fig:termassignment}.
\begin{figure}[bt]
  \centering
  \begin{mathpar}[\footnotesize]
    \begin{prooftree}
      \caption{$ \sqsubseteq $-Lift}\label{rule:calc-mtrans-lift}
      \hypo{\Gamma  \vdash  \gamma_{{\mathrm{1}}}  \preceq  \gamma_{{\mathrm{2}}}}
      \infer1{ \Gamma \vdash  \varepsilon  : \gamma_{{\mathrm{1}}} \sqsubseteq \gamma_{{\mathrm{2}}} }
    \end{prooftree}
    \and
    \begin{prooftree}
      \caption{$ \sqsubseteq $-Trans}\label{rule:calc-mtrans-trans}
      \hypo{ \Gamma \vdash \bmttnt{T_{{\mathrm{1}}}} : \gamma_{{\mathrm{1}}} \sqsubseteq \gamma_{{\mathrm{2}}} }
      \hypo{ \Gamma \vdash \bmttnt{T_{{\mathrm{2}}}} : \gamma_{{\mathrm{2}}} \sqsubseteq \gamma_{{\mathrm{3}}} }
      \infer2{ \Gamma \vdash  \bmttnt{T_{{\mathrm{1}}}}  +  \bmttnt{T_{{\mathrm{2}}}}  : \gamma_{{\mathrm{1}}} \sqsubseteq \gamma_{{\mathrm{3}}} }
    \end{prooftree}
    \and
    \begin{prooftree}
      \caption{$ \sqsubseteq $-$ \mathord{\blacktriangleright} $}\label{rule:calc-mtrans-open}
      \hypo{}
      \infer1{  \Gamma_{{\mathrm{1}}} ^{\position{ \gamma_{{\mathrm{1}}} } }   \bmttsym{,}   \tau : \mathord{\blacktriangleright} ^{\binder{ \gamma_{{\mathrm{2}}} } \within \gamma_{{\mathrm{3}}} }   \bmttsym{,}  \Gamma_{{\mathrm{2}}} \vdash \tau : \gamma_{{\mathrm{1}}} \sqsubseteq \gamma_{{\mathrm{2}}} }
    \end{prooftree}\\
    \begin{prooftree}[separation=1.0em]
      \caption{WF-$ \mathord{\blacktriangleleft} $}\label{rule:calc-ctx-wf-shut}
      \hypo{\vdash   \Gamma ^{\position{ \gamma } }   \hasType \, \bmttkw{ctx}}
      \hypo{  \Gamma ^{\position{ \gamma } }  \vdash \bmttnt{T} : \delta \sqsubseteq \gamma }
      \infer2{\vdash   \Gamma ^{\position{ \gamma } }   \bmttsym{,}   \mathord{\blacktriangleleft} _{ \bmttnt{T} }^{ \delta }   \hasType \, \bmttkw{ctx}}
    \end{prooftree}\\
    \begin{prooftree}[separation=1.0em]
      \caption{Var}\label{rule:type-var}
      \hypo{ \binder{ \bmttmv{x} }\has@{ \gamma_{{\mathrm{1}}} } \bmttnt{A}  \, \in \,  \Gamma ^{\position{ \gamma_{{\mathrm{2}}} } } }
      \hypo{ \Gamma ^{\position{ \gamma_{{\mathrm{2}}} } }   \vdash  \gamma_{{\mathrm{1}}}  \preceq  \gamma_{{\mathrm{2}}}}
      \infer2{ \Gamma ^{\position{ \gamma_{{\mathrm{2}}} } }   \vdash  \bmttmv{x}  \hasType  \bmttnt{A}}
    \end{prooftree}
    \and
    \begin{prooftree}[separation=1.0em]
      \caption{$\to$-I}\label{rule:type-to-i}
      \hypo{\Gamma  \bmttsym{,}   \binder{ \bmttmv{x} }\has@{ \gamma } \bmttnt{A_{{\mathrm{1}}}}   \vdash  \bmttnt{M_{{\mathrm{1}}}}  \hasType  \bmttnt{A_{{\mathrm{2}}}}}
      \hypo{\gamma \, \notin \,  \token{FC}( \bmttnt{A_{{\mathrm{2}}}} ) }
      \infer2{\Gamma  \vdash   \lambda \binder{ \bmttmv{x} }\has@{ \gamma } \bmttnt{A_{{\mathrm{1}}}} \ldotp \bmttnt{M_{{\mathrm{1}}}}   \hasType  \bmttnt{A_{{\mathrm{1}}}}  \mathbin{\rightarrow}  \bmttnt{A_{{\mathrm{2}}}}}
    \end{prooftree}
    \and
    \begin{prooftree}[separation=1.0em]
      \caption{$\to$-E}\label{rule:type-to-e}
      \hypo{\Gamma  \vdash  \bmttnt{M_{{\mathrm{1}}}}  \hasType  \bmttnt{A_{{\mathrm{1}}}}  \mathbin{\rightarrow}  \bmttnt{A_{{\mathrm{2}}}}}
      \hypo{\Gamma  \vdash  \bmttnt{M_{{\mathrm{2}}}}  \hasType  \bmttnt{A_{{\mathrm{1}}}}}
      \infer2{\Gamma  \vdash   \bmttnt{M_{{\mathrm{1}}}}   \bmttnt{M_{{\mathrm{2}}}}   \hasType  \bmttnt{A_{{\mathrm{2}}}}}
    \end{prooftree}
    \and
    \begin{prooftree}
      \caption{$\Box$-I}\label{rule:type-bm-i}
      \hypo{\Gamma  \bmttsym{,}   \tau : \mathord{\blacktriangleright} ^{\binder{ \gamma_{{\mathrm{1}}} } \within \gamma_{{\mathrm{2}}} }   \vdash  \bmttnt{M}  \hasType  \bmttnt{A}}
      \hypo{\gamma_{{\mathrm{1}}} \, \notin \,  \token{FC}( \bmttnt{A} ) }
      \infer2{\Gamma  \vdash   \mathbf{quo} (\binder{ \tau })\lbrace^{\binder{ \gamma_{{\mathrm{1}}} } \within \gamma_{{\mathrm{2}}} }  \bmttnt{M} \rbrace   \hasType   \Box^{\mathord{\succeq}  \gamma_{{\mathrm{2}}} }  \bmttnt{A} }
    \end{prooftree}
    \and
    \begin{prooftree}
      \caption{$\Box$-E}\label{rule:type-bm-e}
      \hypo{ \Gamma ^{\position{ \gamma_{{\mathrm{1}}} } }   \bmttsym{,}   \mathord{\blacktriangleleft} _{ \bmttnt{T} }^{ \gamma_{{\mathrm{2}}} }   \vdash  \bmttnt{M}  \hasType   \Box^{\mathord{\succeq}  \gamma_{{\mathrm{3}}} }  \bmttnt{A} }
      \hypo{ \Gamma ^{\position{ \gamma_{{\mathrm{1}}} } }   \vdash  \gamma_{{\mathrm{3}}}  \preceq  \gamma_{{\mathrm{1}}}}
      \infer2{ \Gamma ^{\position{ \gamma_{{\mathrm{1}}} } }   \vdash   \mathbf{unq} _{ \bmttnt{T} }\lbrace^{ \gamma_{{\mathrm{2}}} } \bmttnt{M} \rbrace   \hasType  \bmttnt{A}}
    \end{prooftree}
    \and
    \begin{prooftree}
      \caption{$\forall$-I}\label{rule:type-polycls-i}
      \hypo{\Gamma  \bmttsym{,}   \binder{ \gamma_{{\mathrm{1}}} } \within \gamma_{{\mathrm{2}}}   \vdash  \bmttnt{M}  \hasType  \bmttnt{A}}
      \infer1{\Gamma  \vdash   \lambda \binder{ \gamma_{{\mathrm{1}}} }\within  \gamma_{{\mathrm{2}}} . \bmttnt{M}   \hasType   \forall  \gamma_{{\mathrm{1}}} \within  \gamma_{{\mathrm{2}}} . \bmttnt{A} }
    \end{prooftree}
    \and
    \begin{prooftree}
      \caption{$\forall$-E}\label{rule:type-polycls-e}
      \hypo{\Gamma  \vdash  \bmttnt{M}  \hasType   \forall  \gamma_{{\mathrm{1}}} \within  \gamma_{{\mathrm{2}}} . \bmttnt{A} }
      \hypo{\Gamma  \vdash  \gamma_{{\mathrm{2}}}  \preceq  \gamma_{{\mathrm{3}}}}
      \infer2{\Gamma  \vdash   \bmttnt{M} \gamma_{{\mathrm{3}}}   \hasType   \bmttnt{A} [ \binder{ \gamma_{{\mathrm{1}}} } \coloneqq \gamma_{{\mathrm{3}}}  ] }
    \end{prooftree}
  \end{mathpar}
  \caption{Curated Rules for the Term Assignment System.}\label{fig:termassignment}
  \vspace{-0.75\baselineskip}
\end{figure}
\ref{rule:calc-mtrans-lift} lifts intuitionistic transition to modal transition with an empty witness $ \varepsilon $, and \ref{rule:calc-mtrans-trans} concatenates two modal transition witnesses. \ref{rule:calc-mtrans-open} uses the atomic modal transition witness $\tau$ introduced by $ \mathord{\blacktriangleright} $.
As well-formedness of $ \mathord{\blacktriangleleft} $ requires modal transition, now $ \mathord{\blacktriangleleft} $ explicitly carries a modal transition witness $\bmttnt{T}$, as shown in \ref{rule:calc-ctx-wf-shut}.

To assign terms to proofs, we assign variables to hypotheses like $ \binder{ \bmttmv{x} }\has@{ \gamma } \bmttnt{A} $. Then, proof terms $\bmttnt{M}$ represent proofs of propositions, as in standard lambda calculi. \Zcref{fig:termassignment} provides derivation rules for typing judgments $\Gamma  \vdash  \bmttnt{M}  \hasType  \bmttnt{A}$.
In the resulting calculus, one can observe that terms for \ref{rule:type-to-i}, \ref{rule:type-bm-i}, \ref{rule:type-bm-e}, \ref{rule:type-polycls-i} correspond to items in context; thus, context structure and its \bml-structure capture the syntactic structure of lambda-terms. In particular, terms for \ref{rule:type-bm-i} and \ref{rule:type-bm-e} appear as quasi-quotation constructs, which is commonly observed in modal lambda-calculi~\cite{book/MartiniM96,journals/jacm/DaviesP01,conf/fossacs/Clouston18,journals/jacm/Davies17}.

As an example, \Zcref{fig:lambdaterms} shows how our calculus represents the example in \Zcref{fig:syntacticstructure:metaml}, and its corresponding \bml-structure, which demonstrates that \bml allows typing programs that were not typed in the standard S4- or LTL-based calculi.
\begin{figure}[bpt]
  \hfil
  \raisebox{-\totalheight}{\includegraphics[page=1]{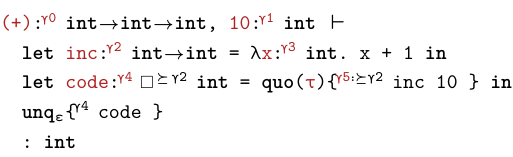}}\hfil
  \raisebox{-\totalheight}{\includegraphics[page=2]{images/csp-lambda-term.pdf}}\hfil
  \vspace{-0.75\baselineskip}
  \caption{The example in \Zcref{fig:syntacticstructure:metaml} written in our calculus, and its corresponding \bml-structure. As we lack constants, we introduce \code{+} and \code{10} as variables. \code{let}-expression is a shorthand for $\lambda$-abstraction and application.}
  \label{fig:lambdaterms}
  \vspace{-0.75\baselineskip}
\end{figure}
Here, $ \mathbf{unq} $ carries an empty modal transition witness $ \varepsilon $: this transition does not include actual modal transitions, and it can be regarded as runtime evaluation. Thus, modal transition witnesses can be regarded as providing staging information. This is quite similar to Kripke-style S4 modal calculus by Davies and Pfenning~\cite{journals/jacm/DaviesP01}, where unquote carries a natural number representing the number of modal transitions. We further discuss how we use modal transition witness for staged evaluation in \Zcref{sec:staging}.

\subsection{Logical Harmony and Reduction Semantics}
In preparation for discussion on proof reduction, we introduce three meta operations and confirm their properties. Classifier substitution $ \bmttsym{(}   \mathord{-}   \bmttsym{)} [ \binder{ \gamma_{{\mathrm{1}}} } \coloneqq \gamma_{{\mathrm{2}}}  ] $ operates on classifiers, contexts, types and terms, substituting free occurrences of $\gamma_{{\mathrm{1}}}$ with $\gamma_{{\mathrm{2}}}$. Modal transition substitution $ \bmttsym{(}   \mathord{-}   \bmttsym{)} [ \binder{ \gamma_{{\mathrm{1}}} } \coloneqq \gamma_{{\mathrm{2}}} , \binder{ \tau }  \coloneqq   \bmttnt{T}  ] $ operates on contexts and terms, substituting free occurrences of $\gamma_{{\mathrm{1}}}$ and $\tau$ with $\gamma_{{\mathrm{2}}}$ and $\bmttnt{T}$, respectively. Variable substitution $ \bmttsym{(}   \mathord{-}   \bmttsym{)} [ \binder{ \gamma_{{\mathrm{1}}} } \coloneqq \gamma_{{\mathrm{2}}} , \binder{ \bmttmv{x} }  \coloneqq   \bmttnt{M}  ] $ operates on terms, substituting free occurrences of $\gamma_{{\mathrm{1}}}$ and $\bmttmv{x}$ with $\gamma_{{\mathrm{2}}}$ and $\bmttnt{M}$, respectively. Then, we can confirm that the following lemmas hold.

\begin{lemma*}[Variable Substitution]\label{claim:varsubst}
  If\/ $ \Gamma_{{\mathrm{1}}} ^{\position{ \gamma_{{\mathrm{1}}} } }   \bmttsym{,}   \binder{ \bmttmv{x} }\has@{ \gamma_{{\mathrm{2}}} } \bmttnt{A}   \bmttsym{,}  \Gamma_{{\mathrm{2}}}  \vdash  \bmttnt{M_{{\mathrm{1}}}}  \hasType  \bmttnt{B}$ and\/ $\Gamma_{{\mathrm{1}}}  \vdash  \bmttnt{M_{{\mathrm{2}}}}  \hasType  \bmttnt{A}$, then\/ $ \Gamma_{{\mathrm{1}}}  \bmttsym{,}  \Gamma_{{\mathrm{2}}} [ \binder{ \gamma_{{\mathrm{2}}} } \coloneqq \gamma_{{\mathrm{1}}}  ]   \vdash   \bmttnt{M_{{\mathrm{1}}}} [ \binder{ \gamma_{{\mathrm{2}}} } \coloneqq \gamma_{{\mathrm{1}}} , \binder{ \bmttmv{x} }  \coloneqq   \bmttnt{M_{{\mathrm{2}}}}  ]   \hasType   \bmttnt{B} [ \binder{ \gamma_{{\mathrm{2}}} } \coloneqq \gamma_{{\mathrm{1}}}  ] $.
\end{lemma*}

\begin{lemma*}[Rebasing]\label{claim:rebasing}
  If\/ $ \Gamma_{{\mathrm{1}}} ^{\position{ \gamma_{{\mathrm{1}}} } }   \bmttsym{,}   \mathord{\blacktriangleleft} _{ \bmttnt{T} }^{ \gamma_{{\mathrm{2}}} }   \bmttsym{,}   \tau : \mathord{\blacktriangleright} ^{\binder{ \gamma_{{\mathrm{3}}} } \within \gamma_{{\mathrm{4}}} }   \bmttsym{,}  \Gamma_{{\mathrm{2}}}  \vdash  \bmttnt{M_{{\mathrm{1}}}}  \hasType  \bmttnt{A}$ and\/ $\Gamma_{{\mathrm{1}}}  \vdash  \gamma_{{\mathrm{4}}}  \preceq  \gamma_{{\mathrm{1}}}$, then\/ $ \Gamma_{{\mathrm{1}}}  \bmttsym{,}  \Gamma_{{\mathrm{2}}} [ \binder{ \gamma_{{\mathrm{3}}} } \coloneqq \gamma_{{\mathrm{1}}} , \binder{ \tau }  \coloneqq   \bmttnt{T}  ]   \vdash   \bmttnt{M_{{\mathrm{1}}}} [ \binder{ \gamma_{{\mathrm{3}}} } \coloneqq \gamma_{{\mathrm{1}}} , \binder{ \tau }  \coloneqq   \bmttnt{T}  ]   \hasType   \bmttnt{A} [ \binder{ \gamma_{{\mathrm{3}}} } \coloneqq \gamma_{{\mathrm{1}}}  ] $.
\end{lemma*}

\begin{lemma*}[Classifier Substitution]\label{claim:clssubst}
  If\/ $\Gamma_{{\mathrm{1}}}  \bmttsym{,}   \binder{ \gamma_{{\mathrm{1}}} } \within \gamma_{{\mathrm{2}}}   \bmttsym{,}  \Gamma_{{\mathrm{2}}}  \vdash  \bmttnt{M}  \hasType  \bmttnt{A}$ and\/ $\Gamma_{{\mathrm{1}}}  \vdash  \gamma_{{\mathrm{2}}}  \preceq  \gamma_{{\mathrm{3}}}$, then\/ $ \Gamma_{{\mathrm{1}}}  \bmttsym{,}  \Gamma_{{\mathrm{2}}} [ \binder{ \gamma_{{\mathrm{1}}} } \coloneqq \gamma_{{\mathrm{3}}}  ]   \vdash   \bmttnt{M} [ \binder{ \gamma_{{\mathrm{1}}} } \coloneqq \gamma_{{\mathrm{3}}}  ]   \hasType   \bmttnt{A} [ \binder{ \gamma_{{\mathrm{1}}} } \coloneqq \gamma_{{\mathrm{3}}}  ] $.
\end{lemma*}


Using these lemmas, we can confirm that the introduction and elimination rules in \bml are well balanced with respect to local soundness and local completeness~\cite{journals/mscs/PfenningD01}.
For the sake of space, we only state the local soundness/completeness patterns  in the following statements.

\begin{lemma*}[Local Soundness Patterns]\label{claim:localsoundness}\leavevmode\samepage
 \begin{enumerate}
  \item $ \Gamma ^{\position{ \gamma_{{\mathrm{1}}} } }   \vdash   \bmttsym{(}   \lambda \binder{ \bmttmv{x} }\has@{ \gamma_{{\mathrm{2}}} } \bmttnt{A} \ldotp \bmttnt{M}   \bmttsym{)}   \bmttnt{N}   \hasType  \bmttnt{B} \implies \Gamma  \vdash   \bmttnt{M} [ \binder{ \gamma_{{\mathrm{2}}} } \coloneqq \gamma_{{\mathrm{1}}} , \binder{ \bmttmv{x} }  \coloneqq   \bmttnt{N}  ]   \hasType  \bmttnt{B}$.
  \item $ \Gamma ^{\position{ \gamma_{{\mathrm{1}}} } }   \vdash   \mathbf{unq} _{ \bmttnt{T} }\lbrace^{ \gamma_{{\mathrm{2}}} }  \mathbf{quo} (\binder{ \tau })\lbrace^{\binder{ \gamma_{{\mathrm{3}}} } \within \gamma_{{\mathrm{4}}} }  \bmttnt{M} \rbrace  \rbrace   \hasType  \bmttnt{A} \implies \Gamma  \vdash   \bmttnt{M} [ \binder{ \gamma_{{\mathrm{3}}} } \coloneqq \gamma_{{\mathrm{1}}} , \binder{ \tau }  \coloneqq   \bmttnt{T}  ]   \hasType  \bmttnt{A}$.
  \item $\Gamma  \vdash   \bmttsym{(}   \lambda \binder{ \gamma_{{\mathrm{1}}} }\within  \gamma_{{\mathrm{2}}} . \bmttnt{M}   \bmttsym{)} \gamma_{{\mathrm{3}}}   \hasType  \bmttnt{A} \implies \Gamma  \vdash   \bmttnt{M} [ \binder{ \gamma_{{\mathrm{1}}} } \coloneqq \gamma_{{\mathrm{3}}}  ]   \hasType  \bmttnt{A}$.
 \end{enumerate}
\end{lemma*}

\begin{lemma*}[Local Completeness Patterns]\label{claim:localcompleteness} ($\bmttmv{x}$, $\delta$ and $\tau$ are chosen fresh.)
 \begin{enumerate}
  \item $\Gamma  \vdash  \bmttnt{M}  \hasType  \bmttnt{A}  \mathbin{\rightarrow}  \bmttnt{B} \implies \Gamma  \vdash   \lambda \binder{ \bmttmv{x} }\has@{ \delta } \bmttnt{A} \ldotp \bmttsym{(}   \bmttnt{M}   \bmttmv{x}   \bmttsym{)}   \hasType  \bmttnt{A}  \mathbin{\rightarrow}  \bmttnt{B}$.
  \item $ \Gamma ^{\position{ \gamma_{{\mathrm{1}}} } }   \vdash  \bmttnt{M}  \hasType   \Box^{\mathord{\succeq}  \gamma_{{\mathrm{2}}} }  \bmttnt{A}  \implies \Gamma  \vdash   \mathbf{quo} (\binder{ \tau })\lbrace^{\binder{ \delta } \within \gamma_{{\mathrm{2}}} }   \mathbf{unq} _{ \tau }\lbrace^{ \gamma_{{\mathrm{1}}} } \bmttnt{M} \rbrace  \rbrace   \hasType   \Box^{\mathord{\succeq}  \gamma_{{\mathrm{2}}} }  \bmttnt{A} $.
  \item $\Gamma  \vdash  \bmttnt{M}  \hasType   \forall  \gamma_{{\mathrm{1}}} \within  \gamma_{{\mathrm{2}}} . \bmttnt{A}  \implies \Gamma  \vdash   \lambda \binder{ \delta }\within  \gamma_{{\mathrm{2}}} . \bmttsym{(}   \bmttnt{M} \delta   \bmttsym{)}   \hasType   \forall  \gamma_{{\mathrm{1}}} \within  \gamma_{{\mathrm{2}}} . \bmttnt{A} $.
 \end{enumerate}
\end{lemma*}

Local soundness and completeness patterns can be regarded as $\beta$-reduction and $\eta$-expansion, respectively.
$ \mathcal{E} ^{ \gamma_{{\mathrm{1}}} }_{[  \gamma_{{\mathrm{2}}}  ]} $ is an evaluation context located at $\gamma_{{\mathrm{1}}}$, and whose hole is located at $\gamma_{{\mathrm{2}}}$.
\par\smallskip
\noindent
\begin{tabular}{lcl}
  $ \mathcal{E} ^{ \gamma_{{\mathrm{1}}} }_{[  \gamma_{{\mathrm{2}}}  ]} $ & $\Coloneqq$ & $ \Box  \text{\ (if\ $\gamma_{{\mathrm{1}}} = \gamma_{{\mathrm{2}}}$)} \mid  \lambda\binder{ \bmttmv{x} }_{ \gamma_{{\mathrm{3}}} }^{ \bmttnt{A} }.  \mathcal{E} ^{ \gamma_{{\mathrm{3}}} }_{[  \gamma_{{\mathrm{2}}}  ]}   \mid   \mathcal{E} ^{ \gamma_{{\mathrm{1}}} }_{[  \gamma_{{\mathrm{2}}}  ]}    \bmttnt{M}  \mid  \bmttnt{M}     \mathcal{E} ^{ \gamma_{{\mathrm{1}}} }_{[  \gamma_{{\mathrm{2}}}  ]}   $ \\
  & $\mid$ & $ \mathbf{quo} ( \tau )\lbrace^{\binder{ \gamma_{{\mathrm{3}}} }\within{ \gamma_{{\mathrm{4}}} } }   \mathcal{E} ^{ \gamma_{{\mathrm{3}}} }_{[  \gamma_{{\mathrm{2}}}  ]}  \rbrace  \mid  \mathbf{unq} _{ \bmttnt{T_{{\mathrm{3}}}} }\lbrace^{ \gamma_{{\mathrm{3}}} }  \mathcal{E} ^{ \gamma_{{\mathrm{3}}} }_{[  \gamma_{{\mathrm{2}}}  ]}  \rbrace  \mid  \lambda \binder{ \gamma_{{\mathrm{3}}} }\within  \gamma_{{\mathrm{4}}} .  \mathcal{E} ^{ \gamma_{{\mathrm{1}}} }_{[  \gamma_{{\mathrm{2}}}  ]}   \mid   \mathcal{E} ^{ \gamma_{{\mathrm{1}}} }_{[  \gamma_{{\mathrm{2}}}  ]}  [  \gamma_{{\mathrm{3}}}  ] $ \\
\end{tabular}
\par\smallskip

Then, we define $\beta$-reduction on raw terms, written $ \bmttnt{M_{{\mathrm{1}}}}  \Rightarrow_{\beta}^{ \gamma }  \bmttnt{M_{{\mathrm{2}}}} $.
Unlike standard \textbeta{}-reduction, ours is a family of relations indexed by
classifiers. This classifier stands for the position where the reduction takes place.
\begin{definition}[$\beta$-Reduction]\label{def:beta}
$ \bmttnt{M_{{\mathrm{1}}}}  \Rightarrow_{\beta}^{ \gamma }  \bmttnt{M_{{\mathrm{2}}}} $ is defined by the following rules.
 \begin{align*}
    \mathcal{E} ^{ \gamma_{{\mathrm{1}}} }_{[  \gamma_{{\mathrm{2}}}  ]}   \bmttsym{[}   \bmttsym{(}   \lambda \binder{ \bmttmv{x} }\has@{ \gamma_{{\mathrm{3}}} } \bmttnt{A} \ldotp \bmttnt{M}   \bmttsym{)}   \bmttnt{N}   \bmttsym{]} &  \Rightarrow_{\beta}^{ \gamma_{{\mathrm{1}}} }   \mathcal{E} ^{ \gamma_{{\mathrm{1}}} }_{[  \gamma_{{\mathrm{2}}}  ]}   \bmttsym{[}   \bmttnt{M} [ \binder{ \gamma_{{\mathrm{3}}} } \coloneqq \gamma_{{\mathrm{2}}} , \binder{ \bmttmv{x} }  \coloneqq   \bmttnt{N}  ]   \bmttsym{]} \\
    \mathcal{E} ^{ \gamma_{{\mathrm{1}}} }_{[  \gamma_{{\mathrm{2}}}  ]}   \bmttsym{[}   \mathbf{unq} _{ \bmttnt{T} }\lbrace^{ \gamma_{{\mathrm{3}}} }  \mathbf{quo} (\binder{ \tau })\lbrace^{\binder{ \gamma_{{\mathrm{4}}} } \within \gamma_{{\mathrm{5}}} }  \bmttnt{M} \rbrace  \rbrace   \bmttsym{]} &  \Rightarrow_{\beta}^{ \gamma_{{\mathrm{1}}} }   \mathcal{E} ^{ \gamma_{{\mathrm{1}}} }_{[  \gamma_{{\mathrm{2}}}  ]}   \bmttsym{[}   \bmttnt{M} [ \binder{ \gamma_{{\mathrm{4}}} } \coloneqq \gamma_{{\mathrm{2}}} , \binder{ \tau }  \coloneqq   \bmttnt{T}  ]   \bmttsym{]}  \\
    \mathcal{E} ^{ \gamma_{{\mathrm{1}}} }_{[  \gamma_{{\mathrm{2}}}  ]}   \bmttsym{[}   \bmttsym{(}   \lambda \binder{ \gamma_{{\mathrm{3}}} }\within  \gamma_{{\mathrm{4}}} . \bmttnt{M}   \bmttsym{)} \gamma_{{\mathrm{5}}}   \bmttsym{]} &  \Rightarrow_{\beta}^{ \gamma_{{\mathrm{1}}} }   \mathcal{E} ^{ \gamma_{{\mathrm{1}}} }_{[  \gamma_{{\mathrm{2}}}  ]}   \bmttsym{[}   \bmttnt{M} [ \binder{ \gamma_{{\mathrm{3}}} } \coloneqq \gamma_{{\mathrm{5}}}  ]   \bmttsym{]} 
 \end{align*}
 We write\/ $ \Rightarrow_{\beta *}^{ \gamma } $ for the reflexive and transitive closure of\/ $ \Rightarrow_{\beta}^{ \gamma } $.
\end{definition}

\begin{theorem*}[Subject Reduction]\label{claim:subjectreduction}
 If\/ $ \Gamma ^{\position{ \gamma } }   \vdash  \bmttnt{M_{{\mathrm{1}}}}  \hasType  \bmttnt{A}$ and $ \bmttnt{M_{{\mathrm{1}}}}  \Rightarrow_{\beta}^{ \gamma }  \bmttnt{M_{{\mathrm{2}}}} $, then\/ $\Gamma  \vdash  \bmttnt{M_{{\mathrm{2}}}}  \hasType  \bmttnt{A}$.
\end{theorem*}

\subsection{Metatheory}
$\beta$-reduction behaves well in the following sense:

\begin{theorem*}[Strong Normalization]\label{claim:SN}
  If\/ $ \Gamma ^{\position{ \gamma } }   \vdash  \bmttnt{M}  \hasType  \bmttnt{A}$, then $\bmttnt{M}$ is strongly normalizing with respect to\/ $ \Rightarrow_{\beta}^{ \gamma } $.
\end{theorem*}

\begin{theorem*}[Confluence]\label{claim:confluence}
  If\/ $ \Gamma ^{\position{ \gamma } }   \vdash  \bmttnt{M_{{\mathrm{1}}}}  \hasType  \bmttnt{A}$, $ \bmttnt{M_{{\mathrm{1}}}}  \Rightarrow_{\beta *}^{ \gamma }  \bmttnt{M_{{\mathrm{2}}}} $ and\/ $ \bmttnt{M_{{\mathrm{1}}}}  \Rightarrow_{\beta *}^{ \gamma }  \bmttnt{M_{{\mathrm{3}}}} $, then there exists $\bmttnt{M_{{\mathrm{4}}}}$ such that $ \bmttnt{M_{{\mathrm{2}}}}  \Rightarrow_{\beta *}^{ \gamma }  \bmttnt{M_{{\mathrm{4}}}} $ and $ \bmttnt{M_{{\mathrm{3}}}}  \Rightarrow_{\beta *}^{ \gamma }  \bmttnt{M_{{\mathrm{4}}}} $.
\end{theorem*}

\begin{corollary}[Uniqueness of $\beta$-Normal Form]
  If\/ $ \Gamma ^{\position{ \gamma } }   \vdash  \bmttnt{M_{{\mathrm{1}}}}  \hasType  \bmttnt{A}$, then $\bmttnt{M_{{\mathrm{1}}}}$ has a unique normal form with respect to\/ $ \Rightarrow_{\beta}^{ \gamma } $.
\end{corollary}

In order to confirm strong normalization, it suffices to reduce it to strong normalization of simply typed lambda calculus, which is a common technique in modal lambda-calculi~\cite{book/MartiniM96, journals/jacm/Davies17}.
Finally, we confirm canonicity and the subformula property.

\begin{definition}\quitvmode\samepage
  \begin{enumerate}
    \item
      A term is said to be \emph{canonical}
      if its outermost term-former is for an introduction rule,
      and is said to be \emph{neutral} otherwise.
    \item
      A \emph{subformula} of a formula is a literal subexpression
      with some classifier possibly renamed.
  \end{enumerate}
\end{definition}

\begin{theorem*}[Canonicity]\label{claim:canonicity}
  If a term is well-typed, closed with respect to term variables, and $\beta$-normal,
  then it is canonical.
\end{theorem*}

\begin{theorem*}[Subformula Property]\label{claim:subformula-property}
  Suppose\/ $ \Gamma ^{\position{ \gamma } }   \vdash  \bmttnt{M}  \hasType  \bmttnt{A}$. If $\bmttnt{M}$ is normal with respect to\/ $ \Rightarrow_{\beta}^{ \gamma } $, then any subterm of $\bmttnt{M}$ satisfies at least one of the following:
  \begin{enumerate}
  \item\label{item:subformula-property:intro} Its type is a subformula of $\bmttnt{A}$;
  \item\label{item:subformula-property:elim} Its type is a subformula of $\bmttnt{B}$ for some $ \binder{ \bmttmv{x} }\has@{ \delta } \bmttnt{B}  \, \in \, \Gamma$.
  \end{enumerate}
\end{theorem*}

\section{Staged Operational Semantics}\label{sec:staging}
This section establishes safety of the term-assignment system for \bml as a
staged programming language. Operational semantics for staged calculi commonly
respect the staging order: redexes at earlier stages are evaluated before later
stages~\cite{journals/jacm/Davies17,journals/corr/abs-1010-3806,conf/ppdp/YuseI06}.
Following this principle, we define \emph{staged reduction}, a weak
stage-respecting reduction relation obtained by restricting \textbeta-reduction.

In our calculus, modal transitions correspond to stage boundaries. An atomic witness $\tau$
introduced by a quotation $ \mathbf{quo} (\binder{ \tau })\lbrace^{\binder{ \gamma_{{\mathrm{1}}} } \within \gamma_{{\mathrm{2}}} }  \bmttnt{M} \rbrace $ marks one such
boundary, while the witness $\bmttnt{T}$ on an unquotation $ \mathbf{unq} _{ \bmttnt{T} }\lbrace^{ \gamma } \bmttnt{M} \rbrace $
records the boundaries it escapes. In particular, an unquotation with an empty witness $ \varepsilon $ crosses no
stage boundary and therefore represents runtime evaluation.

We begin by stratifying terms by modal transition witnesses. A witness records
the stage of a term relative to the top-level stage at which evaluation takes
place: $ \varepsilon $ denotes the top-level stage, whereas a non-empty
witness $\bmttnt{T}$ denotes a future stage. We write $ \bmttnt{M} ^{ \bmttnt{T} } $ for a
term valid at stage $\bmttnt{T}$.
Values include lambda abstractions, classifier abstractions, and quotations;
the first two reflect our weak-head reduction strategy. A quotation is a value
when its body contains no subterms at the surrounding top-level stage. Its body
is therefore indexed by $ \varepsilon $, relative to the generated stage itself.
\par\smallskip
\noindent
\begin{tabular}{llcl}
  $ \bmttnt{M} ^{ \bmttnt{T} } $ & $\in  \mathbf{Term}^{ \bmttnt{T} } $ & $\Coloneqq$ & $\bmttmv{x} \mid  \lambda \binder{ \bmttmv{x} }\has@{ \gamma } \bmttnt{A} \ldotp  \bmttnt{M} ^{ \bmttnt{T} }   \mid    \bmttnt{M_{{\mathrm{1}}}} ^{ \bmttnt{T} }    \bmttnt{M_{{\mathrm{2}}}}  ^{ \bmttnt{T} }  \mid  \mathbf{quo} (\binder{ \tau })\lbrace^{\binder{ \gamma_{{\mathrm{1}}} } \within \gamma_{{\mathrm{2}}} }   \bmttnt{M} ^{ \bmttnt{T}  \bmttsym{,}  \tau }  \rbrace $ \\
  & & $\mid$ & $ \mathbf{unq} _{ \bmttnt{T_{{\mathrm{1}}}} }\lbrace^{ \gamma }  \bmttnt{M} ^{ \bmttnt{T_{{\mathrm{2}}}} }  \rbrace  \text{\ (if\ $ \bmttnt{T_{{\mathrm{2}}}}  +  \bmttnt{T_{{\mathrm{1}}}}  = \bmttnt{T}$)} \mid  \lambda \binder{ \gamma_{{\mathrm{1}}} }\within  \gamma_{{\mathrm{2}}} .  \bmttnt{M} ^{ \bmttnt{T} }   \mid   \bmttnt{M} ^{ \bmttnt{T} }  \gamma $ \\
  $v$ & $\in  \mathbf{Val} $ & $\Coloneqq$ & $ \lambda \binder{ \bmttmv{x} }\has@{ \gamma } \bmttnt{A} \ldotp  \bmttnt{M} ^{  \varepsilon  }   \mid  \mathbf{quo} (\binder{ \tau })\lbrace^{\binder{ \gamma_{{\mathrm{1}}} } \within \gamma_{{\mathrm{2}}} }   \bmttnt{M} ^{  \varepsilon  }  \rbrace  \mid  \lambda \binder{ \gamma_{{\mathrm{1}}} }\within  \gamma_{{\mathrm{2}}} .  \bmttnt{M} ^{  \varepsilon  }  $
\end{tabular}
\par\smallskip
\noindent{}Then, we define evaluation contexts. $ \mathcal{E} ^{ \bmttnt{T_{{\mathrm{1}}}} @ \gamma_{{\mathrm{1}}} }_{[  \bmttnt{T_{{\mathrm{2}}}} @ \gamma_{{\mathrm{2}}}  ]} $ is an evaluation context located at the stage $\bmttnt{T_{{\mathrm{1}}}}$ and the classifier $\gamma_{{\mathrm{1}}}$, and whose hole is located at the stage $\bmttnt{T_{{\mathrm{2}}}}$ and the classifier $\gamma_{{\mathrm{2}}}$.
\par\smallskip
\noindent
\begin{tabular}{lcl}
  $ \mathcal{E} ^{ \bmttnt{T_{{\mathrm{1}}}} @ \gamma_{{\mathrm{1}}} }_{[  \bmttnt{T_{{\mathrm{2}}}} @ \gamma_{{\mathrm{2}}}  ]} $ & $\Coloneqq$ & $ \Box  \text{\ (if\ $\bmttnt{T_{{\mathrm{1}}}} = \bmttnt{T_{{\mathrm{2}}}}$ and $\gamma_{{\mathrm{1}}} = \gamma_{{\mathrm{2}}}$)} \mid  \lambda\binder{ \bmttmv{x} }_{ \gamma_{{\mathrm{3}}} }^{ \bmttnt{A} }.  \mathcal{E} ^{ \bmttnt{T_{{\mathrm{1}}}} @ \gamma_{{\mathrm{3}}} }_{[  \bmttnt{T_{{\mathrm{2}}}} @ \gamma_{{\mathrm{2}}}  ]}   \text{\ (if\ $\bmttnt{T_{{\mathrm{1}}}} \neq  \varepsilon $)} \mid   \mathcal{E} ^{ \bmttnt{T_{{\mathrm{1}}}} @ \gamma_{{\mathrm{1}}} }_{[  \bmttnt{T_{{\mathrm{2}}}} @ \gamma_{{\mathrm{2}}}  ]}     \bmttnt{M} ^{ \bmttnt{T_{{\mathrm{1}}}} }  $ \\
  & $\mid$ & $  \bmttnt{M} ^{ \bmttnt{T_{{\mathrm{1}}}} }      \mathcal{E} ^{ \bmttnt{T_{{\mathrm{1}}}} @ \gamma_{{\mathrm{1}}} }_{[  \bmttnt{T_{{\mathrm{2}}}} @ \gamma_{{\mathrm{2}}}  ]}    \mid  \mathbf{quo} ( \tau )\lbrace^{\binder{ \gamma_{{\mathrm{3}}} }\within{ \gamma_{{\mathrm{4}}} } }   \mathcal{E} ^{ \bmttnt{T_{{\mathrm{1}}}}  \bmttsym{,}  \tau @ \gamma_{{\mathrm{3}}} }_{[  \bmttnt{T_{{\mathrm{2}}}} @ \gamma_{{\mathrm{2}}}  ]}  \rbrace  \mid  \mathbf{unq} _{ \bmttnt{T_{{\mathrm{3}}}} }\lbrace^{ \gamma_{{\mathrm{3}}} }  \mathcal{E} ^{ \bmttnt{T_{{\mathrm{4}}}} @ \gamma_{{\mathrm{3}}} }_{[  \bmttnt{T_{{\mathrm{2}}}} @ \gamma_{{\mathrm{2}}}  ]}  \rbrace  \text{\ (if\ $ \bmttnt{T_{{\mathrm{4}}}}  +  \bmttnt{T_{{\mathrm{3}}}}  = \bmttnt{T_{{\mathrm{1}}}}$)}$ \\
  & $\mid$ & $ \lambda \binder{ \gamma_{{\mathrm{3}}} }\within  \gamma_{{\mathrm{4}}} .  \mathcal{E} ^{ \bmttnt{T_{{\mathrm{1}}}} @ \gamma_{{\mathrm{1}}} }_{[  \bmttnt{T_{{\mathrm{2}}}} @ \gamma_{{\mathrm{2}}}  ]}   \text{\ (if\ $\bmttnt{T_{{\mathrm{1}}}} \neq  \varepsilon $)} \mid   \mathcal{E} ^{ \bmttnt{T_{{\mathrm{1}}}} @ \gamma_{{\mathrm{1}}} }_{[  \bmttnt{T_{{\mathrm{2}}}} @ \gamma_{{\mathrm{2}}}  ]}  [  \gamma_{{\mathrm{3}}}  ] $
\end{tabular}
\par\smallskip
\noindent{}Now we define staged reduction $ \bmttnt{M_{{\mathrm{1}}}}  \Rightarrow_{st}^{ \gamma }  \bmttnt{M_{{\mathrm{2}}}} $, restricting redexes to the top-level stage $ \varepsilon $.
\begin{align*}
    \mathcal{E} ^{ \bmttnt{T} @ \gamma_{{\mathrm{1}}} }_{[   \varepsilon  @  \mathord{\boldsymbol{!} }   ]}   \bmttsym{[}    \bmttsym{(}   \lambda \binder{ \bmttmv{x} }\has@{ \gamma_{{\mathrm{2}}} } \bmttnt{A} \ldotp  \bmttnt{M_{{\mathrm{1}}}} ^{  \varepsilon  }    \bmttsym{)}   \bmttnt{M_{{\mathrm{2}}}}  ^{  \varepsilon  }   \bmttsym{]} &  \Rightarrow_{st}^{ \gamma_{{\mathrm{1}}} }   \mathcal{E} ^{ \bmttnt{T} @ \gamma_{{\mathrm{1}}} }_{[   \varepsilon  @  \mathord{\boldsymbol{!} }   ]}   \bmttsym{[}    \bmttnt{M_{{\mathrm{1}}}} ^{  \varepsilon  }  [ \binder{ \gamma_{{\mathrm{2}}} } \coloneqq  \mathord{\boldsymbol{!} }  , \binder{ \bmttmv{x} }  \coloneqq    \bmttnt{M_{{\mathrm{2}}}} ^{  \varepsilon  }   ]   \bmttsym{]}  \\
    \mathcal{E} ^{ \bmttnt{T_{{\mathrm{1}}}} @ \gamma_{{\mathrm{1}}} }_{[  \bmttnt{T_{{\mathrm{2}}}} @ \gamma_{{\mathrm{2}}}  ]}   \bmttsym{[}   \mathbf{unq} _{ \bmttnt{T_{{\mathrm{2}}}} }\lbrace^{  \mathord{\boldsymbol{!} }  }  \mathbf{quo} (\binder{ \tau })\lbrace^{\binder{ \gamma_{{\mathrm{3}}} } \within \gamma_{{\mathrm{4}}} }   \bmttnt{M_{{\mathrm{1}}}} ^{ \tau }  \rbrace  \rbrace   \bmttsym{]} &  \Rightarrow_{st}^{ \gamma_{{\mathrm{1}}} }   \mathcal{E} ^{ \bmttnt{T_{{\mathrm{1}}}} @ \gamma_{{\mathrm{1}}} }_{[  \bmttnt{T_{{\mathrm{2}}}} @ \gamma_{{\mathrm{2}}}  ]}   \bmttsym{[}    \bmttnt{M_{{\mathrm{1}}}} ^{ \tau }  [ \binder{ \gamma_{{\mathrm{3}}} } \coloneqq \gamma_{{\mathrm{2}}} , \binder{ \tau }  \coloneqq   \bmttnt{T_{{\mathrm{2}}}}  ]   \bmttsym{]}  \\
    \mathcal{E} ^{ \bmttnt{T_{{\mathrm{1}}}} @ \gamma_{{\mathrm{1}}} }_{[   \varepsilon  @  \mathord{\boldsymbol{!} }   ]}   \bmttsym{[}   \bmttsym{(}   \lambda \binder{ \gamma_{{\mathrm{2}}} }\within  \gamma_{{\mathrm{3}}} .  \bmttnt{M} ^{  \varepsilon  }    \bmttsym{)} \gamma_{{\mathrm{4}}}   \bmttsym{]} &  \Rightarrow_{st}^{ \gamma_{{\mathrm{1}}} }   \mathcal{E} ^{ \bmttnt{T_{{\mathrm{1}}}} @ \gamma_{{\mathrm{1}}} }_{[   \varepsilon  @  \mathord{\boldsymbol{!} }   ]}   \bmttsym{[}    \bmttnt{M} ^{  \varepsilon  }  [ \binder{ \gamma_{{\mathrm{2}}} } \coloneqq \gamma_{{\mathrm{4}}}  ]   \bmttsym{]} 
\end{align*}
$ \Rightarrow_{st*}^{ \gamma } $ is defined as the reflexive--transitive closure of $ \Rightarrow_{st}^{ \gamma } $.

We can confirm that staged reduction enjoys the expected safety properties.
First, we confirm that staged reduction is a restriction of $\beta$-reduction.
\begin{lemma*}\label{claim:staged-reduction-is-beta-reduction}
  If $ \bmttnt{M_{{\mathrm{1}}}}  \Rightarrow_{st}^{ \gamma }  \bmttnt{M_{{\mathrm{2}}}} $, then $ \bmttnt{M_{{\mathrm{1}}}}  \Rightarrow_{\beta}^{ \gamma }  \bmttnt{M_{{\mathrm{2}}}} $.
\end{lemma*}

We next establish preservation and progress, the standard safety properties for
typed lambda calculi.  In the statements below,
$ { \Gamma }^{  \varepsilon  }_{[   \varepsilon   ]} $ denotes a well-staged context that both starts and
ends at the top-level stage $ \varepsilon $; the formal definition is given in
the appendix.
\begin{theorem*}[Preservation]\label{claim:staged-reduction-preservation}
  If\/ $ { \Gamma }^{  \varepsilon   \position{ \gamma_{{\mathrm{1}}} } }_{[   \varepsilon   ]}   \vdash   \bmttnt{M_{{\mathrm{1}}}} ^{  \varepsilon  }   \hasType  \bmttnt{A}$ and $  \bmttnt{M_{{\mathrm{1}}}} ^{  \varepsilon  }   \Rightarrow_{st}^{ \gamma_{{\mathrm{1}}} }   \bmttnt{M_{{\mathrm{2}}}} ^{  \varepsilon  }  $,
  then $ { \Gamma }^{  \varepsilon   \position{ \gamma_{{\mathrm{1}}} } }_{[   \varepsilon   ]}   \vdash   \bmttnt{M_{{\mathrm{2}}}} ^{  \varepsilon  }   \hasType  \bmttnt{A}$.
\end{theorem*}

For progress, we require the context to contain no variables at
the top-level stage, as expressed by
$ \mathbf{Dom}_{\token{V} }^{  \varepsilon  }(  { \Gamma }^{  \varepsilon  }_{[   \varepsilon   ]}  )  =  \emptyset $.  The operation
$ \mathbf{Dom}_{\token{V} }^{  \varepsilon  }(  \mathord{-}  ) $ is also defined in the appendix.
\begin{theorem*}[Progress]\label{claim:staged-reduction-progress}
  If\/ $ { \Gamma }^{  \varepsilon  }_{[   \varepsilon   ]}   \vdash   \bmttnt{M_{{\mathrm{1}}}} ^{  \varepsilon  }   \hasType  \bmttnt{A}$ and\/ $ \mathbf{Dom}_{\token{V} }^{  \varepsilon  }(  { \Gamma }^{  \varepsilon  }_{[   \varepsilon   ]}  )  =  \emptyset $,
  then there exists $ \bmttnt{M_{{\mathrm{2}}}} ^{  \varepsilon  } $ such that $  \bmttnt{M_{{\mathrm{1}}}} ^{  \varepsilon  }   \Rightarrow_{st}^{  \mathord{\boldsymbol{!} }  }   \bmttnt{M_{{\mathrm{2}}}} ^{  \varepsilon  }  $,
  or $  \bmttnt{M_{{\mathrm{1}}}} ^{  \varepsilon  }  \in  \mathbf{Val}  $.
\end{theorem*}

Finally, we establish the safety of offline code generation: a well-staged term of bounded modal type evaluates to a code value whose body is well typed.\footnote{Similar properties are referred to as \emph{binding-time correctness}~\cite{journals/jacm/Davies17} and \emph{type-safe residualization}~\cite{conf/flops/HanadaI14}.}

\begin{theorem*}[Safety of Offline Code Generation]\label{claim:safety-of-offline-code-generation}
  If\/ $ {  \tau_{{\mathrm{1}}} : \mathord{\blacktriangleright} ^{\binder{ \gamma_{{\mathrm{1}}} } \within  \mathord{\boldsymbol{!} }  }   \bmttsym{,}  \Gamma }^{  \varepsilon   \position{ \gamma_{{\mathrm{2}}} } }_{[   \varepsilon   ]}   \bmttsym{,}   \mathord{\blacktriangleleft} _{ \tau_{{\mathrm{1}}} }^{  \mathord{\boldsymbol{!} }  }   \vdash   \bmttnt{M_{{\mathrm{1}}}} ^{  \varepsilon  }   \hasType   \Box^{\mathord{\succeq}  \gamma_{{\mathrm{2}}} }  \bmttnt{A} $,
  then there exists $v =  \mathbf{quo} (\binder{ \tau_{{\mathrm{2}}} })\lbrace^{\binder{ \gamma_{{\mathrm{3}}} } \within \gamma_{{\mathrm{2}}} }   \bmttnt{M_{{\mathrm{2}}}} ^{  \varepsilon  }  \rbrace $ such that
  $ \bmttnt{M_{{\mathrm{1}}}}  \Rightarrow_{st*}^{  \mathord{\boldsymbol{!} }  }  v $ and\/
  $  { \Gamma }^{  \varepsilon   \position{ \gamma_{{\mathrm{2}}} } }_{[   \varepsilon   ]}  [ \binder{ \gamma_{{\mathrm{1}}} } \coloneqq  \mathord{\boldsymbol{!} }   ]   \vdash     \bmttnt{M_{{\mathrm{2}}}} ^{  \varepsilon  }  [ \binder{ \gamma_{{\mathrm{3}}} } \coloneqq \gamma_{{\mathrm{2}}}  ]  [ \binder{ \gamma_{{\mathrm{1}}} } \coloneqq  \mathord{\boldsymbol{!} }   ]   \hasType   \bmttnt{A} [ \binder{ \gamma_{{\mathrm{1}}} } \coloneqq  \mathord{\boldsymbol{!} }   ] $.
\end{theorem*}

\section{Relation to S4- and LTL-based Calculi}\label{sec:s4andltl}
To show that \bml retains the established modal accounts of staging while providing a foundation for additional staging capabilities, we relate it to the standard S4- and LTL-based calculi. For space reasons, we state only the main results here and defer the detailed constructions and proofs to the appendix.

The following theorem shows that \bml recovers S4 modal logic when the modality is restricted to $\Box^{ \exclam }$.
\begin{theorem}\label{claim:semantic-isomorphism}
The $\Box\to$-fragment of $\CS4$~\cite{alechina+2001categorical} is isomorphic, modulo logical equivalence, to the $\Box^{ \exclam }\to$-fragment of \bml.
\end{theorem}

At the computational level, the following theorem shows that the standard S4- and LTL-based staging calculi admit type-preserving syntactic embeddings into the term-assignment system of \bml.
Thus, \bml retains their respective staging cores within a single calculus.
\begin{theorem}\label{claim:embed-s4-and-ltl}
The Kripke-style S4 modal lambda calculus~\cite{journals/jacm/DaviesP01} and \lamcirc~\cite{journals/jacm/Davies17} each admit a type-preserving syntactic embedding into the term-assignment system of\/ \bml.
\end{theorem}

\section{Related Work}\label{sec:related}
\paragraph*{Modal-Logical Foundations for Multi-Stage Programming}
Since S4- and LTL-based calculi account for different fragments of
MetaML-style MSP, earlier work combined S4- and LTL-like modalities in a single
calculus~\cite{conf/esop/MoggiTBS99,conf/ppdp/YuseI06}. However, the two
modalities remain distinct and are not freely interchangeable. This makes the
distinction too coarse for code that is both scope-dependent and executable by
\texttt{run}.

\Textcite{conf/popl/TahaN03} introduced \lamalpha{} with
\emph{environment classifiers}, giving a finer-grained type-theoretic account
of MetaML-style MSP with CSP. Its modal-logical status was later investigated
by \Textcite{journals/corr/abs-1010-3806}, who proposed \lamtri{}, a modal
lambda calculus that generalizes \lamcirc{} by abstracting over stage
transitions. However, the logically founded core of \lamtri{} does not account
for MetaML-style CSP, leaving the modal foundation for CSP-dependent code
executed by \texttt{run} incomplete.

\bml is formulated as a generalization of Kripke/Fitch-style S4 modal lambda calculi~\cite{journals/jacm/DaviesP01,conf/fossacs/Clouston18, journals/pacmpl/ValliappanRC22}. Its use of classifiers is inspired by \emph{refined environment classifiers}~\cite{conf/aplas/KiselyovKS16}, which track scope dependencies of generated code. However, their formal calculus does not support MetaML-style CSP or classifier polymorphism. \bml goes beyond this system by internalizing scope dependencies
in a modal connective and integrating polymorphic classifiers, thereby accounting for CSP-dependent and scope-polymorphic code within a proof theory and Kripke semantics. The embeddings in \Zcref{sec:s4andltl} further show that \bml retains the standard S4 and LTL staging cores.

\emph{Contextual Modal Logic} ($\CML$)
~\cite{journals/tocl/NanevskiPP08,conf/esop/MuraseNI23} extends $\CS4$
with explicit contexts. The contextual modality $[\Gamma \vdash A]$
represents code valid under assumptions $\Gamma$, and thus gives another
account of open code. However, its binding discipline differs from
MetaML-style MSP: dependencies are represented by explicit contexts and
substitutions, rather than by direct binding from
surrounding lexical binders. Thus, $\CML$ does not directly account for
CSP-dependent programs such as \Zcref{fig:metaml-example-csp}; this reflects a
difference in design goals rather than a limitation.

\paragraph*{Hybridizing Modal and First-Order Languages}
\bml is not the first attempt to hybridize modal and first-order languages.
Hybrid logic extends modal logic with \emph{nominals}, which name individual
worlds, and satisfaction operators that make it possible to assert that a
formula holds at a named world~\cite{book/Brauner2011}. This gives modal
formulas a limited first-order flavor: worlds can be referred to explicitly
inside the object language while retaining the modal structure of accessibility
relations.

The classifiers of \bml play a role similar to nominals, since they name
elements of the underlying Kripke structure. However, the technical details and
design goals are rather different. Hybrid logic is primarily concerned with
reasoning about named worlds inside modal formulas, whereas \bml uses
classifiers to describe the scope-and-stage structure of programs. Moreover,
\bml integrates classifiers directly into its modal connective: $ \Box^{\mathord{\succeq}  \gamma }  \bmttnt{A} $
states that $\bmttnt{A}$ is available across a modal transition relative to the
scope named by $\gamma$. In hybrid logic, by contrast, the modal operators and
satisfaction statements remain separate pieces of syntax. Finally, nominals in
hybrid logic are themselves propositions, while classifiers in \bml are not
formulas; they are first-order objects used to index scopes and modalities.

Graded modal types also enrich modalities with indices, but their grades track
quantitative information such as resource usage~\cite{journals/pacmpl/OrchardLE19}.
Classifiers in \bml instead name scopes in the underlying Kripke structure, and
therefore serve as logical handles on lexical resources across stage
boundaries.

\section{Conclusion and Future Work}\label{sec:conclusion}
We proposed \bml, a constructive modal logic with scope-bounded modalities
and first-order-style quantification over scope names. We gave a
natural-deduction system and a sound and complete Kripke semantics, together
with a Curry--Howard calculus with standard metatheory and a safe staged
reduction semantics. \bml recovers the staging capabilities characterized by
the standard S4- and LTL-based accounts, while also accounting for explicit scope dependencies, as manifested by CSP. The
significance of this result is not merely that \bml accommodates a broader
class of staged programs. Rather, it demonstrates that the simple idea
underlying bounded modalities---making scope dependencies explicit in the
modality itself---provides an effective logical principle for unifying these
capabilities. These results establish \bml as a modal-logical foundation for
MetaML-style staging.

Beyond staging, modal lambda calculi have been used to study computational
effects~\cite{journals/iandc/Moggi91}, distributed
programming~\cite{conf/lics/VIICHP04}, functional reactive
programming~\cite{conf/icfp/Krishnaswami13}, and information-flow
control~\cite{conf/popl/AbadiBHR99}. Whether scope-bounded modalities can
offer similar benefits in these areas remains future work.

\bibliography{bibliography}

\clearpage
\appendix
\section{Full Definitions for Section~\ref{sec:logic:nd}}
\begin{definition}
  $ \token{FC}( \bmttnt{A} ) $ represents a set of free classifiers in $\bmttnt{A}$.
  \begin{align*}
     \token{FC}( \bmttnt{p} )  &=  \emptyset  \\
     \token{FC}( \bmttnt{A}  \mathbin{\rightarrow}  \bmttnt{B} )  &=   \token{FC}( \bmttnt{A} )  \cup  \token{FC}( \bmttnt{B} )   \\
     \token{FC}(  \Box^{\mathord{\succeq}  \gamma }  \bmttnt{A}  )  &=   \token{FC}( \bmttnt{A} )  \cup \bmttsym{\{}  \gamma  \bmttsym{\}} \\
     \token{FC}(  \forall  \gamma_{{\mathrm{1}}} \within  \gamma_{{\mathrm{2}}} . \bmttnt{A}  )  &=  \bmttsym{(}    \token{FC}( \bmttnt{A} )  - \bmttsym{\{}  \gamma_{{\mathrm{1}}}  \bmttsym{\}}   \bmttsym{)} \cup \bmttsym{\{}  \gamma_{{\mathrm{2}}}  \bmttsym{\}} 
  \end{align*}
\end{definition}

\begin{definition}
 $ \mathbf{Dom}_{\token{C} }( \Gamma ) $ represents a set of classifiers declared in $\Gamma$.
 \begin{align*}
    \mathbf{Dom}_{\token{C} }(  \varepsilon  )  &= \bmttsym{\{}   \exclam   \bmttsym{\}} \\
    \mathbf{Dom}_{\token{C} }( \Gamma  \bmttsym{,}   { \bmttnt{A} }^{\binder{ \gamma } }  )  &=   \mathbf{Dom}_{\token{C} }( \Gamma )  \cup \bmttsym{\{}  \gamma  \bmttsym{\}}  \\
    \mathbf{Dom}_{\token{C} }( \Gamma  \bmttsym{,}   \mathord{\blacktriangleright} ^{\binder{ \gamma_{{\mathrm{1}}} } \within \gamma_{{\mathrm{2}}} }  )  &=   \mathbf{Dom}_{\token{C} }( \Gamma )  \cup \bmttsym{\{}  \gamma_{{\mathrm{1}}}  \bmttsym{\}}  \\
    \mathbf{Dom}_{\token{C} }( \Gamma  \bmttsym{,}   \mathord{\blacktriangleleft} ^{ \gamma }  )  &=  \mathbf{Dom}_{\token{C} }( \Gamma )  \\
    \mathbf{Dom}_{\token{C} }( \Gamma  \bmttsym{,}   \binder{ \gamma_{{\mathrm{1}}} } \within \gamma_{{\mathrm{2}}}  )  &=   \mathbf{Dom}_{\token{C} }( \Gamma )  \cup \bmttsym{\{}  \gamma_{{\mathrm{1}}}  \bmttsym{\}} 
 \end{align*}
\end{definition}

\begin{definition}
 A well-formed context judgment $\vdash  \Gamma  \hasType \, \bmttkw{ctx}$ and well-formed type judgment $\Gamma  \vdash  \bmttnt{A}  \hasType \, \bmttkw{prop}$ are derived by rules listed in~\Zcref{fig:wellformedness-judgments}.
\end{definition}

\begin{figure}[bpt]
  \raggedright
  \fbox{\mbox{$\vdash  \Gamma  \hasType \, \bmttkw{ctx}$}}
  \vspace{0.5em}

  \centering
  \begin{prooftree}
    \caption{WF-Emp}\label{rule:wf-ctx-emp}
    \hypo{\rule{0pt}{1em}}
    \infer1{\vdash   \varepsilon   \hasType \, \bmttkw{ctx}}
  \end{prooftree}
  \quad
  \begin{prooftree}
    \caption{WF-Hyp}\label{rule:wf-ctx-hyp}
    \hypo{\vdash  \Gamma  \hasType \, \bmttkw{ctx}}
    \hypo{\Gamma  \vdash  \bmttnt{A}  \hasType \, \bmttkw{prop}}
    \hypo{\gamma \, \notin \,  \mathbf{Dom}_{\token{C} }( \Gamma ) }
    \infer3{\vdash  \Gamma  \bmttsym{,}   { \bmttnt{A} }^{\binder{ \gamma } }   \hasType \, \bmttkw{ctx}}
  \end{prooftree}
  \vspace{1em}

  \begin{prooftree}
    \caption{WF-$ \mathord{\blacktriangleright} $}\label{rule:wf-ctx-open}
    \hypo{\vdash  \Gamma  \hasType \, \bmttkw{ctx}}
    \hypo{\gamma_{{\mathrm{1}}} \, \notin \,  \mathbf{Dom}_{\token{C} }( \Gamma ) }
    \hypo{\gamma_{{\mathrm{2}}} \, \in \,  \mathbf{Dom}_{\token{C} }( \Gamma ) }
    \infer3{\vdash  \Gamma  \bmttsym{,}   \mathord{\blacktriangleright} ^{\binder{ \gamma_{{\mathrm{1}}} } \within \gamma_{{\mathrm{2}}} }   \hasType \, \bmttkw{ctx}}
  \end{prooftree}
  \quad
  \begin{prooftree}
    \caption{WF-$ \mathord{\blacktriangleleft} $}
    \hypo{\vdash   \Gamma ^{\position{ \gamma } }   \hasType \, \bmttkw{ctx}}
    \hypo{  \Gamma ^{\position{ \gamma } }  \vdash \delta \sqsubseteq \gamma }
    \infer2{\vdash   \Gamma ^{\position{ \gamma } }   \bmttsym{,}   \mathord{\blacktriangleleft} ^{ \delta }   \hasType \, \bmttkw{ctx}}
  \end{prooftree}
  \vspace{1em}

  \begin{prooftree}
    \caption{WF-cls}\label{rule:wf-ctx-cls}
    \hypo{\vdash  \Gamma  \hasType \, \bmttkw{ctx}}
    \hypo{\gamma_{{\mathrm{1}}} \, \notin \,  \mathbf{Dom}_{\token{C} }( \Gamma ) }
    \hypo{\gamma_{{\mathrm{2}}} \, \in \,  \mathbf{Dom}_{\token{C} }( \Gamma ) }
    \infer3{\vdash  \Gamma  \bmttsym{,}   \binder{ \gamma_{{\mathrm{1}}} } \within \gamma_{{\mathrm{2}}}   \hasType \, \bmttkw{ctx}}
  \end{prooftree}
  \vspace{1.5em}

  \raggedright
  \fbox{\mbox{$\Gamma  \vdash  \bmttnt{A}  \hasType \, \bmttkw{prop}$}}
  \vspace{0.5em}

  \centering
  \begin{prooftree}
    \caption{WF-atom}\label{rule:wf-prop-atom}
    \hypo{\vdash  \Gamma  \hasType \, \bmttkw{ctx}}
    \infer1{\Gamma  \vdash  \bmttnt{p}  \hasType \, \bmttkw{prop}}
  \end{prooftree}
  \quad
  \begin{prooftree}
    \caption{WF-$ \mathbin{\rightarrow} $}\label{rule:wf-prop-to}
    \hypo{\Gamma  \vdash  \bmttnt{A}  \hasType \, \bmttkw{prop}}
    \hypo{\Gamma  \vdash  \bmttnt{B}  \hasType \, \bmttkw{prop}}
    \infer2{\Gamma  \vdash  \bmttnt{A}  \mathbin{\rightarrow}  \bmttnt{B}  \hasType \, \bmttkw{prop}}
  \end{prooftree}
  \vspace{1em}

  \begin{prooftree}
    \caption{WF-$\Box$}\label{rule:wf-prop-bm}
    \hypo{\Gamma  \vdash  \bmttnt{A}  \hasType \, \bmttkw{prop}}
    \hypo{\gamma \, \in \,  \mathbf{Dom}_{\token{C} }( \Gamma ) }
    \infer2{\Gamma  \vdash   \Box^{\mathord{\succeq}  \gamma }  \bmttnt{A}   \hasType \, \bmttkw{prop}}
  \end{prooftree}
  \quad
  \begin{prooftree}
    \caption{WF-$\forall$}\label{rule:wf-prop-polycls}
    \hypo{\Gamma  \bmttsym{,}   \binder{ \gamma_{{\mathrm{1}}} } \within \gamma_{{\mathrm{2}}}   \vdash  \bmttnt{A}  \hasType \, \bmttkw{prop}}
    \infer1{\Gamma  \vdash   \forall  \gamma_{{\mathrm{1}}} \within  \gamma_{{\mathrm{2}}} . \bmttnt{A}   \hasType \, \bmttkw{prop}}
  \end{prooftree}

  \caption{Derivation Rules for Well-Formedness Judgments}
  \label{fig:wellformedness-judgments}
\end{figure}

\begin{definition}
 A classifier substitution $ \bmttsym{(}   \mathord{-}   \bmttsym{)} [ \binder{ \gamma_{{\mathrm{1}}} } \coloneqq \gamma_{{\mathrm{2}}}  ] $ is a meta operation on classifiers, contexts and types, which replaces free occurrences of $\gamma_{{\mathrm{1}}}$ with $\gamma_{{\mathrm{2}}}$.

 \begin{align*}
   \gamma_{{\mathrm{1}}} [ \binder{ \gamma_{{\mathrm{2}}} } \coloneqq \gamma_{{\mathrm{3}}}  ]  &= \begin{cases}
                      \gamma_{{\mathrm{3}}} & \text{\ if $\gamma_{{\mathrm{1}}}  \bmttsym{=}  \gamma_{{\mathrm{2}}}$} \\
                      \gamma_{{\mathrm{1}}} & \text{\ otherwise}
                     \end{cases} \\
  & \\
   \bmttnt{p} [ \binder{ \gamma_{{\mathrm{1}}} } \coloneqq \gamma_{{\mathrm{2}}}  ]  &= \bmttnt{p} \\
   \bmttsym{(}  \bmttnt{A}  \mathbin{\rightarrow}  \bmttnt{B}  \bmttsym{)} [ \binder{ \gamma_{{\mathrm{1}}} } \coloneqq \gamma_{{\mathrm{2}}}  ]  &=  \bmttnt{A} [ \binder{ \gamma_{{\mathrm{1}}} } \coloneqq \gamma_{{\mathrm{2}}}  ]   \mathbin{\rightarrow}    \bmttnt{B} [ \binder{ \gamma_{{\mathrm{1}}} } \coloneqq \gamma_{{\mathrm{2}}}  ]  \\
    \Box^{\mathord{\succeq}  \gamma_{{\mathrm{1}}} }  \bmttnt{A}  [ \binder{ \gamma_{{\mathrm{2}}} } \coloneqq \gamma_{{\mathrm{3}}}  ]  &=  \Box^{\mathord{\succeq}   \gamma_{{\mathrm{1}}} [ \binder{ \gamma_{{\mathrm{2}}} } \coloneqq \gamma_{{\mathrm{3}}}  ]  }   \bmttnt{A} [ \binder{ \gamma_{{\mathrm{2}}} } \coloneqq \gamma_{{\mathrm{3}}}  ]   \\
   \bmttsym{(}   \forall  \gamma_{{\mathrm{1}}} \within  \gamma_{{\mathrm{2}}} . \bmttnt{A}   \bmttsym{)} [ \binder{ \gamma_{{\mathrm{3}}} } \coloneqq \gamma_{{\mathrm{4}}}  ]  &=  \forall  \gamma_{{\mathrm{1}}} \within   \gamma_{{\mathrm{2}}} [ \binder{ \gamma_{{\mathrm{3}}} } \coloneqq \gamma_{{\mathrm{4}}}  ]  .  \bmttnt{A} [ \binder{ \gamma_{{\mathrm{3}}} } \coloneqq \gamma_{{\mathrm{4}}}  ]  \\
  & \qquad \text{where $\gamma_{{\mathrm{1}}} \, \notin \, \bmttsym{\{}  \gamma_{{\mathrm{3}}}  \bmttsym{,}  \gamma_{{\mathrm{4}}}  \bmttsym{\}}$}\\
  & \\
    \varepsilon  [ \binder{ \gamma_{{\mathrm{1}}} } \coloneqq \gamma_{{\mathrm{2}}}  ]  &=  \varepsilon  \\
   \bmttsym{(}   { \bmttnt{A} }^{\binder{ \gamma_{{\mathrm{1}}} } }   \bmttsym{,}  \Gamma  \bmttsym{)} [ \binder{ \gamma_{{\mathrm{2}}} } \coloneqq \gamma_{{\mathrm{3}}}  ]  &=  {  \bmttnt{A} [ \binder{ \gamma_{{\mathrm{2}}} } \coloneqq \gamma_{{\mathrm{3}}}  ]  }^{\binder{ \gamma_{{\mathrm{1}}} } }   \bmttsym{,}   \Gamma [ \binder{ \gamma_{{\mathrm{2}}} } \coloneqq \gamma_{{\mathrm{3}}}  ] \\
  & \qquad \text{where $\gamma_{{\mathrm{1}}} \, \notin \, \bmttsym{\{}  \gamma_{{\mathrm{2}}}  \bmttsym{,}  \gamma_{{\mathrm{3}}}  \bmttsym{\}}$}\\
   \bmttsym{(}   \mathord{\blacktriangleright} ^{\binder{ \gamma_{{\mathrm{1}}} } \within \gamma_{{\mathrm{2}}} }   \bmttsym{,}  \Gamma  \bmttsym{)} [ \binder{ \gamma_{{\mathrm{3}}} } \coloneqq \gamma_{{\mathrm{4}}}  ]  &=  \mathord{\blacktriangleright} ^{\binder{ \gamma_{{\mathrm{1}}} } \within   \gamma_{{\mathrm{2}}} [ \binder{ \gamma_{{\mathrm{3}}} } \coloneqq \gamma_{{\mathrm{4}}}  ]   }   \bmttsym{,}    \Gamma [ \binder{ \gamma_{{\mathrm{3}}} } \coloneqq \gamma_{{\mathrm{4}}}  ]  \\
  & \qquad \text{where $\gamma_{{\mathrm{1}}} \, \notin \, \bmttsym{\{}  \gamma_{{\mathrm{3}}}  \bmttsym{,}  \gamma_{{\mathrm{4}}}  \bmttsym{\}}$}\\
   \bmttsym{(}   \mathord{\blacktriangleleft} ^{ \gamma_{{\mathrm{1}}} }   \bmttsym{,}  \Gamma  \bmttsym{)} [ \binder{ \gamma_{{\mathrm{2}}} } \coloneqq \gamma_{{\mathrm{3}}}  ]  &=  \mathord{\blacktriangleleft} ^{   \gamma_{{\mathrm{1}}} [ \binder{ \gamma_{{\mathrm{2}}} } \coloneqq \gamma_{{\mathrm{3}}}  ]   }   \bmttsym{,}   \Gamma [ \binder{ \gamma_{{\mathrm{2}}} } \coloneqq \gamma_{{\mathrm{3}}}  ] \\
   \bmttsym{(}   \binder{ \gamma_{{\mathrm{1}}} } \within \gamma_{{\mathrm{2}}}   \bmttsym{,}  \Gamma  \bmttsym{)} [ \binder{ \gamma_{{\mathrm{3}}} } \coloneqq \gamma_{{\mathrm{4}}}  ]  &=  \binder{ \gamma_{{\mathrm{1}}} } \within   \gamma_{{\mathrm{2}}} [ \binder{ \gamma_{{\mathrm{3}}} } \coloneqq \gamma_{{\mathrm{4}}}  ]     \bmttsym{,}   \Gamma [ \binder{ \gamma_{{\mathrm{3}}} } \coloneqq \gamma_{{\mathrm{4}}}  ] \\
  & \qquad \text{where $\gamma_{{\mathrm{1}}} \, \notin \, \bmttsym{\{}  \gamma_{{\mathrm{3}}}  \bmttsym{,}  \gamma_{{\mathrm{4}}}  \bmttsym{\}}$}\\
 \end{align*}
\end{definition}

\section{Proofs for Section~\ref{sec:logic:semantics}}
\zref[claim]{claim:semantic-monotonicity}

\begin{proof}
  It suffices to check that for all $d' ⪰_w d$,
  if $w, d ⊨^{ρ} A$, then $w, d' ⊨^{ρ} A$.
  We proceed by induction on $A$.
  \begin{case}[$A \equiv p$]
    Follows from the fact that $V_w(p)$ is upward-closed.
  \end{case}
  \begin{case}[$A \equiv B \to C$]
    By definition.
  \end{case}
  \begin{case}[$A \equiv □^{⪰ γ} B$]
    Suppose $w, d ⊨^{ρ} □^{⪰ γ} B$ and $d ⪯_w d'$.
    Take $e ⊒_w d'$ such that $e ⪰_w ρ(γ)$.
    By left-stability we have $d ⊑_w e$, so that
    $w, e ⊩^{ρ} B$ since $w, d ⊨^{ρ} □^{⪰ γ} B$, which
    implies $w, d' ⊨^{ρ} □^{⪰ γ} B$.
  \end{case}
  \begin{case}[$A \equiv ∀\cls(γ_2 i> γ_1). B$]
    Immediate from the IH\@.\qedhere
  \end{case}
\end{proof}

\section{Proofs for Section~\ref{sec:logic:completeness}}
\begin{lemma}\label{claim:semantic-monotonicity:transitions}
  Suppose $w ≼ v$.
  \begin{enumerate}
  \item\label{item:semantic-monotonicity:scopes}
    If $w ⊩^{ρ} γ_1 ⪯ γ_2$,
    then $v ⊩^{ρ} γ_1 ⪯ γ_2$.
  \item\label{item:semantic-monotonicity:stages}
    If $w ⊩^{ρ} γ_1 ⊑ γ_2$,
    then $v ⊩^{ρ} γ_1 ⊑ γ_2$.
  \end{enumerate}
\end{lemma}

\begin{proof}
  Straightforward.
\end{proof}

\begin{corollary}\label{claim:semantic-monotonicity:contexts}
  If $w ≼ v$ and $w ⊩^{ρ} \Gamma$,
  then $v ⊩^{ρ} \Gamma$.
\end{corollary}

\begin{lemma}[Substitution]\label{claim:semantic-substitution}
  \begin{math}
    w, d ⊩^{ρ ⋅ [γ_2 ↦ ρ(γ_1)] } A
    \iff
    w, d ⊩^{ρ} A\w[γ_2 / γ_1]
  \end{math}.
\end{lemma}

\begin{proof}
  By induction on $A$.
\end{proof}

\zref[claim]{claim:Kripke-soundness}

\begin{proof}
  \zcref{item:Kripke-soundness:scopes,item:Kripke-soundness:stages} are
  mostly straightforward, so
  we show \zcref{item:Kripke-soundness:types} here.
  We proceed by induction on derivation.

  Suppose $Γ ⊢ A$.
  Assuming $w ⊩^{ρ} Γ$,
  we show $w, ρ(\pos(Γ)) ⊩^{ρ} A$.
  We analyze the last rule of the derivation:
  \begin{case}[\ref{rule:derive-hyp}]
    Assume
    \begin{prooftree*}
      \hypo{A^{\binder{γ}} \in Γ}
      \hypo{Γ ⊢ γ ⪯ \pos(Γ)}
      \infer2{Γ ⊢ A}
    \end{prooftree*}
    By assumption we have $w, ρ(γ) ⊩^{ρ} A$,
    and from~\zcref{item:Kripke-soundness:scopes},
    $ρ(γ) ⪯_w ρ(\pos(Γ))$ holds.
    By \Zcref{claim:semantic-monotonicity}
    we see $w, ρ(\pos(Γ)) ⊩^{ρ} A$.
  \end{case}
  \begin{case}[\ref{rule:derive-to-i}]
    Assume
    \begin{prooftree*}
      \hypo{Γ, B^{\binder{γ}} ⊢ C}
      \hypo{γ \notin \FC(C)}
      \infer2{Γ ⊢ B \to C}
    \end{prooftree*}
    Take $v ≽ w$ and $d ⪰_v ρ(\pos(Γ))$, and
    suppose $v, d ⊩^{ρ} B$.
    By \Zcref{claim:semantic-monotonicity:contexts}
    we have $v ⊩^{ρ} Γ$, so
    letting $ρ' = ρ ⋅ [γ ↦ d]$
    we obtain $v ⊩^{ρ'} Γ, B^{\binder{γ}}$.
    By the IH, $v, d ⊩^{ρ'} C$ holds, and so does
    $v, d ⊩^{ρ} C$ as $γ \notin \FC(C)$, which
    implies $w, ρ(\pos(Γ)) ⊩^{ρ} B \to C$.
  \end{case}
  \begin{case}[\ref{rule:derive-to-e}]
    Assume
    \begin{prooftree*}
      \hypo{Γ ⊢ B \to C}
      \hypo{Γ ⊢ B}
      \infer2{Γ ⊢ C}
    \end{prooftree*}
    By the IH, we have $w, ρ(\pos(Γ)) ⊩^{ρ} B \to C$ and
    $w, ρ(\pos(Γ)) ⊩^{ρ} B$.
    Since $w ≼ w$ and $ρ(\pos(Γ)) ⪯_w ρ(\pos(Γ))$,
    we obtain $w, ρ(\pos(Γ)) ⊩^{ρ} C$.
  \end{case}
  \begin{case}[\ref{rule:derive-bm-i}]
    Assume
    \begin{prooftree*}
      \hypo{Γ, \open(γ_2 i> γ_1) ⊢ B}
      \hypo{γ_2 \notin \FC(B)}
      \infer2{Γ ⊢ □^{⪰ γ_1} B}
    \end{prooftree*}
    Take $v ≽ w$ and $d ⊒_v ρ(\pos(Γ))$ with $d ≽_v ρ(γ_1)$.
    By \Zcref{claim:semantic-monotonicity:contexts}
    we have $v ⊩^{ρ} Γ$, so
    letting $ρ' = ρ ⋅ [γ_2 ↦ d]$
    we obtain $v ⊩^{ρ'} Γ, \open(γ_2 i> γ_1)$.
    By the IH, $v, d ⊩^{ρ'} B$ holds, and so does
    $v, d ⊩^{ρ} B$ as $γ_2 \notin \FC(B)$, which
    implies $w, ρ(\pos(Γ)) ⊩^{ρ} □^{⪰ γ_1} B$.
  \end{case}
  \begin{case}[\ref{rule:derive-bm-e}]
    Assume
    \begin{prooftree*}
      \hypo{Γ, ^{γ} ⊢ □^{⪰ γ_1} B}
      \hypo{Γ ⊢ γ_1 ⪯ \pos(Γ)}
      \infer2{Γ ⊢ B}
    \end{prooftree*}
    From $w ⊩^{ρ} Γ$,
    we have $w ⊩^{ρ} Γ, ^{γ}$, and
    by \zcref{item:Kripke-soundness:scopes},
    also $ρ(γ_1) ⪯_w ρ(\pos(Γ))$.
    Since $Γ, ^{γ}$ is well-formed,
    there should be a subderivation of $Γ ⊢ γ ⊑ \pos(Γ)$,
    which gives $ρ(γ) ⊑_w ρ(\pos(Γ))$
    by~\zcref{item:Kripke-soundness:stages}.
    Applying the IH to $Γ, ^{γ} ⊢ □^{⪰ γ_1} B$,
    we have $w, ρ(γ) ⊩^{ρ} □^{⪰ γ_1} B$, and
    from $w ≼ w$, we see $w, ρ(\pos(Γ)) ⊩^{ρ} B$.
  \end{case}
  \begin{case}[\ref{rule:derive-polycls-i}]
    Assume
    \begin{prooftree*}
      \hypo{Γ, \cls(γ_2 i> γ_1) ⊢ B}
      \infer1{Γ ⊢ ∀\cls(γ_2 i> γ_1). B}
    \end{prooftree*}
    Take $v ≽ w$ and $d ⪰_v ρ(γ_1)$.
    By \Zcref{claim:semantic-monotonicity:contexts} we have
    $v ⊩^{ρ} Γ$, so letting
    $ρ' = ρ ⋅ [γ_2 ↦ d]$ we obtain
    $v ⊩^{ρ'} Γ, \cls(γ_2 i> γ_1)$.
    By the IH, $v, ρ(\pos(Γ)) ⊩^{ρ'} B$ holds, which
    implies $w, ρ(\pos(Γ)) ⊩^{ρ} ∀\cls(γ_2 i> γ_1). B$.
  \end{case}
  \begin{case}[\ref{rule:derive-polycls-e}]
    Assume
    \begin{prooftree*}
      \hypo{Γ ⊢ ∀\cls(γ_2 i> γ_1). B}
      \hypo{Γ ⊢ γ_1 ⪯ γ}
      \infer2{Γ ⊢ B\w[γ_2 / γ]}
    \end{prooftree*}
    By the IH,
    we have $w, ρ(\pos(Γ)) ⊩^{ρ} ∀\cls(γ_2 i> γ_1). B$, and
    from~\zcref{item:Kripke-soundness:scopes}, $ρ(γ_1) ⪯_w ρ(γ)$ holds.
    Since $w ≼ w$,
    we obtain $w, ρ(\pos(Γ)) ⊩^{ρ ⋅ [γ_2 ↦ ρ(γ)]} B$, and
    \Zcref{claim:semantic-substitution} yields
    $w, ρ(\pos(Γ)) ⊩^{ρ} B\w[γ_2 / γ]$.
    \qedhere
  \end{case}
\end{proof}

\begin{lemma}\label{claim:admissible-rules-for-completeness}
  The following rules are admissible:
  \begin{enumerate}
    \item Inversion of\/ \zcref[noname]{rule:type-bm-i}:
      \begin{prooftree*}
        \hypo{Γ ⊢ □^{⪰ γ} A}
        \infer1{Γ, \open(γ' i> γ) ⊢ A}
      \end{prooftree*}
    \item General weakening for top-level subderivations:
      \begin{mathpar}
        \begin{prooftree}
          \hypo{Γ, ^{\tp}, Γ' ⊢ A}
          \infer1{Γ, Δ, ^{\tp}, Γ' ⊢ A}
        \end{prooftree}
        \and
        \begin{prooftree}
          \hypo{Γ, ^{\tp}, Γ' ⊢ γ_1 ⊴ γ_2}
          \infer1{Γ, Δ, ^{\tp}, Γ' ⊢ γ_1 ⊴ γ_2}
        \end{prooftree}
      \end{mathpar}
      where $\mathord{⊴} \in \{\mathord{⪯}, \mathord{⊑}\}$.
  \end{enumerate}
\end{lemma}

\begin{proof}\quitvmode\samepage
  \begin{enumerate}
    \item
      Using \Zcref{claim:monotonicity}, we have
      \begin{prooftree*}
        \hypo{Γ ⊢ □^{⪰ γ} A}
        \infer[dashed]1{Γ, \open(γ' i> γ), ^{\pos(Γ)} ⊢ □^{⪰ γ} A}
        \by1(rule:type-bm-e){Γ, \open(γ' i> γ) ⊢ A}
      \end{prooftree*}
    \item
      By induction on derivation.
      \qedhere
  \end{enumerate}
\end{proof}

\zref[claim]{claim:truth-lemma}

\begin{proof}
  By induction on the size of~$A$,
  using \Zcref{claim:admissible-rules-for-completeness}.
  \begin{case}[$A \equiv p$]
    By definition.
  \end{case}
  \begin{case}[$A \equiv B \to C$]\quitvmode\samepage
    \proofsubparagraph{($\Longleftarrow$)}

    Suppose $Γ, ^{\tp} ⊢ □^{⪰ γ} (B \to C)$.
    Take $Δ \can{≽} Γ$ and $δ \can{⪰_{Δ}} γ$
    satisfying $Δ, δ ⊩^{\can{ρ_{Γ}}} B$.
    Since $\can{ρ_{Δ}} \restricted_{\DomC(Γ)} = \can{ρ_{Γ}}$,
    we have $Δ, δ \can{⊩} B$, and
    the IH yields $Δ, ^{\tp} ⊢ □^{⪰ δ} B$.
    Then we can derive
    \begin{prooftree*}
      \hypo{Γ, ^{\tp} ⊢ □^{⪰ γ} (B \to C)}
      \infer[dashed]1{Δ, ^{\tp} ⊢ □^{⪰ γ} (B \to C)}
      \infer[dashed]1{Δ, ^{\tp}, \open(δ' i> δ) ⊢ B \to C}
      \hypo{Δ, ^{\tp} ⊢ □^{⪰ δ} B}
      \infer[dashed]1{Δ, ^{\tp}, \open(δ' i> δ) ⊢ B}
      \by2(rule:type-to-e){Δ, ^{\tp}, \open(δ' i> δ) ⊢ C}
      \by1(rule:type-bm-i){Δ, ^{\tp} ⊢ □^{⪰ δ} C}
    \end{prooftree*}
    By the IH, $Δ, δ \can{⊩} C$ holds, and
    so does $Δ, δ ⊩^{\can{ρ_{Γ}}} C$, which
    implies $Γ, γ \can{⊩} B \to C$.

    \proofsubparagraph{($\Longrightarrow$)}

    We argue by contrapositive.
    Suppose $Γ, ^{\tp} ⊬ □^{⪰ γ} (B \to C)$.
    Let $Δ \equiv Γ, ^{\tp}, \open(γ' i> γ), B^{\binder{δ}}$.
    Then we must have $Δ, ^{\tp} ⊬ □^{⪰ δ} C$;
    otherwise we could derive
    \begin{prooftree*}
      \hypo{Δ, ^{\tp} ⊢ □^{⪰ δ} C}
      \by1(rule:type-bm-e)
        {Γ, ^{\tp}, \open(γ' i> γ), B^{\binder{δ}} ⊢ C}
      \by1(rule:type-to-i){Γ, ^{\tp}, \open(γ' i> γ) ⊢ B \to C}
      \by1(rule:type-bm-i){Γ, ^{\tp} ⊢ □^{⪰ γ} (B \to C)}
    \end{prooftree*}
    a contradiction.
    By the IH, we have $Δ, δ \can{⊮} C$, and also
    $Δ, δ ⊮^{\can{ρ_{Γ}}} C$ as $\FC(C) \subseteq \DomC(Γ)$.
    Applying a similar argument to
    \begin{prooftree*}
      \by0(rule:type-var){Δ, ^{\tp}, \open(δ' i> δ) ⊢ B}
      \by1(rule:type-bm-i){Δ, ^{\tp} ⊢ □^{⪰ δ} B}
    \end{prooftree*}
    yields $Δ, δ ⊩^{\can{ρ_{Γ}}} B$.
    Since $Δ \can{≽} Γ$ and $δ \can{⪰_{Δ}} γ$,
    we see $Γ, γ \can{⊮} B \to C$.
  \end{case}
  \begin{case}[$A \equiv {□^{⪰ γ_1} B}$]\quitvmode\samepage
    \proofsubparagraph{($\Longleftarrow$)}

    Suppose $Γ, ^{\tp} ⊢ □^{⪰ γ} □^{⪰ γ_1} B$.
    Take $Δ \can{≽} Γ$ and
    $δ \can{⊒_{Δ}} γ$ satisfying $δ \can{⪰_{Δ}} γ_1$.
    Then we can derive
    \begin{prooftree*}
      \hypo{Γ, ^{\tp} ⊢ □^{⪰ γ} □^{⪰ γ_1} B}
      \infer[dashed]1{
        Δ, ^{\tp}, \open(δ' i> δ), ^{γ}, ^{\tp}
        ⊢ □^{⪰ γ} □^{⪰ γ_1} B
      }
      \by1(rule:type-bm-e)
        {Δ, ^{\tp}, \open(δ' i> δ), ^{γ} ⊢ □^{⪰ γ_1} B}
      \by1(rule:type-bm-e)
        {Δ, ^{\tp}, \open(δ' i> δ) ⊢ B}
      \by1(rule:type-bm-i)
        {Δ, ^{\tp} ⊢ □^{⪰ δ} B}
    \end{prooftree*}
    By the IH, $Δ, δ \can{⊩} B$ holds,
    and so does $Δ, δ ⊩^{\can{ρ_{Γ}}} B$,
    which implies $Γ, γ \can{⊩} □^{⪰ γ_1} B$.

    \proofsubparagraph{($\Longrightarrow$)}

    We argue by contrapositive.
    Suppose $Γ, ^{\tp} ⊬ □^{⪰ γ} □^{⪰ γ_1} B$.
    Let $Δ \equiv Γ, ^{\tp}, \open(γ' i> γ), \open(δ i> γ_1)$.
    Then we must have $Δ, ^{\tp} ⊬ □^{⪰ δ} B$; otherwise
    we could derive
    \begin{prooftree*}
      \hypo{Δ, ^{\tp} ⊢ □^{⪰ δ} B}
      \by1(rule:type-bm-e)
        {Γ, ^{\tp}, \open(γ' i> γ), \open(δ i> γ_1) ⊢ B}
      \by1(rule:type-bm-i)
        {Γ, ^{\tp}, \open(γ' i> γ) ⊢ □^{⪰ γ_1} B}
      \by1(rule:type-bm-i)
        {Γ, ^{\tp} ⊢ □^{⪰ γ} □^{⪰ γ_1} B}
    \end{prooftree*}
    a contradiction.
    By the IH, we have $Δ, δ \can{⊮} B$ and
    hence $Δ, δ ⊮^{\can{ρ_{Γ}}} B$.
    As $Δ \can{≽} Γ$ and $δ \can{⊒_{Δ}} γ$ with $δ \can{⪰_{Δ}} γ_1$,
    we see $Γ, γ \can{⊮} □^{⪰ γ_1} B$.
  \end{case}
  \begin{case}[$A \equiv ∀\cls(γ_2 i> γ_1). B$]\quitvmode\samepage
    \proofsubparagraph{($\Longleftarrow$)}

    Suppose $Γ, ^{\tp} ⊢ □^{⪰ γ} (∀\cls(γ_2 i> γ_1). B)$.
    Take $Δ \can{≽} Γ$ and $δ \can{⪰_{Δ}} γ_1$.
    Then we can derive
    \begin{prooftree*}
      \hypo{Γ, ^{\tp} ⊢ □^{⪰ γ} (∀\cls(γ_2 i> γ_1). B)}
      \infer[dashed]1{Δ, ^{\tp} ⊢ □^{⪰ γ} (∀\cls(γ_2 i> γ_1). B)}
      \infer[dashed]1{Δ, ^{\tp}, \open(γ' i> γ) ⊢ ∀\cls(γ_2 i> γ_1). B}
      \by1(rule:type-polycls-e)
        {Δ, ^{\tp}, \open(γ' i> γ) ⊢ B\w[γ_2 / δ]}
      \by1(rule:type-bm-i)
        {Δ, ^{\tp} ⊢ □^{⪰ γ} B\w[γ_2 / δ]}
    \end{prooftree*}
    By the IH, $Δ, γ \can{⊩} B\w[γ_2 / δ]$ holds, and
    so does $Δ, γ ⊩^{\can{ρ_{Γ}} ⋅ [γ_2 ↦ δ]} B$, which
    yields $Γ, γ \can{⊩} ∀\cls(γ_2 i> γ_1). B$.

    \proofsubparagraph{($\Longrightarrow$)}

    We argue by contrapositive.
    Suppose $Γ, ^{\tp} ⊬ □^{⪰ γ} (∀\cls(γ_2 i> γ_1). B)$.
    Let $Δ \equiv Γ, ^{\tp}, \open(γ' i> γ), \cls(γ_2 i> γ_1)$.
    Then we must have $Δ, ^{\tp} ⊬ □^{⪰ γ} B$;
    otherwise we could derive
    \begin{prooftree*}
      \hypo{Δ, ^{\tp} ⊢ □^{⪰ γ} B}
      \by1(rule:type-bm-e)
        {Γ, ^{\tp}, \open(γ' i> γ), \cls(γ_2 i> γ_1) ⊢ B}
      \by1(rule:type-polycls-i)
        {Γ, ^{\tp}, \open(γ' i> γ) ⊢ ∀\cls(γ_2 i> γ_1). B}
      \by1(rule:type-bm-i)
        {Γ, ^{\tp} ⊢ □^{⪰ γ} (∀\cls(γ_2 i> γ_1). B)}
    \end{prooftree*}
    a contradiction.
    By the IH, we have $Δ, γ \can{⊮} B$, and hence
    $Δ, γ ⊮^{\can{ρ_{Γ}} ⋅ [γ_2 ↦ γ_2]} B$.
    Since $Δ \can{≽} Γ$, we see $Γ, γ \can{⊮} ∀\cls(γ_2 i> γ_1). B$.
    \qedhere
  \end{case}
\end{proof}

\zref[claim]{claim:Kripke-completeness}

\begin{proof}
  By the contrapositive, where
  we can take the canonical model as countermodel
  by \Zcref{claim:truth-lemma}.
\end{proof}

\section{Full Definitions and Proofs for Section~\ref{sec:calc}}
\begin{definition}
  $ \token{FC}( \bmttnt{M} ) $ and $ \token{FC}( \Gamma ) $ represent sets of free classifiers in $\bmttnt{M}$ and $\Gamma$, respectively.
 \begin{align*}
   \token{FC}( \bmttmv{x} )  &=  \emptyset  \\
   \token{FC}(  \lambda \binder{ \bmttmv{x} }\has@{ \gamma } \bmttnt{A} \ldotp \bmttnt{M}  )  &=  \bmttsym{(}    \token{FC}( \bmttnt{M} )  - \bmttsym{\{}  \gamma  \bmttsym{\}}   \bmttsym{)} \cup  \token{FC}( \bmttnt{A} )  \\
   \token{FC}(  \bmttnt{M_{{\mathrm{1}}}}   \bmttnt{M_{{\mathrm{2}}}}  )  &=   \token{FC}( \bmttnt{M_{{\mathrm{1}}}} )  \cup  \token{FC}( \bmttnt{M_{{\mathrm{2}}}} )  \\
   \token{FC}(  \mathbf{quo} (\binder{ \tau })\lbrace^{\binder{ \gamma_{{\mathrm{1}}} } \within \gamma_{{\mathrm{2}}} }  \bmttnt{M} \rbrace  )  &=  \bmttsym{(}    \token{FC}( \bmttnt{M} )  \cup \bmttsym{\{}  \gamma_{{\mathrm{2}}}  \bmttsym{\}}   \bmttsym{)} - \bmttsym{\{}  \gamma_{{\mathrm{1}}}  \bmttsym{\}} \\
   \token{FC}(  \mathbf{unq} _{ \bmttnt{T} }\lbrace^{ \gamma } \bmttnt{M} \rbrace  )  &=   \token{FC}( \bmttnt{M} )  \cup \bmttsym{\{}  \gamma  \bmttsym{\}} \\
   \token{FC}(  \lambda \binder{ \gamma_{{\mathrm{1}}} }\within  \gamma_{{\mathrm{2}}} . \bmttnt{M}  )  &=  \bmttsym{(}    \token{FC}( \bmttnt{M} )  - \bmttsym{\{}  \gamma_{{\mathrm{1}}}  \bmttsym{\}}   \bmttsym{)} \cup \bmttsym{\{}  \gamma_{{\mathrm{2}}}  \bmttsym{\}} \\
   \token{FC}(  \bmttnt{M} \gamma  )  &=   \token{FC}( \bmttnt{M} )  \cup \bmttsym{\{}  \gamma  \bmttsym{\}}  \\
  & \\
   \token{FC}(  \varepsilon  )  &=  \emptyset  \\
   \token{FC}(  \binder{ \bmttmv{x} }\has@{ \gamma } \bmttnt{A}   \bmttsym{,}  \Gamma )  &=  \bmttsym{(}    \token{FC}( \Gamma )  - \bmttsym{\{}  \gamma  \bmttsym{\}}   \bmttsym{)} \cup  \token{FC}( \bmttnt{A} )   \\
   \token{FC}(  \tau : \mathord{\blacktriangleright} ^{\binder{ \gamma_{{\mathrm{1}}} } \within \gamma_{{\mathrm{2}}} }   \bmttsym{,}  \Gamma )  &=  \bmttsym{(}    \token{FC}( \Gamma )  - \bmttsym{\{}  \gamma_{{\mathrm{1}}}  \bmttsym{\}}   \bmttsym{)} \cup \bmttsym{\{}  \gamma_{{\mathrm{2}}}  \bmttsym{\}}  \\
   \token{FC}(  \mathord{\blacktriangleleft} _{ \bmttnt{T} }^{ \gamma }   \bmttsym{,}  \Gamma )  &=   \token{FC}( \Gamma )  \cup \bmttsym{\{}  \gamma  \bmttsym{\}}  \\
   \token{FC}(  \binder{ \gamma_{{\mathrm{1}}} } \within \gamma_{{\mathrm{2}}}   \bmttsym{,}  \Gamma )  &=  \bmttsym{(}    \token{FC}( \Gamma )  - \bmttsym{\{}  \gamma_{{\mathrm{1}}}  \bmttsym{\}}   \bmttsym{)} \cup \bmttsym{\{}  \gamma_{{\mathrm{2}}}  \bmttsym{\}} 
 \end{align*}
\end{definition}

\begin{definition}
 $ \token{FT}( \bmttnt{M} ) $ and $ \token{FT}( \Gamma ) $ represent sets of free occurrences of atomic modal transitions in $\bmttnt{M}$ and $\Gamma$, respectively.
 \begin{align*}
   \token{FT}( \bmttmv{x} )  &=  \emptyset  \\
   \token{FT}(  \lambda \binder{ \bmttmv{x} }\has@{ \gamma } \bmttnt{A} \ldotp \bmttnt{M}  )  &=  \token{FT}( \bmttnt{M} )  \\
   \token{FT}(  \bmttnt{M_{{\mathrm{1}}}}   \bmttnt{M_{{\mathrm{2}}}}  )  &=   \token{FT}( \bmttnt{M_{{\mathrm{1}}}} )  \cup  \token{FT}( \bmttnt{M_{{\mathrm{2}}}} )   \\
   \token{FT}(  \mathbf{quo} (\binder{ \tau })\lbrace^{\binder{ \gamma_{{\mathrm{1}}} } \within \gamma_{{\mathrm{2}}} }  \bmttnt{M} \rbrace  )  &=   \token{FT}( \bmttnt{M} )  -  \{ \tau \}   \\
   \token{FT}(  \mathbf{unq} _{  \overrightarrow{ \tau }  }\lbrace^{ \gamma } \bmttnt{M} \rbrace  )  &=   \token{FT}( \bmttnt{M} )  \cup  \overrightarrow{ \tau }   \\
   \token{FT}(  \lambda \binder{ \gamma_{{\mathrm{1}}} }\within  \gamma_{{\mathrm{2}}} . \bmttnt{M}  )  &=  \token{FT}( \bmttnt{M} )  \\
   \token{FT}(  \bmttnt{M} \gamma  )  &=  \token{FT}( \bmttnt{M} )  \\
  & \\
   \token{FT}(  \varepsilon  )  &=  \emptyset  \\
   \token{FT}(  \binder{ \bmttmv{x} }\has@{ \gamma } \bmttnt{A}   \bmttsym{,}  \Gamma )  &=  \token{FT}( \Gamma )  \\
   \token{FT}(  \tau : \mathord{\blacktriangleright} ^{\binder{ \gamma_{{\mathrm{1}}} } \within \gamma_{{\mathrm{2}}} }   \bmttsym{,}  \Gamma )  &=   \token{FT}( \Gamma )  -  \{ \tau \}   \\
   \token{FT}(  \mathord{\blacktriangleleft} _{  \overrightarrow{ \tau }  }^{ \gamma }   \bmttsym{,}  \Gamma )  &=   \token{FT}( \Gamma )  \cup  \overrightarrow{ \tau }   \\
   \token{FT}(  \binder{ \gamma_{{\mathrm{1}}} } \within \gamma_{{\mathrm{2}}}   \bmttsym{,}  \Gamma )  &=  \token{FT}( \Gamma ) 
 \end{align*}
\end{definition}

\begin{definition}
 $ \token{FV}( \bmttnt{M} ) $ represents the set of free variables in $\bmttnt{M}$.
 \begin{align*}
   \token{FV}( \bmttmv{x} )  &=  \{ \bmttmv{x} \}                   \\
   \token{FV}(  \lambda \binder{ \bmttmv{x} }\has@{ \gamma } \bmttnt{A} \ldotp \bmttnt{M}  )  &=   \token{FV}( \bmttnt{M} )  -  \{ \bmttmv{x} \}   \\
   \token{FV}(  \bmttnt{M_{{\mathrm{1}}}}   \bmttnt{M_{{\mathrm{2}}}}  )  &=   \token{FV}( \bmttnt{M_{{\mathrm{1}}}} )  \cup  \token{FV}( \bmttnt{M_{{\mathrm{2}}}} )    \\
   \token{FV}(  \mathbf{quo} (\binder{ \tau })\lbrace^{\binder{ \gamma_{{\mathrm{1}}} } \within \gamma_{{\mathrm{2}}} }  \bmttnt{M} \rbrace  )  &=  \token{FV}( \bmttnt{M} )       \\
   \token{FV}(  \mathbf{unq} _{ \bmttnt{T} }\lbrace^{ \gamma } \bmttnt{M} \rbrace  )  &=  \token{FV}( \bmttnt{M} )               \\
   \token{FV}(  \lambda \binder{ \gamma_{{\mathrm{1}}} }\within  \gamma_{{\mathrm{2}}} . \bmttnt{M}  )  &=  \token{FV}( \bmttnt{M} )           \\
   \token{FV}(  \bmttnt{M} \gamma  )  &=  \token{FV}( \bmttnt{M} ) 
 \end{align*}
\end{definition}

\begin{definition}
 $ \mathbf{Dom}_{\token{C} }( \Gamma ) $ represents a set of classifiers declared in $\Gamma$.
 \begin{align*}
    \mathbf{Dom}_{\token{C} }(  \varepsilon  )  &= \bmttsym{\{}   \exclam   \bmttsym{\}} \\
    \mathbf{Dom}_{\token{C} }( \Gamma  \bmttsym{,}   \binder{ \bmttmv{x} }\has@{ \gamma } \bmttnt{A}  )  &=   \mathbf{Dom}_{\token{C} }( \Gamma )  \cup \bmttsym{\{}  \gamma  \bmttsym{\}}  \\
    \mathbf{Dom}_{\token{C} }( \Gamma  \bmttsym{,}   \tau : \mathord{\blacktriangleright} ^{\binder{ \gamma_{{\mathrm{1}}} } \within \gamma_{{\mathrm{2}}} }  )  &=   \mathbf{Dom}_{\token{C} }( \Gamma )  \cup \bmttsym{\{}  \gamma_{{\mathrm{1}}}  \bmttsym{\}}  \\
    \mathbf{Dom}_{\token{C} }( \Gamma  \bmttsym{,}   \mathord{\blacktriangleleft} _{ \bmttnt{T} }^{ \gamma }  )  &=  \mathbf{Dom}_{\token{C} }( \Gamma )  \\
    \mathbf{Dom}_{\token{C} }( \Gamma  \bmttsym{,}   \binder{ \gamma_{{\mathrm{1}}} } \within \gamma_{{\mathrm{2}}}  )  &=   \mathbf{Dom}_{\token{C} }( \Gamma )  \cup \bmttsym{\{}  \gamma_{{\mathrm{1}}}  \bmttsym{\}} 
 \end{align*}
\end{definition}

\begin{definition}
  $ \mathbf{Dom}_{\token{T} }( \Gamma ) $ represents a set of atomic modal transition witnesses declared in $\Gamma$.
  \begin{align*}
    \mathbf{Dom}_{\token{T} }(  \varepsilon  )  &=  \emptyset  \\
    \mathbf{Dom}_{\token{T} }(  \binder{ \bmttmv{x} }\has@{ \gamma } \bmttnt{A}   \bmttsym{,}  \Gamma )  &=  \mathbf{Dom}_{\token{T} }( \Gamma )  \\
    \mathbf{Dom}_{\token{T} }(  \tau : \mathord{\blacktriangleright} ^{\binder{ \gamma_{{\mathrm{1}}} } \within \gamma_{{\mathrm{2}}} }   \bmttsym{,}  \Gamma )  &=   \mathbf{Dom}_{\token{T} }( \Gamma )  \cup  \{ \tau \}   \\
    \mathbf{Dom}_{\token{T} }(  \mathord{\blacktriangleleft} _{ \bmttnt{T} }^{ \gamma }   \bmttsym{,}  \Gamma )  &=  \mathbf{Dom}_{\token{T} }( \Gamma )  \\
    \mathbf{Dom}_{\token{T} }(  \binder{ \gamma_{{\mathrm{1}}} } \within \gamma_{{\mathrm{2}}}   \bmttsym{,}  \Gamma )  &=  \mathbf{Dom}_{\token{T} }( \Gamma ) 
  \end{align*}
\end{definition}

\begin{definition}
  $ \mathbf{Dom}_{\token{V} }( \Gamma ) $ represents a set of variables declared in $\Gamma$.
  \begin{align*}
    \mathbf{Dom}_{\token{V} }(  \varepsilon  )  &=  \emptyset  \\
    \mathbf{Dom}_{\token{V} }(  \binder{ \bmttmv{x} }\has@{ \gamma } \bmttnt{A}   \bmttsym{,}  \Gamma )  &=   \{ \bmttmv{x} \}  \cup  \mathbf{Dom}_{\token{V} }( \Gamma )   \\
    \mathbf{Dom}_{\token{V} }(  \tau : \mathord{\blacktriangleright} ^{\binder{ \gamma_{{\mathrm{1}}} } \within \gamma_{{\mathrm{2}}} }   \bmttsym{,}  \Gamma )  &=  \mathbf{Dom}_{\token{V} }( \Gamma )  \\
    \mathbf{Dom}_{\token{V} }(  \mathord{\blacktriangleleft} _{ \bmttnt{T} }^{ \gamma }   \bmttsym{,}  \Gamma )  &=  \mathbf{Dom}_{\token{V} }( \Gamma )  \\
    \mathbf{Dom}_{\token{V} }(  \binder{ \gamma_{{\mathrm{1}}} } \within \gamma_{{\mathrm{2}}}   \bmttsym{,}  \Gamma )  &=  \mathbf{Dom}_{\token{V} }( \Gamma ) 
  \end{align*}
\end{definition}

\begin{definition}
 A classifier substitution $ \bmttsym{(}   \mathord{-}   \bmttsym{)} [ \binder{ \gamma_{{\mathrm{1}}} } \coloneqq \gamma_{{\mathrm{2}}}  ] $ is a meta operation on terms and contexts, which replaces free occurrences of $\gamma_{{\mathrm{1}}}$ with $\gamma_{{\mathrm{2}}}$.
 \begin{align*}
   \bmttmv{x} [ \binder{ \gamma_{{\mathrm{1}}} } \coloneqq \gamma_{{\mathrm{2}}}  ]  &=  \bmttmv{x} \\
   \bmttsym{(}   \lambda \binder{ \bmttmv{x} }\has@{ \gamma_{{\mathrm{1}}} } \bmttnt{A} \ldotp \bmttnt{M}   \bmttsym{)} [ \binder{ \gamma_{{\mathrm{2}}} } \coloneqq \gamma_{{\mathrm{3}}}  ]  &=   \lambda \binder{ \bmttmv{x} }\has@{ \gamma_{{\mathrm{1}}} }  \bmttnt{A} [ \binder{ \gamma_{{\mathrm{2}}} } \coloneqq \gamma_{{\mathrm{3}}}  ]  \ldotp \bmttsym{(}   \bmttnt{M} [ \binder{ \gamma_{{\mathrm{2}}} } \coloneqq \gamma_{{\mathrm{3}}}  ]   \bmttsym{)}  \\
  & \qquad \text{where $\gamma_{{\mathrm{1}}} \, \notin \, \bmttsym{\{}  \gamma_{{\mathrm{2}}}  \bmttsym{,}  \gamma_{{\mathrm{3}}}  \bmttsym{\}}$}\\
   \bmttsym{(}   \bmttnt{M}   \bmttnt{N}   \bmttsym{)} [ \binder{ \gamma_{{\mathrm{1}}} } \coloneqq \gamma_{{\mathrm{2}}}  ]  &=  \bmttsym{(}   \bmttnt{M} [ \binder{ \gamma_{{\mathrm{1}}} } \coloneqq \gamma_{{\mathrm{2}}}  ]   \bmttsym{)}   \bmttsym{(}   \bmttnt{N} [ \binder{ \gamma_{{\mathrm{1}}} } \coloneqq \gamma_{{\mathrm{2}}}  ]   \bmttsym{)}  \\
    \mathbf{quo} (\binder{ \tau })\lbrace^{\binder{ \gamma_{{\mathrm{1}}} } \within \gamma_{{\mathrm{2}}} }  \bmttnt{M} \rbrace  [ \binder{ \gamma_{{\mathrm{3}}} } \coloneqq \gamma_{{\mathrm{4}}}  ]  &=  \mathbf{quo} (\binder{ \tau })\lbrace^{\binder{ \gamma_{{\mathrm{1}}} } \within  \gamma_{{\mathrm{2}}} [ \binder{ \gamma_{{\mathrm{3}}} } \coloneqq \gamma_{{\mathrm{4}}}  ]  }   \bmttnt{M} [ \binder{ \gamma_{{\mathrm{3}}} } \coloneqq \gamma_{{\mathrm{4}}}  ]  \rbrace  \\
  & \qquad \text{where $\gamma_{{\mathrm{1}}} \, \notin \, \bmttsym{\{}  \gamma_{{\mathrm{3}}}  \bmttsym{,}  \gamma_{{\mathrm{4}}}  \bmttsym{\}}$}\\
    \mathbf{unq} _{ \bmttnt{T} }\lbrace^{ \gamma_{{\mathrm{1}}} } \bmttnt{M} \rbrace  [ \binder{ \gamma_{{\mathrm{2}}} } \coloneqq \gamma_{{\mathrm{3}}}  ]  &=  \mathbf{unq} _{ \bmttnt{T} }\lbrace^{  \gamma_{{\mathrm{1}}} [ \binder{ \gamma_{{\mathrm{2}}} } \coloneqq \gamma_{{\mathrm{3}}}  ]  }  \bmttnt{M} [ \binder{ \gamma_{{\mathrm{2}}} } \coloneqq \gamma_{{\mathrm{3}}}  ]  \rbrace  \\
   \bmttsym{(}   \lambda \binder{ \gamma_{{\mathrm{1}}} }\within  \gamma_{{\mathrm{2}}} . \bmttnt{M}   \bmttsym{)} [ \binder{ \gamma_{{\mathrm{3}}} } \coloneqq \gamma_{{\mathrm{4}}}  ]  &=  \lambda \binder{ \gamma_{{\mathrm{1}}} }\within   \gamma_{{\mathrm{2}}} [ \binder{ \gamma_{{\mathrm{3}}} } \coloneqq \gamma_{{\mathrm{4}}}  ]  . \bmttsym{(}   \bmttnt{M} [ \binder{ \gamma_{{\mathrm{3}}} } \coloneqq \gamma_{{\mathrm{4}}}  ]   \bmttsym{)} \\
  & \qquad \text{where $\gamma_{{\mathrm{1}}} \, \notin \, \bmttsym{\{}  \gamma_{{\mathrm{3}}}  \bmttsym{,}  \gamma_{{\mathrm{4}}}  \bmttsym{\}}$}\\
   \bmttsym{(}   \bmttnt{M} \gamma_{{\mathrm{1}}}   \bmttsym{)} [ \binder{ \gamma_{{\mathrm{2}}} } \coloneqq \gamma_{{\mathrm{3}}}  ]  &=  \bmttsym{(}   \bmttnt{M} [ \binder{ \gamma_{{\mathrm{2}}} } \coloneqq \gamma_{{\mathrm{3}}}  ]   \bmttsym{)}   \gamma_{{\mathrm{1}}} [ \binder{ \gamma_{{\mathrm{2}}} } \coloneqq \gamma_{{\mathrm{3}}}  ]    \\
  & \\
    \varepsilon  [ \binder{ \gamma_{{\mathrm{1}}} } \coloneqq \gamma_{{\mathrm{2}}}  ]  &=  \varepsilon  \\
   \bmttsym{(}   \binder{ \bmttmv{x} }\has@{ \gamma_{{\mathrm{1}}} } \bmttnt{A}   \bmttsym{,}  \Gamma  \bmttsym{)} [ \binder{ \gamma_{{\mathrm{2}}} } \coloneqq \gamma_{{\mathrm{3}}}  ]  &=   \binder{ \bmttmv{x} }\has@{ \gamma_{{\mathrm{1}}} }  \bmttnt{A} [ \binder{ \gamma_{{\mathrm{2}}} } \coloneqq \gamma_{{\mathrm{3}}}  ]    \bmttsym{,}  \Gamma [ \binder{ \gamma_{{\mathrm{2}}} } \coloneqq \gamma_{{\mathrm{3}}}  ] \\
  & \qquad \text{where $\gamma_{{\mathrm{1}}} \, \notin \, \bmttsym{\{}  \gamma_{{\mathrm{2}}}  \bmttsym{,}  \gamma_{{\mathrm{3}}}  \bmttsym{\}}$}\\
   \bmttsym{(}   \tau : \mathord{\blacktriangleright} ^{\binder{ \gamma_{{\mathrm{1}}} } \within \gamma_{{\mathrm{2}}} }   \bmttsym{,}  \Gamma  \bmttsym{)} [ \binder{ \gamma_{{\mathrm{3}}} } \coloneqq \gamma_{{\mathrm{4}}}  ]  &=  \tau : \mathord{\blacktriangleright} ^{\binder{ \gamma_{{\mathrm{1}}} } \within   \gamma_{{\mathrm{2}}} [ \binder{ \gamma_{{\mathrm{3}}} } \coloneqq \gamma_{{\mathrm{4}}}  ]   }   \bmttsym{,}    \Gamma [ \binder{ \gamma_{{\mathrm{3}}} } \coloneqq \gamma_{{\mathrm{4}}}  ]  \\
  & \qquad \text{where $\gamma_{{\mathrm{1}}} \, \notin \, \bmttsym{\{}  \gamma_{{\mathrm{3}}}  \bmttsym{,}  \gamma_{{\mathrm{4}}}  \bmttsym{\}}$}\\
   \bmttsym{(}   \mathord{\blacktriangleleft} _{ \bmttnt{T} }^{ \gamma_{{\mathrm{1}}} }   \bmttsym{,}  \Gamma  \bmttsym{)} [ \binder{ \gamma_{{\mathrm{2}}} } \coloneqq \gamma_{{\mathrm{3}}}  ]  &=   \mathord{\blacktriangleleft} _{ \bmttnt{T} }^{   \gamma_{{\mathrm{1}}} [ \binder{ \gamma_{{\mathrm{2}}} } \coloneqq \gamma_{{\mathrm{3}}}  ]   }   \bmttsym{,}  \Gamma [ \binder{ \gamma_{{\mathrm{2}}} } \coloneqq \gamma_{{\mathrm{3}}}  ] \\
   \bmttsym{(}   \binder{ \gamma_{{\mathrm{1}}} } \within \gamma_{{\mathrm{2}}}   \bmttsym{,}  \Gamma  \bmttsym{)} [ \binder{ \gamma_{{\mathrm{3}}} } \coloneqq \gamma_{{\mathrm{4}}}  ]  &=   \binder{ \gamma_{{\mathrm{1}}} } \within   \gamma_{{\mathrm{2}}} [ \binder{ \gamma_{{\mathrm{3}}} } \coloneqq \gamma_{{\mathrm{4}}}  ]     \bmttsym{,}  \Gamma [ \binder{ \gamma_{{\mathrm{3}}} } \coloneqq \gamma_{{\mathrm{4}}}  ] \\
  & \qquad \text{where $\gamma_{{\mathrm{1}}} \, \notin \, \bmttsym{\{}  \gamma_{{\mathrm{3}}}  \bmttsym{,}  \gamma_{{\mathrm{4}}}  \bmttsym{\}}$}\\
 \end{align*}
\end{definition}

\begin{definition}
  $ \bmttnt{T_{{\mathrm{1}}}} [ \binder{ \tau } \coloneqq \bmttnt{T_{{\mathrm{2}}}}  ] $ is a modal transition witness that is obtained by replacing free occurrences of $\tau$ in $\bmttnt{T_{{\mathrm{1}}}}$ with $\bmttnt{T_{{\mathrm{2}}}}$. Formally, it is defined as follows:
  \begin{align*}
      \varepsilon  [ \binder{ \tau } \coloneqq \bmttnt{T_{{\mathrm{2}}}}  ]  &=  \varepsilon  \\
     \bmttsym{(}  \tau_{{\mathrm{1}}}  \bmttsym{,}  \bmttnt{T_{{\mathrm{1}}}}  \bmttsym{)} [ \binder{ \tau } \coloneqq \bmttnt{T_{{\mathrm{2}}}}  ]  &= \begin{cases}
       \bmttnt{T_{{\mathrm{2}}}}  +  \bmttsym{(}   \bmttnt{T_{{\mathrm{1}}}} [ \binder{ \tau } \coloneqq \bmttnt{T_{{\mathrm{2}}}}  ]   \bmttsym{)}  & \text{\ where $\tau_{{\mathrm{1}}}  \bmttsym{=}  \tau$} \\
      \tau_{{\mathrm{1}}}  \bmttsym{,}  \bmttsym{(}   \bmttnt{T_{{\mathrm{1}}}} [ \binder{ \tau } \coloneqq \bmttnt{T_{{\mathrm{2}}}}  ]   \bmttsym{)} & \text{\ otherwise}
    \end{cases}
  \end{align*}
\end{definition}

\begin{definition}
  An atomic modal transition witness substitution $ \bmttsym{(}   \mathord{-}   \bmttsym{)} [ \binder{ \gamma_{{\mathrm{1}}} } \coloneqq \gamma_{{\mathrm{2}}} , \binder{ \tau }  \coloneqq   \bmttnt{T}  ] $ is a meta operation on terms and contexts, which replaces free occurrences of $\gamma_{{\mathrm{1}}}$ and $\tau$ with $\gamma_{{\mathrm{2}}}$ and $\bmttnt{T}, respectively$.
  \begin{align*}
     \bmttmv{x} [ \binder{ \gamma_{{\mathrm{1}}} } \coloneqq \gamma_{{\mathrm{2}}} , \binder{ \tau }  \coloneqq   \bmttnt{T}  ]  &=  \bmttmv{x} \\
     \bmttsym{(}   \lambda \binder{ \bmttmv{x} }\has@{ \gamma_{{\mathrm{1}}} } \bmttnt{A} \ldotp \bmttnt{M}   \bmttsym{)} [ \binder{ \gamma_{{\mathrm{2}}} } \coloneqq \gamma_{{\mathrm{3}}} , \binder{ \tau }  \coloneqq   \bmttnt{T}  ]  &=   \lambda \binder{ \bmttmv{x} }\has@{ \gamma_{{\mathrm{1}}} }  \bmttnt{A} [ \binder{ \gamma_{{\mathrm{2}}} } \coloneqq \gamma_{{\mathrm{3}}}  ]  \ldotp \bmttsym{(}   \bmttnt{M} [ \binder{ \gamma_{{\mathrm{2}}} } \coloneqq \gamma_{{\mathrm{3}}} , \binder{ \tau }  \coloneqq   \bmttnt{T}  ]   \bmttsym{)}  \\
    & \qquad \text{where $\gamma_{{\mathrm{1}}} \, \notin \, \bmttsym{\{}  \gamma_{{\mathrm{2}}}  \bmttsym{,}  \gamma_{{\mathrm{3}}}  \bmttsym{\}}$}\\
     \bmttsym{(}   \bmttnt{M}   \bmttnt{N}   \bmttsym{)} [ \binder{ \gamma_{{\mathrm{1}}} } \coloneqq \gamma_{{\mathrm{2}}} , \binder{ \tau }  \coloneqq   \bmttnt{T}  ]  &=  \bmttsym{(}   \bmttnt{M} [ \binder{ \gamma_{{\mathrm{1}}} } \coloneqq \gamma_{{\mathrm{2}}} , \binder{ \tau }  \coloneqq   \bmttnt{T}  ]   \bmttsym{)}   \bmttsym{(}   \bmttnt{N} [ \binder{ \gamma_{{\mathrm{1}}} } \coloneqq \gamma_{{\mathrm{2}}} , \binder{ \tau }  \coloneqq   \bmttnt{T}  ]   \bmttsym{)}  \\
      \mathbf{quo} (\binder{ \tau_{{\mathrm{1}}} })\lbrace^{\binder{ \gamma_{{\mathrm{1}}} } \within \gamma_{{\mathrm{2}}} }  \bmttnt{M} \rbrace  [ \binder{ \gamma_{{\mathrm{3}}} } \coloneqq \gamma_{{\mathrm{4}}} , \binder{ \tau_{{\mathrm{2}}} }  \coloneqq   \bmttnt{T}  ]  &=  \mathbf{quo} (\binder{ \tau_{{\mathrm{1}}} })\lbrace^{\binder{ \gamma_{{\mathrm{1}}} } \within  \gamma_{{\mathrm{2}}} [ \binder{ \gamma_{{\mathrm{3}}} } \coloneqq \gamma_{{\mathrm{4}}}  ]  }   \bmttnt{M} [ \binder{ \gamma_{{\mathrm{3}}} } \coloneqq \gamma_{{\mathrm{4}}} , \binder{ \tau_{{\mathrm{2}}} }  \coloneqq   \bmttnt{T}  ]  \rbrace  \\
    & \qquad \text{where $\gamma_{{\mathrm{1}}} \, \notin \, \bmttsym{\{}  \gamma_{{\mathrm{3}}}  \bmttsym{,}  \gamma_{{\mathrm{4}}}  \bmttsym{\}}$ and $\tau_{{\mathrm{1}}} \, \notin \,   \{ \tau_{{\mathrm{2}}} \}  \cup \bmttnt{T} $}\\
      \mathbf{unq} _{ \bmttnt{T_{{\mathrm{1}}}} }\lbrace^{ \gamma_{{\mathrm{1}}} } \bmttnt{M} \rbrace  [ \binder{ \gamma_{{\mathrm{2}}} } \coloneqq \gamma_{{\mathrm{3}}} , \binder{ \tau }  \coloneqq   \bmttnt{T_{{\mathrm{2}}}}  ]  &=  \mathbf{unq} _{  \bmttnt{T_{{\mathrm{1}}}} [ \binder{ \tau } \coloneqq \bmttnt{T_{{\mathrm{2}}}}  ]  }\lbrace^{  \gamma_{{\mathrm{1}}} [ \binder{ \gamma_{{\mathrm{2}}} } \coloneqq \gamma_{{\mathrm{3}}}  ]  }  \bmttnt{M} [ \binder{ \gamma_{{\mathrm{2}}} } \coloneqq \gamma_{{\mathrm{3}}} , \binder{ \tau }  \coloneqq   \bmttnt{T_{{\mathrm{2}}}}  ]  \rbrace  \\
     \bmttsym{(}   \lambda \binder{ \gamma_{{\mathrm{1}}} }\within  \gamma_{{\mathrm{2}}} . \bmttnt{M}   \bmttsym{)} [ \binder{ \gamma_{{\mathrm{3}}} } \coloneqq \gamma_{{\mathrm{4}}} , \binder{ \tau }  \coloneqq   \bmttnt{T}  ]  &=  \lambda \binder{ \gamma_{{\mathrm{1}}} }\within   \gamma_{{\mathrm{2}}} [ \binder{ \gamma_{{\mathrm{3}}} } \coloneqq \gamma_{{\mathrm{4}}}  ]  . \bmttsym{(}   \bmttnt{M} [ \binder{ \gamma_{{\mathrm{3}}} } \coloneqq \gamma_{{\mathrm{4}}} , \binder{ \tau }  \coloneqq   \bmttnt{T}  ]   \bmttsym{)} \\
    & \qquad \text{where $\gamma_{{\mathrm{1}}} \, \notin \, \bmttsym{\{}  \gamma_{{\mathrm{3}}}  \bmttsym{,}  \gamma_{{\mathrm{4}}}  \bmttsym{\}}$}\\
     \bmttsym{(}   \bmttnt{M} \gamma_{{\mathrm{1}}}   \bmttsym{)} [ \binder{ \gamma_{{\mathrm{2}}} } \coloneqq \gamma_{{\mathrm{3}}} , \binder{ \tau }  \coloneqq   \bmttnt{T}  ]  &=  \bmttsym{(}   \bmttnt{M} [ \binder{ \gamma_{{\mathrm{2}}} } \coloneqq \gamma_{{\mathrm{3}}} , \binder{ \tau }  \coloneqq   \bmttnt{T}  ]   \bmttsym{)}   \gamma_{{\mathrm{1}}} [ \binder{ \gamma_{{\mathrm{2}}} } \coloneqq \gamma_{{\mathrm{3}}}  ]    \\
    & \\
      \varepsilon  [ \binder{ \gamma_{{\mathrm{1}}} } \coloneqq \gamma_{{\mathrm{2}}} , \binder{ \tau }  \coloneqq   \bmttnt{T}  ]  &=  \varepsilon  \\
     \bmttsym{(}   \binder{ \bmttmv{x} }\has@{ \gamma_{{\mathrm{1}}} } \bmttnt{A}   \bmttsym{,}  \Gamma  \bmttsym{)} [ \binder{ \gamma_{{\mathrm{2}}} } \coloneqq \gamma_{{\mathrm{3}}} , \binder{ \tau }  \coloneqq   \bmttnt{T}  ]  &=   \binder{ \bmttmv{x} }\has@{ \gamma_{{\mathrm{1}}} }  \bmttnt{A} [ \binder{ \gamma_{{\mathrm{2}}} } \coloneqq \gamma_{{\mathrm{3}}}  ]    \bmttsym{,}  \Gamma [ \binder{ \gamma_{{\mathrm{2}}} } \coloneqq \gamma_{{\mathrm{3}}} , \binder{ \tau }  \coloneqq   \bmttnt{T}  ] \\
    & \qquad \text{where $\gamma_{{\mathrm{1}}} \, \notin \, \bmttsym{\{}  \gamma_{{\mathrm{2}}}  \bmttsym{,}  \gamma_{{\mathrm{3}}}  \bmttsym{\}}$}\\
     \bmttsym{(}   \tau_{{\mathrm{1}}} : \mathord{\blacktriangleright} ^{\binder{ \gamma_{{\mathrm{1}}} } \within \gamma_{{\mathrm{2}}} }   \bmttsym{,}  \Gamma  \bmttsym{)} [ \binder{ \gamma_{{\mathrm{3}}} } \coloneqq \gamma_{{\mathrm{4}}} , \binder{ \tau_{{\mathrm{2}}} }  \coloneqq   \bmttnt{T}  ]  &=  \tau_{{\mathrm{1}}} : \mathord{\blacktriangleright} ^{\binder{ \gamma_{{\mathrm{1}}} } \within   \gamma_{{\mathrm{2}}} [ \binder{ \gamma_{{\mathrm{3}}} } \coloneqq \gamma_{{\mathrm{4}}}  ]   }   \bmttsym{,}    \Gamma [ \binder{ \gamma_{{\mathrm{3}}} } \coloneqq \gamma_{{\mathrm{4}}} , \binder{ \tau_{{\mathrm{2}}} }  \coloneqq   \bmttnt{T}  ]  \\
    & \qquad \text{where $\gamma_{{\mathrm{1}}} \, \notin \, \bmttsym{\{}  \gamma_{{\mathrm{3}}}  \bmttsym{,}  \gamma_{{\mathrm{4}}}  \bmttsym{\}}$ and $\tau_{{\mathrm{1}}} \, \notin \,   \{ \tau_{{\mathrm{2}}} \}  \cup \bmttnt{T} $}\\
     \bmttsym{(}   \mathord{\blacktriangleleft} _{ \bmttnt{T_{{\mathrm{1}}}} }^{ \gamma_{{\mathrm{1}}} }   \bmttsym{,}  \Gamma  \bmttsym{)} [ \binder{ \gamma_{{\mathrm{2}}} } \coloneqq \gamma_{{\mathrm{3}}} , \binder{ \tau }  \coloneqq   \bmttnt{T_{{\mathrm{2}}}}  ]  &=    \mathord{\blacktriangleleft} _{  \bmttnt{T_{{\mathrm{1}}}} [ \binder{ \tau } \coloneqq \bmttnt{T_{{\mathrm{2}}}}  ]  }^{ \gamma_{{\mathrm{1}}} }  [ \binder{ \gamma_{{\mathrm{2}}} } \coloneqq \gamma_{{\mathrm{3}}}  ]   \bmttsym{,}  \Gamma [ \binder{ \gamma_{{\mathrm{2}}} } \coloneqq \gamma_{{\mathrm{3}}} , \binder{ \tau }  \coloneqq   \bmttnt{T_{{\mathrm{2}}}}  ] \\
     \bmttsym{(}   \binder{ \gamma_{{\mathrm{1}}} } \within \gamma_{{\mathrm{2}}}   \bmttsym{,}  \Gamma  \bmttsym{)} [ \binder{ \gamma_{{\mathrm{3}}} } \coloneqq \gamma_{{\mathrm{4}}} , \binder{ \tau }  \coloneqq   \bmttnt{T}  ]  &=   \binder{ \gamma_{{\mathrm{1}}} } \within   \gamma_{{\mathrm{2}}} [ \binder{ \gamma_{{\mathrm{3}}} } \coloneqq \gamma_{{\mathrm{4}}}  ]     \bmttsym{,}  \Gamma [ \binder{ \gamma_{{\mathrm{3}}} } \coloneqq \gamma_{{\mathrm{4}}} , \binder{ \tau }  \coloneqq   \bmttnt{T}  ] \\
    & \qquad \text{where $\gamma_{{\mathrm{1}}} \, \notin \, \bmttsym{\{}  \gamma_{{\mathrm{3}}}  \bmttsym{,}  \gamma_{{\mathrm{4}}}  \bmttsym{\}}$}
  \end{align*}
\end{definition}

\begin{definition}
A variable substitution $ \bmttsym{(}   \mathord{-}   \bmttsym{)} [ \binder{ \gamma_{{\mathrm{1}}} } \coloneqq \gamma_{{\mathrm{2}}} , \binder{ \bmttmv{x} }  \coloneqq   \bmttnt{M}  ] $ is a meta operation on terms that replaces free occurrences of $\gamma_{{\mathrm{1}}}$ and $\bmttmv{x}$ with $\gamma_{{\mathrm{2}}}$ and $\bmttnt{M}$, respectively.
 \begin{align*}
   \bmttmv{x} [ \binder{ \gamma_{{\mathrm{1}}} } \coloneqq \gamma_{{\mathrm{2}}} , \binder{ \bmttmv{y} }  \coloneqq   \bmttnt{M}  ]  &=  \begin{cases}
                             M & \text{\ where $\bmttmv{x}  \bmttsym{=}  \bmttmv{y}$} \\
                             x & \text{\ otherwise}
                            \end{cases} \\
   \bmttsym{(}   \lambda \binder{ \bmttmv{x} }\has@{ \gamma_{{\mathrm{1}}} } \bmttnt{A} \ldotp \bmttnt{M}   \bmttsym{)} [ \binder{ \gamma_{{\mathrm{2}}} } \coloneqq \gamma_{{\mathrm{3}}} , \binder{ \bmttmv{y} }  \coloneqq   \bmttnt{N}  ]  &=   \lambda \binder{ \bmttmv{x} }\has@{ \gamma_{{\mathrm{1}}} }  \bmttnt{A} [ \binder{ \gamma_{{\mathrm{2}}} } \coloneqq \gamma_{{\mathrm{3}}}  ]  \ldotp \bmttsym{(}   \bmttnt{M} [ \binder{ \gamma_{{\mathrm{2}}} } \coloneqq \gamma_{{\mathrm{3}}} , \binder{ \bmttmv{y} }  \coloneqq   \bmttnt{N}  ]   \bmttsym{)}  \\
  & \qquad \text{where $\gamma_{{\mathrm{1}}} \, \notin \, \bmttsym{\{}  \gamma_{{\mathrm{2}}}  \bmttsym{,}  \gamma_{{\mathrm{3}}}  \bmttsym{\}}$ and $\bmttmv{x} \, \notin \,   \{ \bmttmv{y} \}  \cup  \token{FV}( \bmttnt{N} )  $}\\
   \bmttsym{(}   \bmttnt{M_{{\mathrm{1}}}}   \bmttnt{M_{{\mathrm{2}}}}   \bmttsym{)} [ \binder{ \gamma_{{\mathrm{1}}} } \coloneqq \gamma_{{\mathrm{2}}} , \binder{ \bmttmv{x} }  \coloneqq   \bmttnt{N}  ]  &=   \bmttnt{M_{{\mathrm{1}}}} [ \binder{ \gamma_{{\mathrm{1}}} } \coloneqq \gamma_{{\mathrm{2}}} , \binder{ \bmttmv{x} }  \coloneqq   \bmttnt{N}  ]      \bmttnt{M_{{\mathrm{2}}}} [ \binder{ \gamma_{{\mathrm{1}}} } \coloneqq \gamma_{{\mathrm{2}}} , \binder{ \bmttmv{x} }  \coloneqq   \bmttnt{N}  ]    \\
    \mathbf{quo} (\binder{ \tau })\lbrace^{\binder{ \gamma_{{\mathrm{1}}} } \within \gamma_{{\mathrm{2}}} }  \bmttnt{M} \rbrace  [ \binder{ \gamma_{{\mathrm{3}}} } \coloneqq \gamma_{{\mathrm{4}}} , \binder{ \bmttmv{x} }  \coloneqq   \bmttnt{N}  ]  &=  \mathbf{quo} (\binder{ \tau })\lbrace^{\binder{ \gamma_{{\mathrm{1}}} } \within  \gamma_{{\mathrm{2}}} [ \binder{ \gamma_{{\mathrm{3}}} } \coloneqq \gamma_{{\mathrm{4}}}  ]  }   \bmttnt{M} [ \binder{ \gamma_{{\mathrm{3}}} } \coloneqq \gamma_{{\mathrm{4}}} , \binder{ \bmttmv{x} }  \coloneqq   \bmttnt{N}  ]  \rbrace  \\
  & \qquad \text{where $\gamma_{{\mathrm{1}}} \, \notin \, \bmttsym{\{}  \gamma_{{\mathrm{3}}}  \bmttsym{,}  \gamma_{{\mathrm{4}}}  \bmttsym{\}}$}\\
    \mathbf{unq} _{ \bmttnt{T} }\lbrace^{ \gamma_{{\mathrm{1}}} } \bmttnt{M} \rbrace  [ \binder{ \gamma_{{\mathrm{2}}} } \coloneqq \gamma_{{\mathrm{3}}} , \binder{ \bmttmv{x} }  \coloneqq   \bmttnt{N}  ]  &=  \mathbf{unq} _{ \bmttnt{T} }\lbrace^{  \gamma_{{\mathrm{1}}} [ \binder{ \gamma_{{\mathrm{2}}} } \coloneqq \gamma_{{\mathrm{3}}}  ]  }  \bmttnt{M} [ \binder{ \gamma_{{\mathrm{2}}} } \coloneqq \gamma_{{\mathrm{3}}} , \binder{ \bmttmv{x} }  \coloneqq   \bmttnt{N}  ]  \rbrace  \\
   \bmttsym{(}   \lambda \binder{ \gamma_{{\mathrm{1}}} }\within  \gamma_{{\mathrm{2}}} . \bmttnt{M}   \bmttsym{)} [ \binder{ \gamma_{{\mathrm{3}}} } \coloneqq \gamma_{{\mathrm{4}}} , \binder{ \bmttmv{x} }  \coloneqq   \bmttnt{N}  ]  &=  \lambda \binder{ \gamma_{{\mathrm{1}}} }\within   \gamma_{{\mathrm{2}}} [ \binder{ \gamma_{{\mathrm{3}}} } \coloneqq \gamma_{{\mathrm{4}}}  ]  . \bmttsym{(}   \bmttnt{M} [ \binder{ \gamma_{{\mathrm{3}}} } \coloneqq \gamma_{{\mathrm{4}}} , \binder{ \bmttmv{x} }  \coloneqq   \bmttnt{N}  ]   \bmttsym{)} \\
  & \qquad \text{where $\gamma_{{\mathrm{1}}} \, \notin \, \bmttsym{\{}  \gamma_{{\mathrm{3}}}  \bmttsym{,}  \gamma_{{\mathrm{4}}}  \bmttsym{\}}$}\\
   \bmttsym{(}   \bmttnt{M} \gamma_{{\mathrm{1}}}   \bmttsym{)} [ \binder{ \gamma_{{\mathrm{2}}} } \coloneqq \gamma_{{\mathrm{3}}} , \binder{ \bmttmv{x} }  \coloneqq   \bmttnt{N}  ]  &=   \bmttnt{M} [ \binder{ \gamma_{{\mathrm{2}}} } \coloneqq \gamma_{{\mathrm{3}}} , \binder{ \bmttmv{x} }  \coloneqq   \bmttnt{N}  ]    \gamma_{{\mathrm{1}}} [ \binder{ \gamma_{{\mathrm{2}}} } \coloneqq \gamma_{{\mathrm{3}}}  ]   
 \end{align*}
\end{definition}

\Zcref{claim:varsubst,claim:rebasing,claim:clssubst} are proved as parts of the following lemmas.

\begin{lemma}[Variable Substitution (Full Version)]
 Suppose $\Delta_{{\mathrm{1}}} =  \Gamma_{{\mathrm{1}}} ^{\position{ \gamma_{{\mathrm{1}}} } }   \bmttsym{,}   \binder{ \bmttmv{x} }\has@{ \gamma_{{\mathrm{2}}} } \bmttnt{A}   \bmttsym{,}  \Gamma_{{\mathrm{2}}}$, and $\Delta_{{\mathrm{2}}} =  \Gamma_{{\mathrm{1}}}  \bmttsym{,}  \Gamma_{{\mathrm{2}}} [ \binder{ \gamma_{{\mathrm{2}}} } \coloneqq \gamma_{{\mathrm{1}}}  ] $. Then, the following statements hold.
 \begin{enumerate}
 \item $\vdash  \Delta_{{\mathrm{1}}}  \hasType \, \bmttkw{ctx} \implies\ \vdash  \Delta_{{\mathrm{2}}}  \hasType \, \bmttkw{ctx}$.
 \item $\Delta_{{\mathrm{1}}}  \vdash  \bmttnt{A}  \hasType \, \bmttkw{type} \implies \Delta_{{\mathrm{2}}}  \vdash   \bmttnt{A} [ \binder{ \gamma_{{\mathrm{2}}} } \coloneqq \gamma_{{\mathrm{1}}}  ]   \hasType \, \bmttkw{type}$.
 \item $\Delta_{{\mathrm{1}}}  \vdash  \delta_{{\mathrm{1}}}  \preceq  \delta_{{\mathrm{2}}} \implies \Delta_{{\mathrm{2}}}  \vdash   \delta_{{\mathrm{1}}} [ \binder{ \gamma_{{\mathrm{2}}} } \coloneqq \gamma_{{\mathrm{1}}}  ]   \preceq   \delta_{{\mathrm{2}}} [ \binder{ \gamma_{{\mathrm{2}}} } \coloneqq \gamma_{{\mathrm{1}}}  ] $.
 \item $ \Delta_{{\mathrm{1}}} \vdash \bmttnt{T} : \delta_{{\mathrm{1}}} \sqsubseteq \delta_{{\mathrm{2}}}  \implies  \Delta_{{\mathrm{2}}} \vdash \bmttnt{T} :  \delta_{{\mathrm{1}}} [ \binder{ \gamma_{{\mathrm{2}}} } \coloneqq \gamma_{{\mathrm{1}}}  ]  \sqsubseteq  \delta_{{\mathrm{2}}} [ \binder{ \gamma_{{\mathrm{2}}} } \coloneqq \gamma_{{\mathrm{1}}}  ]  $.
 \item $\Delta_{{\mathrm{1}}}  \vdash  \bmttnt{M_{{\mathrm{1}}}}  \hasType  \bmttnt{B}$ and $\Gamma_{{\mathrm{1}}}  \vdash  \bmttnt{M_{{\mathrm{2}}}}  \hasType  \bmttnt{A} \implies \Delta_{{\mathrm{2}}}  \vdash   \bmttnt{M_{{\mathrm{1}}}} [ \binder{ \gamma_{{\mathrm{2}}} } \coloneqq \gamma_{{\mathrm{1}}} , \binder{ \bmttmv{x} }  \coloneqq   \bmttnt{M_{{\mathrm{2}}}}  ]   \hasType   \bmttnt{B} [ \binder{ \gamma_{{\mathrm{2}}} } \coloneqq \gamma_{{\mathrm{1}}}  ] $.
 \end{enumerate}
\end{lemma}

\begin{proof}
 By mutual induction on the derivations of the antecedent judgments in the five statements.
 To prove the case of typing judgment, we use \Zcref{claim:monotonicity} for the base case where $\bmttnt{M_{{\mathrm{1}}}}$ is variable.
\end{proof}

\begin{lemma}[Rebasing (Full Version)]
  Suppose $\Delta_{{\mathrm{1}}} = ( \Gamma_{{\mathrm{1}}} ^{\position{ \gamma_{{\mathrm{1}}} } }   \bmttsym{,}   \mathord{\blacktriangleleft} _{ \bmttnt{T_{{\mathrm{1}}}} }^{ \gamma_{{\mathrm{2}}} }   \bmttsym{,}   \tau_{{\mathrm{1}}} : \mathord{\blacktriangleright} ^{\binder{ \gamma_{{\mathrm{3}}} } \within \gamma_{{\mathrm{4}}} }   \bmttsym{,}  \Gamma_{{\mathrm{2}}})$, and $\Delta_{{\mathrm{2}}} =  \Gamma_{{\mathrm{1}}}  \bmttsym{,}  \Gamma_{{\mathrm{2}}} [ \binder{ \gamma_{{\mathrm{3}}} } \coloneqq \gamma_{{\mathrm{1}}} , \binder{ \tau_{{\mathrm{1}}} }  \coloneqq   \bmttnt{T_{{\mathrm{1}}}}  ] $. If $\Gamma_{{\mathrm{1}}}  \vdash  \gamma_{{\mathrm{4}}}  \preceq  \gamma_{{\mathrm{1}}}$, then the following statements hold,
  \begin{enumerate}
  \item $\vdash  \Delta_{{\mathrm{1}}}  \hasType \, \bmttkw{ctx} \implies \ \vdash  \Delta_{{\mathrm{2}}}  \hasType \, \bmttkw{ctx}$.
  \item $\Delta_{{\mathrm{1}}}  \vdash  \bmttnt{A}  \hasType \, \bmttkw{type} \implies \Delta_{{\mathrm{2}}}  \vdash   \bmttnt{A} [ \binder{ \gamma_{{\mathrm{3}}} } \coloneqq \gamma_{{\mathrm{1}}}  ]   \hasType \, \bmttkw{type}$.
  \item $\Delta_{{\mathrm{1}}}  \vdash  \delta_{{\mathrm{1}}}  \preceq  \delta_{{\mathrm{2}}} \implies \Delta_{{\mathrm{2}}}  \vdash   \delta_{{\mathrm{1}}} [ \binder{ \gamma_{{\mathrm{3}}} } \coloneqq \gamma_{{\mathrm{1}}}  ]   \preceq   \delta_{{\mathrm{2}}} [ \binder{ \gamma_{{\mathrm{3}}} } \coloneqq \gamma_{{\mathrm{1}}}  ] $.
  \item $ \Delta_{{\mathrm{1}}} \vdash \bmttnt{T_{{\mathrm{2}}}} : \delta_{{\mathrm{1}}} \sqsubseteq \delta_{{\mathrm{2}}}  \implies  \Delta_{{\mathrm{2}}} \vdash  \bmttnt{T_{{\mathrm{2}}}} [ \binder{ \tau_{{\mathrm{1}}} } \coloneqq \bmttnt{T_{{\mathrm{1}}}}  ]  :  \delta_{{\mathrm{1}}} [ \binder{ \gamma_{{\mathrm{3}}} } \coloneqq \gamma_{{\mathrm{1}}}  ]  \sqsubseteq  \delta_{{\mathrm{2}}} [ \binder{ \gamma_{{\mathrm{3}}} } \coloneqq \gamma_{{\mathrm{1}}}  ]  $.
  \item $\Delta_{{\mathrm{1}}}  \vdash  \bmttnt{M_{{\mathrm{1}}}}  \hasType  \bmttnt{A} \implies \Delta_{{\mathrm{2}}}  \vdash   \bmttnt{M_{{\mathrm{1}}}} [ \binder{ \gamma_{{\mathrm{3}}} } \coloneqq \gamma_{{\mathrm{1}}} , \binder{ \tau_{{\mathrm{1}}} }  \coloneqq   \bmttnt{T_{{\mathrm{1}}}}  ]   \hasType   \bmttnt{A} [ \binder{ \gamma_{{\mathrm{3}}} } \coloneqq \gamma_{{\mathrm{1}}}  ] $.
  \end{enumerate}
\end{lemma}

\begin{proof}
 By mutual induction on the first derivation of each statement.
\end{proof}

\begin{lemma}[Classifier Substitution (Full Version)]
 Suppose $\Delta_{{\mathrm{1}}} = \Gamma_{{\mathrm{1}}}  \bmttsym{,}   \binder{ \gamma_{{\mathrm{1}}} } \within \gamma_{{\mathrm{2}}}   \bmttsym{,}  \Gamma_{{\mathrm{2}}}$ and $\Delta_{{\mathrm{2}}} =  \Gamma_{{\mathrm{1}}}  \bmttsym{,}  \Gamma_{{\mathrm{2}}} [ \binder{ \gamma_{{\mathrm{1}}} } \coloneqq \gamma_{{\mathrm{3}}}  ] $. Given $\Gamma_{{\mathrm{1}}}  \vdash  \gamma_{{\mathrm{2}}}  \preceq  \gamma_{{\mathrm{3}}}$, then the following statements hold.
 \begin{enumerate}
  \item $\Delta_{{\mathrm{1}}}  \vdash  \bmttnt{A}  \hasType \, \bmttkw{type} \implies \Delta_{{\mathrm{2}}}  \vdash   \bmttnt{A} [ \binder{ \gamma_{{\mathrm{1}}} } \coloneqq \gamma_{{\mathrm{3}}}  ]   \hasType \, \bmttkw{type}$.
  \item $\vdash  \Delta_{{\mathrm{1}}}  \hasType \, \bmttkw{ctx} \implies \ \vdash  \Delta_{{\mathrm{2}}}  \hasType \, \bmttkw{ctx}$.
  \item $\Delta_{{\mathrm{1}}}  \vdash  \delta_{{\mathrm{1}}}  \preceq  \delta_{{\mathrm{2}}} \implies \Delta_{{\mathrm{2}}}  \vdash   \delta_{{\mathrm{1}}} [ \binder{ \gamma_{{\mathrm{1}}} } \coloneqq \gamma_{{\mathrm{3}}}  ]   \preceq   \delta_{{\mathrm{2}}} [ \binder{ \gamma_{{\mathrm{1}}} } \coloneqq \gamma_{{\mathrm{3}}}  ] $.
  \item $ \Delta_{{\mathrm{1}}} \vdash \bmttnt{T} : \delta_{{\mathrm{1}}} \sqsubseteq \delta_{{\mathrm{2}}}  \implies  \Delta_{{\mathrm{2}}} \vdash \bmttnt{T} :  \delta_{{\mathrm{1}}} [ \binder{ \gamma_{{\mathrm{1}}} } \coloneqq \gamma_{{\mathrm{3}}}  ]  \sqsubseteq  \delta_{{\mathrm{2}}} [ \binder{ \gamma_{{\mathrm{1}}} } \coloneqq \gamma_{{\mathrm{3}}}  ]  $.
  \item $\Delta_{{\mathrm{1}}}  \vdash  \bmttnt{M}  \hasType  \bmttnt{A} \implies \Delta_{{\mathrm{2}}}  \vdash   \bmttnt{M} [ \binder{ \gamma_{{\mathrm{1}}} } \coloneqq \gamma_{{\mathrm{3}}}  ]   \hasType   \bmttnt{A} [ \binder{ \gamma_{{\mathrm{1}}} } \coloneqq \gamma_{{\mathrm{3}}}  ] $.
 \end{enumerate}
\end{lemma}

\begin{proof}
 By mutual induction on the first derivation of each statement.
\end{proof}

\zref[claim]{claim:localsoundness}
\begin{proof}
 Easy to prove with \Zcref{claim:varsubst, claim:rebasing, claim:clssubst}.
\end{proof}

\zref[claim]{claim:localcompleteness}
\begin{proof}
 Easy to prove with \Zcref{claim:monotonicity}.
\end{proof}

\zref[claim]{claim:subjectreduction}
\begin{proof}
 By induction on the derivation of $ \bmttnt{M_{{\mathrm{1}}}}  \Rightarrow_{\beta}^{ \gamma }  \bmttnt{M_{{\mathrm{2}}}} $. For base cases, we apply \Zcref{claim:localsoundness}.
\end{proof}

To prove \Zcref{claim:SN}, we follow the proof strategy by \Textcite{book/MartiniM96}, who reduced strong normalization of the S4 modal lambda calculus to that of the simply typed lambda calculus~\cite{book/SorensenU2006}.

We first define a translation from types and terms of our calculus to those of STLC.

\begin{definition}
  $ \lvert  \mathord{-}  \rvert $ is a translation from types, contexts and terms of our calculus to those of STLC, which is defined as follows:
  \begin{align*}
     \lvert \bmttnt{p} \rvert  &= \bmttnt{p} \\
     \lvert \bmttnt{A}  \mathbin{\rightarrow}  \bmttnt{B} \rvert  &=  \lvert \bmttnt{A} \rvert   \mathbin{\rightarrow}   \lvert \bmttnt{B} \rvert  \\
     \lvert  \Box^{\mathord{\succeq}  \gamma }  \bmttnt{A}  \rvert  &=  \lvert \bmttnt{A} \rvert  \\
     \lvert  \forall  \gamma_{{\mathrm{1}}} \within  \gamma_{{\mathrm{2}}} . \bmttnt{A}  \rvert  &=  \lvert \bmttnt{A} \rvert  \\
    & \\
     \lvert  \varepsilon  \rvert  &=  \varepsilon  \\
     \lvert  \binder{ \bmttmv{x} }\has@{ \gamma } \bmttnt{A}   \bmttsym{,}  \Gamma \rvert  &= \bmttmv{x}  \hasType   \lvert \bmttnt{A} \rvert   \bmttsym{,}   \lvert \Gamma \rvert  \\
     \lvert  \tau : \mathord{\blacktriangleright} ^{\binder{ \gamma_{{\mathrm{1}}} } \within \gamma_{{\mathrm{2}}} }   \bmttsym{,}  \Gamma \rvert  &=  \lvert \Gamma \rvert  \\
     \lvert  \mathord{\blacktriangleleft} _{ \bmttnt{T} }^{ \gamma }   \bmttsym{,}  \Gamma \rvert  &=  \lvert \Gamma \rvert  \\
     \lvert  \binder{ \gamma_{{\mathrm{1}}} } \within \gamma_{{\mathrm{2}}}   \bmttsym{,}  \Gamma \rvert  &=  \lvert \Gamma \rvert  \\
    & \\
     \lvert \bmttmv{x} \rvert  &= \bmttmv{x} \\
     \lvert  \lambda \binder{ \bmttmv{x} }\has@{ \gamma } \bmttnt{A} \ldotp \bmttnt{M}  \rvert  &=  \lambda  \bmttmv{x} :  \lvert \bmttnt{A} \rvert   .   \lvert \bmttnt{M} \rvert   \\
     \lvert   \bmttnt{M_{{\mathrm{1}}}}   \bmttnt{M_{{\mathrm{2}}}}   \rvert  &=    \lvert \bmttnt{M_{{\mathrm{1}}}} \rvert   \    \lvert \bmttnt{M_{{\mathrm{2}}}} \rvert    \\
     \lvert  \mathbf{quo} (\binder{ \tau })\lbrace^{\binder{ \gamma_{{\mathrm{1}}} } \within \gamma_{{\mathrm{2}}} }  \bmttnt{M} \rbrace  \rvert  &=  \lvert \bmttnt{M} \rvert  \\
     \lvert  \mathbf{unq} _{ \bmttnt{T} }\lbrace^{ \gamma } \bmttnt{M} \rbrace  \rvert  &=  \lvert \bmttnt{M} \rvert  \\
     \lvert  \lambda \binder{ \gamma_{{\mathrm{1}}} }\within  \gamma_{{\mathrm{2}}} . \bmttnt{M}  \rvert  &=  \lvert \bmttnt{M} \rvert  \\
     \lvert  \bmttnt{M} \gamma  \rvert  &=  \lvert \bmttnt{M} \rvert 
  \end{align*}
\end{definition}

\begin{lemma}\label{claim:tost-preserves-typing}
  If\/ $\Gamma  \vdash  \bmttnt{M}  \hasType  \bmttnt{A}$, then $  \lvert \Gamma \rvert  \vdash^{\rightarrow}  \lvert \bmttnt{M} \rvert  \hasType  \lvert \bmttnt{A} \rvert  $.
\end{lemma}

\begin{proof}
  By induction on the derivation of $\Gamma  \vdash  \bmttnt{M}  \hasType  \bmttnt{A}$.
\end{proof}

\begin{lemma}\label{claim:tost-preserves-subst}
  $ \lvert  \bmttnt{M_{{\mathrm{1}}}} [ \binder{ \gamma_{{\mathrm{1}}} } \coloneqq \gamma_{{\mathrm{2}}} , \binder{ \bmttmv{x} }  \coloneqq   \bmttnt{M_{{\mathrm{2}}}}  ]  \rvert  =  \lvert \bmttnt{M_{{\mathrm{1}}}} \rvert   \bmttsym{[}  \bmttmv{x}  \coloneqq   \lvert \bmttnt{M_{{\mathrm{2}}}} \rvert   \bmttsym{]}$
\end{lemma}

\begin{proof}
  By induction on the structure of $\bmttnt{M_{{\mathrm{1}}}}$.
\end{proof}

\begin{lemma}\label{claim:tost-preserves-beta}
  If $ \bmttnt{M_{{\mathrm{1}}}}  \Rightarrow_{\beta}^{ \gamma }  \bmttnt{M_{{\mathrm{2}}}} $, then $ \lvert \bmttnt{M_{{\mathrm{1}}}} \rvert   \Rightarrow_{\beta}   \lvert \bmttnt{M_{{\mathrm{2}}}} \rvert $ or $ \lvert \bmttnt{M_{{\mathrm{1}}}} \rvert  =  \lvert \bmttnt{M_{{\mathrm{2}}}} \rvert $.
  In particular, if the redex is of the form $ \bmttsym{(}   \lambda \binder{ \bmttmv{x} }\has@{ \gamma } \bmttnt{A} \ldotp \bmttnt{M}   \bmttsym{)}   \bmttnt{N} $, then $ \lvert \bmttnt{M_{{\mathrm{1}}}} \rvert   \Rightarrow_{\beta}   \lvert \bmttnt{M_{{\mathrm{2}}}} \rvert $ and $ \lvert \bmttnt{M_{{\mathrm{1}}}} \rvert  =  \lvert \bmttnt{M_{{\mathrm{2}}}} \rvert $ otherwise.
\end{lemma}

\begin{proof}
  By induction on the structure of the evaluation context. In the base case, an ordinary \textbeta{}-redex maps to an STLC \textbeta{}-step by \Zcref{claim:tost-preserves-subst}, whereas the other two redexes erase to equality. The inductive cases follow by compatibility of STLC reduction.
\end{proof}

\zref[claim]{claim:SN}

\begin{proof}
  We argue by contradiction. Suppose that $\bmttnt{M}= \bmttnt{M_{{\mathrm{1}}}}$ begins an infinite reduction sequence
  \[
    \bmttnt{M_{{\mathrm{1}}}}  \Rightarrow_{\beta}^{ \gamma }  \bmttnt{M_{{\mathrm{2}}}}  \Rightarrow_{\beta}^{ \gamma }  \bmttnt{M_{{\mathrm{3}}}}  \Rightarrow_{\beta}^{ \gamma }  \cdots
  \]
  By \Zcref{claim:tost-preserves-typing}, $ \lvert \bmttnt{M_{{\mathrm{1}}}} \rvert $ is a well-typed term of the simply typed lambda calculus.

  By \Zcref{claim:tost-preserves-beta}, this sequence is mapped to
  \[
     \lvert \bmttnt{M_{{\mathrm{1}}}} \rvert   \Rightarrow_{\beta}^{?}   \lvert \bmttnt{M_{{\mathrm{2}}}} \rvert   \Rightarrow_{\beta}^{?}   \lvert \bmttnt{M_{{\mathrm{3}}}} \rvert   \Rightarrow_{\beta}^{?}  \cdots,
  \]
  where $ \Rightarrow_{\beta}^{?} $ denotes $( \Rightarrow_{\beta} ) \cup (=)$.
  We claim that this sequence contains infinitely many proper $ \Rightarrow_{\beta} $-steps.
  Indeed, suppose that it contains only finitely many such steps.
  Then there is some $\bmttnt{k}$ such that
  \[
     \lvert  { \bmttnt{M} }_{ \bmttnt{k} }  \rvert  =  \lvert  { \bmttnt{M} }_{ \bmttnt{k}  \bmttsym{+}  1 }  \rvert  =  \lvert  { \bmttnt{M} }_{ \bmttnt{k}  \bmttsym{+}  2 }  \rvert  = \cdots.
  \]
  Hence every step in the tail
  \[
     { \bmttnt{M} }_{ \bmttnt{k} }   \Rightarrow_{\beta}^{ \gamma }   { \bmttnt{M} }_{ \bmttnt{k}  \bmttsym{+}  1 }   \Rightarrow_{\beta}^{ \gamma }   { \bmttnt{M} }_{ \bmttnt{k}  \bmttsym{+}  2 }   \Rightarrow_{\beta}^{ \gamma }  \cdots
  \]
  contracts a redex of the form $ \bmttsym{(}   \lambda \binder{ \gamma_{{\mathrm{1}}} }\within  \gamma_{{\mathrm{2}}} . \bmttnt{N}   \bmttsym{)} \gamma_{{\mathrm{3}}} $ or
  $ \mathbf{unq} _{ \bmttnt{T} }\lbrace^{ \gamma_{{\mathrm{1}}} }  \mathbf{quo} (\binder{ \tau })\lbrace^{\binder{ \gamma_{{\mathrm{2}}} } \within \gamma_{{\mathrm{3}}} }  \bmttnt{N} \rbrace  \rbrace $.
  Each such contraction strictly decreases the number of non-STLC term constructors, where classifier and modal-transition annotations are ignored. Indeed, classifier \textbeta{}-reduction removes a classifier abstraction and application, while quote--unquote reduction removes a quotation and an unquotation.

  Therefore the translated sequence contains infinitely many proper $ \Rightarrow_{\beta} $-steps.
  By selecting those steps, we obtain an infinite $ \Rightarrow_{\beta} $-reduction sequence in STLC, contradicting the strong normalization of STLC~\cite{book/SorensenU2006}.
\end{proof}

\zref[claim]{claim:confluence}

\begin{proof}
  From Newman's lemma~\cite{book/SorensenU2006} and \Zcref{claim:SN}, we need
  only to prove weak confluence.
  Weak confluence follows by a standard analysis of two competing one-step reductions. The reduction rules are left-linear and have no critical overlaps, while reductions in disjoint subterms commute.
\end{proof}

\begin{lemma}\label{claim:neutral-terms}
  Suppose\/ $Γ ⊢ M \has A$.
  If $M$ is $\beta$-normal and neutral,
  then there exists some $\var(x : B)@{γ} \in Γ$ such that
  $x \in \FV(M)$ and
  $A$ is a subformula of~$B$.
\end{lemma}

\begin{proof}
  By induction on derivation.
  For the last rule of the derivation, there are four possibilities:
  \begin{case}[\ref{rule:type-var}]
    Obvious.
  \end{case}
  \begin{case}[\ref{rule:type-to-e}]
    Assume
    \begin{prooftree*}
      \hypo{Γ ⊢ N \has B \to C}
      \hypo{Γ ⊢ P \has B}
      \infer2{Γ ⊢ N P \has C}
    \end{prooftree*}
    Since $N P$ is $\beta$-normal, $N$ is $\beta$-normal and neutral.
    By the IH there exists some $\var(x : D)@{δ} \in Γ$ such that
    $x \in \FV(N)$ and
    $B \to C$ is a subformula of~$D$.
    Then $x \in \FV(N P)$ and $C$ is also a subformula of~$D$, hence
    $\var(x : D)@{δ}$ meets the condition.
  \end{case}
  \begin{case}[\ref{rule:type-bm-e}]
    Assume
    \begin{prooftree*}
      \hypo{Γ, ^{γ} ⊢ N \has □^{⪰ γ_1} B}
      \hypo{Γ ⊢ γ_1 ⪯ \pos(Γ)}
      \infer2{Γ ⊢ \unq@{γ}{N} \has B}
    \end{prooftree*}
    Since $\unq@{γ}{N}$ is $\beta$-normal, $N$ is $\beta$-normal and neutral.
    By the IH there exists some $\var(x : D)@{δ} \in Γ, ^{γ}$ such that
    $x \in \FV(N)$ and
    $□^{⪰ γ_1} B$ is a subformula of~$D$.
    Then $x \in \FV(\unq@{γ}{N})$ and 
    $B$ is also a subformula of~$D$, hence
    $\var(x : D)@{δ}$ meets the condition.
  \end{case}
  \begin{case}[\ref{rule:type-polycls-e}]
    Assume
    \begin{prooftree*}
      \hypo{Γ ⊢ N \has ∀\cls(γ_2 i> γ_1). B}
      \hypo{Γ ⊢ γ_1 ⪯ γ}
      \infer2{Γ ⊢ Nγ \has B\w[γ_2 / γ]}
    \end{prooftree*}
    Since $Nγ$ is $\beta$-normal, $N$ is $\beta$-normal and neutral.
    By the IH there exists some $\var(x : D)@{δ} \in Γ$ such that
    $x \in \FV(N)$ and
    $∀\cls(γ_2 i> γ_1). B$ is a subformula of~$D$.
    Then $x \in \FV(Nγ)$ and 
    $B\w[γ_2 / γ]$ is also a subformula of~$D$, hence
    $\var(x : D)@{δ}$ meets the condition.
    Notice that classifier renaming~$\w[γ_2 / γ]$ here
    is allowed in the definition of subformula.
    \qedhere
  \end{case}
\end{proof}

\zref[claim]{claim:canonicity}

\begin{proof}
  If not canonical,
  by \Zcref{claim:neutral-terms} it contains a free variable, which
  contradicts the assumption.
\end{proof}

\zref[claim]{claim:subformula-property}

\begin{proof}
  By induction on derivation.
  Since the term $M$ itself clearly
  satisfies~\zcref{item:subformula-property:intro},
  it suffices to check condition for proper subterms.
  \begin{case}[\ref{rule:type-var}]
    No proper subterm exists.
  \end{case}
  \begin{case}[\ref{rule:type-to-i}]
    Assume
    \begin{prooftree*}
      \hypo{Γ, \var(y : B)@{γ} ⊢ N \has C}
      \hypo{γ \notin \FC(C)}
      \infer2{Γ ⊢ λ\var(y : B)@{γ}. N \has B \to C}
    \end{prooftree*}
    Since all proper subterms of $λ\var(y : B)@{γ}. N$ are a subterm of $N$,
    by the IH there are three possibilities for their types:
    \begin{subcase}[a]
      A subformula of $C$.
      Then it is also a subformula of $B \to C$, so
      \zcref{item:subformula-property:intro} holds.
    \end{subcase}
    \begin{subcase}[b]
      A subformula of $B$.
      This is also the case \zcref{item:subformula-property:intro}.
    \end{subcase}
    \begin{subcase}[c]
      A subformula of $D$ for some $\var(x : D)@{δ} \in Γ$.
      Yields \zcref{item:subformula-property:elim}.
    \end{subcase}
  \end{case}
  \begin{case}[\ref{rule:type-to-e}]\label{case:subformula-property:to-elim}
    Assume
    \begin{prooftree*}
      \hypo{Γ ⊢ N \has B \to C}
      \hypo{Γ ⊢ P \has B}
      \infer2{Γ ⊢ N P \has C}
    \end{prooftree*}
    Then $N$ is $\beta$-normal and neutral, so
    applying \Zcref{claim:neutral-terms} to $N$,
    we see that there exists some $\var(x : D)@{δ} \in Γ$
    such that $B \to C$ is a subformula of $D$.
    Together with the IH,
    we see that \zcref{item:subformula-property:elim} holds
    for any proper subterm of $N P$.
  \end{case}
  \begin{case}[\ref{rule:type-bm-i}]
    Follows from the IH\@.
  \end{case}
  \begin{case}[\ref{rule:type-bm-e}]
    Similar to the $\to$-E case.
  \end{case}
  \begin{case}[\ref{rule:type-polycls-i}]
    Follows from the IH\@.
  \end{case}
  \begin{case}[\ref{rule:type-polycls-e}]
    Similar to the $\to$-E case.
    \qedhere
  \end{case}
\end{proof}

\section{Full Proofs for \Zcref{sec:staging}}
\RestateClaim{claim:staged-reduction-is-beta-reduction}
\begin{proof}
  By induction on the structure of evaluation contexts.
\end{proof}

\RestateClaim{claim:staged-reduction-preservation}
\begin{proof}
  This follows directly from \Zcref{claim:subjectreduction, claim:staged-reduction-is-beta-reduction}.
\end{proof}

\begin{definition}
  $ { \Gamma }^{ \bmttnt{T_{{\mathrm{1}}}} }_{[  \bmttnt{T_{{\mathrm{2}}}}  ]} $ is a staged typing context which starts at the stage $\bmttnt{T_{{\mathrm{1}}}}$ and ends at the stage $\bmttnt{T_{{\mathrm{2}}}}$.
  \par
  \smallskip
  \begin{tabular}{lcl}
    $ { \Gamma }^{ \bmttnt{T_{{\mathrm{1}}}} }_{[  \bmttnt{T_{{\mathrm{2}}}}  ]} $ & $\Coloneqq$ & $ \varepsilon  \text{\ (if $\bmttnt{T_{{\mathrm{1}}}} = \bmttnt{T_{{\mathrm{2}}}}$)} \mid  { \Gamma }^{ \bmttnt{T_{{\mathrm{1}}}} }_{[  \bmttnt{T_{{\mathrm{2}}}}  ]}   \bmttsym{,}   \binder{ \bmttmv{x} }\has@{ \gamma } \bmttnt{A}  \mid  { \Gamma }^{ \bmttnt{T_{{\mathrm{1}}}} }_{[  \bmttnt{T_{{\mathrm{3}}}}  ]}   \bmttsym{,}   \tau : \mathord{\blacktriangleright} ^{\binder{ \gamma_{{\mathrm{1}}} } \within \gamma_{{\mathrm{2}}} } \text{\ (if $\bmttnt{T_{{\mathrm{2}}}} = \bmttnt{T_{{\mathrm{3}}}}  \bmttsym{,}  \tau$)}$ \\
    &$\mid$& $ { \Gamma }^{ \bmttnt{T_{{\mathrm{1}}}} }_{[   \bmttnt{T_{{\mathrm{2}}}}  +  \bmttnt{T_{{\mathrm{3}}}}   ]}   \bmttsym{,}   \mathord{\blacktriangleleft} _{ \bmttnt{T_{{\mathrm{3}}}} }^{ \gamma }  \mid  { \Gamma }^{ \bmttnt{T_{{\mathrm{1}}}} }_{[  \bmttnt{T_{{\mathrm{2}}}}  ]}   \bmttsym{,}   \binder{ \gamma_{{\mathrm{1}}} } \within \gamma_{{\mathrm{2}}} $
  \end{tabular}
\end{definition}

\begin{definition}
  $ \mathbf{Dom}_{\token{V} }^{ \bmttnt{T_{{\mathrm{1}}}} }(  { \Gamma }^{ \bmttnt{T_{{\mathrm{2}}}} }_{[  \bmttnt{T_{{\mathrm{3}}}}  ]}  ) $ represents the set of variables that appear at the stage $\bmttnt{T_{{\mathrm{1}}}}$ in the staged typing context $ { \Gamma }^{ \bmttnt{T_{{\mathrm{2}}}} }_{[  \bmttnt{T_{{\mathrm{3}}}}  ]} $.
  \begin{align*}
     \mathbf{Dom}_{\token{V} }^{ \bmttnt{T_{{\mathrm{1}}}} }(  \varepsilon  )  &=  \emptyset  \\
     \mathbf{Dom}_{\token{V} }^{ \bmttnt{T_{{\mathrm{1}}}} }(  { \Gamma }^{ \bmttnt{T_{{\mathrm{2}}}} }_{[  \bmttnt{T_{{\mathrm{3}}}}  ]}   \bmttsym{,}   \binder{ \bmttmv{x} }\has@{ \gamma } \bmttnt{A}  )  &= \begin{cases}
        \mathbf{Dom}_{\token{V} }^{ \bmttnt{T_{{\mathrm{1}}}} }(  { \Gamma }^{ \bmttnt{T_{{\mathrm{2}}}} }_{[  \bmttnt{T_{{\mathrm{3}}}}  ]}  )  \cup  \{ \bmttmv{x} \}   \text{\ if $\bmttnt{T_{{\mathrm{1}}}} = \bmttnt{T_{{\mathrm{3}}}}$} \\
       \mathbf{Dom}_{\token{V} }^{ \bmttnt{T_{{\mathrm{1}}}} }(  { \Gamma }^{ \bmttnt{T_{{\mathrm{2}}}} }_{[  \bmttnt{T_{{\mathrm{3}}}}  ]}  )  \text{\ otherwise}
    \end{cases} \\
     \mathbf{Dom}_{\token{V} }^{ \bmttnt{T_{{\mathrm{1}}}} }(  { \Gamma }^{ \bmttnt{T_{{\mathrm{2}}}} }_{[  \bmttnt{T_{{\mathrm{3}}}}  ]}   \bmttsym{,}   \tau : \mathord{\blacktriangleright} ^{\binder{ \gamma_{{\mathrm{1}}} } \within \gamma_{{\mathrm{2}}} }  )  &=  \mathbf{Dom}_{\token{V} }^{ \bmttnt{T_{{\mathrm{1}}}} }(  { \Gamma }^{ \bmttnt{T_{{\mathrm{2}}}} }_{[  \bmttnt{T_{{\mathrm{3}}}}  ]}  )  \\
     \mathbf{Dom}_{\token{V} }^{ \bmttnt{T_{{\mathrm{1}}}} }(  { \Gamma }^{ \bmttnt{T_{{\mathrm{2}}}} }_{[  \bmttnt{T_{{\mathrm{3}}}}  ]}   \bmttsym{,}   \mathord{\blacktriangleleft} _{ \bmttnt{T_{{\mathrm{4}}}} }^{ \gamma }  )  &=  \mathbf{Dom}_{\token{V} }^{ \bmttnt{T_{{\mathrm{1}}}} }(  { \Gamma }^{ \bmttnt{T_{{\mathrm{2}}}} }_{[  \bmttnt{T_{{\mathrm{3}}}}  ]}  )  \\
     \mathbf{Dom}_{\token{V} }^{ \bmttnt{T_{{\mathrm{1}}}} }(  { \Gamma }^{ \bmttnt{T_{{\mathrm{2}}}} }_{[  \bmttnt{T_{{\mathrm{3}}}}  ]}   \bmttsym{,}   \binder{ \gamma_{{\mathrm{1}}} } \within \gamma_{{\mathrm{2}}}  )  &=  \mathbf{Dom}_{\token{V} }^{ \bmttnt{T_{{\mathrm{1}}}} }(  { \Gamma }^{ \bmttnt{T_{{\mathrm{2}}}} }_{[  \bmttnt{T_{{\mathrm{3}}}}  ]}  ) 
  \end{align*}
\end{definition}

\begin{lemma}\label{claim:init-cls-inversion}
  \begin{enumerate}
    \item\label{item:init-cls-inversion:scope} If\/ $\Gamma  \vdash  \gamma  \preceq   \mathord{\boldsymbol{!} } $, then $\gamma =  \exclam $.
    \item\label{item:init-cls-inversion:stage} If\/ $ \Gamma \vdash \bmttnt{T} : \gamma \sqsubseteq  \mathord{\boldsymbol{!} }  $, then $\gamma =  \exclam $.
  \end{enumerate}
\end{lemma}

\begin{proof}
  \begin{description}
    \item[\Zcref{item:init-cls-inversion:scope}] By induction on the derivation of $\Gamma  \vdash  \gamma  \preceq   \mathord{\boldsymbol{!} } $.
      \begin{description}
        \item[Case \textmd{\Zcref{rule:itrans-refl}}:] The rule requires that $\gamma =  \exclam $.
        \item[Case \textmd{$ \preceq $-Trans}:] We have $\Gamma  \vdash  \gamma  \preceq  \gamma'$ and $\Gamma  \vdash  \gamma'  \preceq   \mathord{\boldsymbol{!} } $ for some $\gamma'$. By induction hypothesis, $\gamma' =  \exclam $. Hence, $\gamma =  \exclam $ by induction hypothesis.
      \end{description}
    \item[\Zcref{item:init-cls-inversion:stage}] By induction on the derivation of $ \Gamma \vdash \bmttnt{T} : \gamma \sqsubseteq  \mathord{\boldsymbol{!} }  $.
      \begin{description}
        \item[Case \textmd{\Zcref{rule:mtrans-lift}}:] We have $\Gamma  \vdash  \gamma  \preceq   \mathord{\boldsymbol{!} } $. By \Zcref{item:init-cls-inversion:scope}, we have $\gamma =  \exclam $.
        \item[Case \textmd{$ \sqsubseteq $-Trans}:] We have $ \Gamma \vdash \bmttnt{T_{{\mathrm{1}}}} : \gamma \sqsubseteq \gamma' $ and $ \Gamma \vdash \bmttnt{T_{{\mathrm{2}}}} : \gamma' \sqsubseteq  \mathord{\boldsymbol{!} }  $ for some $\gamma'$, $\bmttnt{T_{{\mathrm{1}}}}$ and $\bmttnt{T_{{\mathrm{2}}}}$. By induction hypothesis, $\gamma' =  \exclam $. Hence, $\gamma =  \exclam $ by induction hypothesis.
      \end{description}
  \end{description}
\end{proof}

\begin{lemma}\label{claim:top-level-context-position}
  If\/ $\vdash   { \Gamma }^{  \varepsilon  }_{[   \varepsilon   ]}   \hasType \, \bmttkw{ctx}$ and $ \mathbf{Dom}_{\token{V} }^{  \varepsilon  }(  { \Gamma }^{  \varepsilon  }_{[   \varepsilon   ]}  )  =  \emptyset $, then $ \mathrm{pos} (  { \Gamma }^{  \varepsilon  }_{[   \varepsilon   ]}  )  =  \exclam $.
\end{lemma}

\begin{proof}
  By induction on the structure of $ { \Gamma }^{  \varepsilon  }_{[   \varepsilon   ]} $.

  \begin{description}
    \item[Case \textmd{$ { \Gamma }^{  \varepsilon  }_{[   \varepsilon   ]}  =  \varepsilon $}:]
      Immediate from the definition of $ \mathrm{pos} (  \varepsilon  ) $.

    \item[Case \textmd{$ { \Gamma }^{  \varepsilon  }_{[   \varepsilon   ]}  =  { \Gamma_{{\mathrm{1}}} }^{  \varepsilon  }_{[   \varepsilon   ]}   \bmttsym{,}   \binder{ \bmttmv{x} }\has@{ \gamma_{{\mathrm{1}}} } \bmttnt{A} $}:]
      This case is impossible since $ \mathbf{Dom}_{\token{V} }^{  \varepsilon  }(  { \Gamma }^{  \varepsilon  }_{[   \varepsilon   ]}  )  =  \emptyset $.

    \item[Case \textmd{$ { \Gamma }^{  \varepsilon  }_{[   \varepsilon   ]}  =  { \Gamma_{{\mathrm{1}}} }^{  \varepsilon  }_{[   \varepsilon   ]}   \bmttsym{,}   \binder{ \gamma_{{\mathrm{1}}} } \within \gamma_{{\mathrm{2}}} $}:]
      Since classifier declarations do not change the current classifier,
      the result follows immediately from the induction hypothesis for
      $ { \Gamma_{{\mathrm{1}}} }^{  \varepsilon  }_{[   \varepsilon   ]} $.

    \item[Case \textmd{$ { \Gamma }^{  \varepsilon  }_{[   \varepsilon   ]}  =  { \Gamma_{{\mathrm{1}}} }^{  \varepsilon   \position{ \gamma_{{\mathrm{1}}} } }_{[   \varepsilon   ]}   \bmttsym{,}   \mathord{\blacktriangleleft} _{  \varepsilon  }^{ \gamma_{{\mathrm{2}}} } $}:]
      Inverting well-formedness gives
      $  { \Gamma_{{\mathrm{1}}} }^{  \varepsilon  }_{[   \varepsilon   ]}  \vdash  \varepsilon  : \gamma_{{\mathrm{2}}} \sqsubseteq \gamma_{{\mathrm{1}}} $.
      The induction hypothesis gives $\gamma_{{\mathrm{1}}} =  \exclam $; hence
      $  { \Gamma_{{\mathrm{1}}} }^{  \varepsilon  }_{[   \varepsilon   ]}  \vdash  \varepsilon  : \gamma_{{\mathrm{2}}} \sqsubseteq  \mathord{\boldsymbol{!} }  $.
      By \Zcref{item:init-cls-inversion:stage}, $\gamma_{{\mathrm{2}}} =  \exclam $.
      Hence, $ \mathrm{pos} (  { \Gamma }^{  \varepsilon  }_{[   \varepsilon   ]}  )  = \gamma_{{\mathrm{2}}} =  \exclam $.

    \item[Case \textmd{$ { \Gamma }^{  \varepsilon  }_{[   \varepsilon   ]}  =  {  { \Gamma_{{\mathrm{1}}} }^{  \varepsilon  }_{[   \varepsilon   ]}   \bmttsym{,}   \tau_{{\mathrm{1}}} : \mathord{\blacktriangleright} ^{\binder{ \gamma_{{\mathrm{1}}} } \within \gamma_{{\mathrm{2}}} }   \bmttsym{,}  \Gamma_{{\mathrm{2}}} }^{  \varepsilon   \position{ \gamma_{{\mathrm{3}}} } }_{[  \bmttnt{T}  ]}   \bmttsym{,}   \mathord{\blacktriangleleft} _{ \tau_{{\mathrm{1}}}  \bmttsym{,}  \bmttnt{T} }^{ \gamma_{{\mathrm{4}}} } $}:]
      Inverting well-formedness gives
      \[
          {  { \Gamma_{{\mathrm{1}}} }^{  \varepsilon  }_{[   \varepsilon   ]}   \bmttsym{,}   \tau_{{\mathrm{1}}} : \mathord{\blacktriangleright} ^{\binder{ \gamma_{{\mathrm{1}}} } \within \gamma_{{\mathrm{2}}} }   \bmttsym{,}  \Gamma_{{\mathrm{2}}} }^{  \varepsilon   \position{ \gamma_{{\mathrm{3}}} } }_{[  \bmttnt{T}  ]}  \vdash \tau_{{\mathrm{1}}}  \bmttsym{,}  \bmttnt{T} : \gamma_{{\mathrm{4}}} \sqsubseteq \gamma_{{\mathrm{3}}} .
      \]
      The induction hypothesis for $ { \Gamma_{{\mathrm{1}}} }^{  \varepsilon  }_{[   \varepsilon   ]} $ gives
      $ \mathrm{pos} (  { \Gamma_{{\mathrm{1}}} }^{  \varepsilon  }_{[   \varepsilon   ]}  )  =  \exclam $.  Inverting the displayed modal-transition
      derivation at the first witness $\tau_{{\mathrm{1}}}$, and using
      \ref{rule:calc-mtrans-open}, therefore yields an empty-witness prefix
      \[
          {  { \Gamma_{{\mathrm{1}}} }^{  \varepsilon  }_{[   \varepsilon   ]}   \bmttsym{,}   \tau_{{\mathrm{1}}} : \mathord{\blacktriangleright} ^{\binder{ \gamma_{{\mathrm{1}}} } \within \gamma_{{\mathrm{2}}} }   \bmttsym{,}  \Gamma_{{\mathrm{2}}} }^{  \varepsilon   \position{ \gamma_{{\mathrm{3}}} } }_{[  \bmttnt{T}  ]}  \vdash  \varepsilon  : \gamma_{{\mathrm{4}}} \sqsubseteq  \mathord{\boldsymbol{!} }  .
      \]
      By \Zcref{item:init-cls-inversion:stage}, $\gamma_{{\mathrm{4}}} =  \exclam $.
      Hence, $ \mathrm{pos} (  { \Gamma }^{  \varepsilon  }_{[   \varepsilon   ]}  )  = \gamma_{{\mathrm{4}}} =  \exclam $.
  \end{description}
\end{proof}

\begin{lemma}[Decomposition]\label{claim:staged-reduction-decomposition}
  \begin{enumerate}
    \item\label{item:staged-reduction-decomposition:runtime}   If\/ $ { \Gamma }^{  \varepsilon  }_{[   \varepsilon   ]}   \vdash   \bmttnt{M_{{\mathrm{1}}}} ^{  \varepsilon  }   \hasType  \bmttnt{A}$ and\/ $ \mathbf{Dom}_{\token{V} }^{  \varepsilon  }(  { \Gamma }^{  \varepsilon  }_{[   \varepsilon   ]}  )  =  \emptyset $, then there exists a staged evaluation context $ \mathcal{E} ^{  \varepsilon  @  \mathord{\boldsymbol{!} }  }_{[  \bmttnt{T} @ \gamma  ]} $ and a staged redex $ \bmttnt{M_{{\mathrm{2}}}} ^{ \bmttnt{T} } $ such that $ \bmttnt{M_{{\mathrm{1}}}} ^{  \varepsilon  }  =  \mathcal{E} ^{  \varepsilon  @  \mathord{\boldsymbol{!} }  }_{[  \bmttnt{T} @ \gamma  ]}   \bmttsym{[}   \bmttnt{M_{{\mathrm{2}}}} ^{ \bmttnt{T} }   \bmttsym{]}$, or $ \bmttnt{M_{{\mathrm{1}}}} ^{  \varepsilon  } $ is a value.
    \item\label{item:staged-reduction-decomposition:future}   If\/ $ { \Gamma }^{  \varepsilon   \position{ \gamma_{{\mathrm{1}}} } }_{[  \tau_{{\mathrm{1}}}  \bmttsym{,}  \bmttnt{T_{{\mathrm{1}}}}  ]}   \vdash   \bmttnt{M_{{\mathrm{1}}}} ^{ \tau_{{\mathrm{1}}}  \bmttsym{,}  \bmttnt{T_{{\mathrm{1}}}} }   \hasType  \bmttnt{A}$ and\/ $ \mathbf{Dom}_{\token{V} }^{  \varepsilon  }(  { \Gamma }^{  \varepsilon   \position{ \gamma_{{\mathrm{1}}} } }_{[  \tau_{{\mathrm{1}}}  \bmttsym{,}  \bmttnt{T_{{\mathrm{1}}}}  ]}  )  =  \emptyset $, then there exists a staged evaluation context $ \mathcal{E} ^{ \tau_{{\mathrm{1}}}  \bmttsym{,}  \bmttnt{T_{{\mathrm{1}}}} @ \gamma_{{\mathrm{1}}} }_{[  \bmttnt{T_{{\mathrm{2}}}} @ \gamma_{{\mathrm{2}}}  ]} $ and a staged redex $ \bmttnt{M_{{\mathrm{2}}}} ^{ \bmttnt{T_{{\mathrm{2}}}} } $ such that $ \bmttnt{M_{{\mathrm{1}}}} ^{ \tau_{{\mathrm{1}}}  \bmttsym{,}  \bmttnt{T_{{\mathrm{1}}}} }  =  \mathcal{E} ^{ \tau_{{\mathrm{1}}}  \bmttsym{,}  \bmttnt{T_{{\mathrm{1}}}} @ \gamma_{{\mathrm{1}}} }_{[  \bmttnt{T_{{\mathrm{2}}}} @ \gamma_{{\mathrm{2}}}  ]}   \bmttsym{[}   \bmttnt{M_{{\mathrm{2}}}} ^{ \bmttnt{T_{{\mathrm{2}}}} }   \bmttsym{]}$, or $ \bmttnt{M_{{\mathrm{1}}}} ^{ \tau_{{\mathrm{1}}}  \bmttsym{,}  \bmttnt{T_{{\mathrm{1}}}} }  =  \bmttnt{M_{{\mathrm{3}}}} ^{ \bmttnt{T_{{\mathrm{1}}}} } $ for some term $ \bmttnt{M_{{\mathrm{3}}}} ^{ \bmttnt{T_{{\mathrm{1}}}} } $.
  \end{enumerate}
\end{lemma}

\begin{proof}
  We prove the two statements simultaneously, by structural induction on
  $ \bmttnt{M_{{\mathrm{1}}}} ^{  \varepsilon  } $ and $ \bmttnt{M_{{\mathrm{1}}}} ^{ \tau_{{\mathrm{1}}}  \bmttsym{,}  \bmttnt{T_{{\mathrm{1}}}} } $, respectively.

  The only point worth noting is that the variable case cannot occur in
  \Zcref{item:staged-reduction-decomposition:runtime}.  By
  \Zcref{claim:top-level-context-position}, $ \mathrm{pos} (  { \Gamma }^{  \varepsilon  }_{[   \varepsilon   ]}  )  =  \exclam $.
  Inversion of \ref{rule:type-var} would therefore require a declaration
  $ \binder{ \bmttmv{x} }\has@{ \gamma } \bmttnt{A}  \, \in \,  { \Gamma }^{  \varepsilon  }_{[   \varepsilon   ]} $ together with $ { \Gamma }^{  \varepsilon  }_{[   \varepsilon   ]}   \vdash  \gamma  \preceq   \mathord{\boldsymbol{!} } $.
  By \Zcref{claim:init-cls-inversion}, this forces $\gamma =  \exclam $,
  which is incompatible with the construction of a well-formed staged context
  at stage $ \varepsilon $.  All remaining cases follow directly from the induction hypotheses and the definition of staged evaluation contexts.
\end{proof}

\RestateClaim{claim:staged-reduction-progress}
\begin{proof}
  This follows directly from \Zcref{claim:staged-reduction-decomposition}.
\end{proof}

\begin{lemma}\label{claim:safety-offline-code-generation-aux}
  If\/ $ {  \tau : \mathord{\blacktriangleright} ^{\binder{ \gamma_{{\mathrm{1}}} } \within  \mathord{\boldsymbol{!} }  }   \bmttsym{,}  \Gamma }^{  \varepsilon  }_{[  \bmttnt{T}  ]}   \vdash   \bmttnt{M} ^{ \bmttnt{T} }   \hasType  \bmttnt{A}$,
  then $  { \Gamma }^{  \varepsilon  }_{[  \bmttnt{T}  ]}  [ \binder{ \gamma_{{\mathrm{1}}} } \coloneqq  \mathord{\boldsymbol{!} }   ]   \vdash    \bmttnt{M} ^{ \bmttnt{T} }  [ \binder{ \gamma_{{\mathrm{1}}} } \coloneqq  \mathord{\boldsymbol{!} }   ]   \hasType   \bmttnt{A} [ \binder{ \gamma_{{\mathrm{1}}} } \coloneqq  \mathord{\boldsymbol{!} }   ] $.
\end{lemma}

\begin{proof}
  By induction on the derivation of $ {  \tau : \mathord{\blacktriangleright} ^{\binder{ \gamma_{{\mathrm{1}}} } \within  \mathord{\boldsymbol{!} }  }   \bmttsym{,}  \Gamma }^{  \varepsilon  }_{[  \bmttnt{T}  ]}   \vdash   \bmttnt{M} ^{ \bmttnt{T} }   \hasType  \bmttnt{A}$.
\end{proof}

\RestateClaim{claim:safety-of-offline-code-generation}
\begin{proof}
  By \Zcref{claim:staged-reduction-is-beta-reduction} and
  \Zcref{claim:SN}, staged reduction from $\bmttnt{M_{{\mathrm{1}}}}$ terminates.
  Let $v$ be a normal form such that
  $ \bmttnt{M_{{\mathrm{1}}}}  \Rightarrow_{st*}^{  \mathord{\boldsymbol{!} }  }  v $.
  By repeated application of
  \Zcref{claim:staged-reduction-preservation}, we have
  \[
     {  \tau_{{\mathrm{1}}} : \mathord{\blacktriangleright} ^{\binder{ \gamma_{{\mathrm{1}}} } \within  \mathord{\boldsymbol{!} }  }   \bmttsym{,}  \Gamma }^{  \varepsilon   \position{ \gamma_{{\mathrm{2}}} } }_{[   \varepsilon   ]}   \bmttsym{,}   \mathord{\blacktriangleleft} _{ \tau_{{\mathrm{1}}} }^{  \mathord{\boldsymbol{!} }  }   \vdash   v ^{  \varepsilon  }   \hasType   \Box^{\mathord{\succeq}  \gamma_{{\mathrm{2}}} }  \bmttnt{A} .
  \]
  The context on the left has no variables at stage $ \varepsilon $.
  Since $v$ cannot take a staged-reduction step,
  \Zcref{claim:staged-reduction-progress} implies that it is a value.

  By inspection of the forms of values and inversion of this typing
  judgment, $v$ must be a quotation.  Thus, for some
  $\tau_{{\mathrm{2}}}$, $\gamma_{{\mathrm{3}}}$, and $\bmttnt{M_{{\mathrm{2}}}}$,
  \[
    v =  \mathbf{quo} (\binder{ \tau_{{\mathrm{2}}} })\lbrace^{\binder{ \gamma_{{\mathrm{3}}} } \within \gamma_{{\mathrm{2}}} }   \bmttnt{M_{{\mathrm{2}}}} ^{  \varepsilon  }  \rbrace .
  \]
  Inverting the typing derivation further yields
  \[
     {  \tau_{{\mathrm{1}}} : \mathord{\blacktriangleright} ^{\binder{ \gamma_{{\mathrm{1}}} } \within  \mathord{\boldsymbol{!} }  }   \bmttsym{,}  \Gamma }^{  \varepsilon   \position{ \gamma_{{\mathrm{2}}} } }_{[   \varepsilon   ]}   \bmttsym{,}   \mathord{\blacktriangleleft} _{ \tau_{{\mathrm{1}}} }^{  \mathord{\boldsymbol{!} }  }   \bmttsym{,}   \tau_{{\mathrm{2}}} : \mathord{\blacktriangleright} ^{\binder{ \gamma_{{\mathrm{3}}} } \within \gamma_{{\mathrm{2}}} }   \vdash   \bmttnt{M_{{\mathrm{2}}}} ^{  \varepsilon  }   \hasType  \bmttnt{A},
  \]
  with $\gamma_{{\mathrm{3}}} \, \notin \,  \token{FC}( \bmttnt{A} ) $.
  By reflexivity, the prefix ending in $ { \Gamma }^{  \varepsilon   \position{ \gamma_{{\mathrm{2}}} } }_{[   \varepsilon   ]} $
  derives $\gamma_{{\mathrm{2}}} \, \preceq \, \gamma_{{\mathrm{2}}}$.  Hence, applying
  \Zcref{claim:rebasing} to the adjacent $ \mathord{\blacktriangleleft} _{ \tau_{{\mathrm{1}}} }^{  \mathord{\boldsymbol{!} }  } $ and
  $ \tau_{{\mathrm{2}}} : \mathord{\blacktriangleright} ^{\binder{ \gamma_{{\mathrm{3}}} } \within \gamma_{{\mathrm{2}}} } $ gives
  \[
     {  \tau_{{\mathrm{1}}} : \mathord{\blacktriangleright} ^{\binder{ \gamma_{{\mathrm{1}}} } \within  \mathord{\boldsymbol{!} }  }   \bmttsym{,}  \Gamma }^{  \varepsilon   \position{ \gamma_{{\mathrm{2}}} } }_{[   \varepsilon   ]}   \vdash    \bmttnt{M_{{\mathrm{2}}}} ^{  \varepsilon  }  [ \binder{ \gamma_{{\mathrm{3}}} } \coloneqq \gamma_{{\mathrm{2}}}  ]   \hasType  \bmttnt{A}.
  \]
  Here the substitution of $\tau_{{\mathrm{1}}}$ for $\tau_{{\mathrm{2}}}$ is vacuous because
  $ \bmttnt{M_{{\mathrm{2}}}} ^{  \varepsilon  } $ contains no free occurrence of $\tau_{{\mathrm{2}}}$, and
  $ \bmttnt{A} [ \binder{ \gamma_{{\mathrm{3}}} } \coloneqq \gamma_{{\mathrm{2}}}  ]  = \bmttnt{A}$ by the side condition above.
  Applying \Zcref{claim:safety-offline-code-generation-aux} to this judgment gives
  \[
      { \Gamma }^{  \varepsilon   \position{ \gamma_{{\mathrm{2}}} } }_{[   \varepsilon   ]}  [ \binder{ \gamma_{{\mathrm{1}}} } \coloneqq  \mathord{\boldsymbol{!} }   ]   \vdash     \bmttnt{M_{{\mathrm{2}}}} ^{  \varepsilon  }  [ \binder{ \gamma_{{\mathrm{3}}} } \coloneqq \gamma_{{\mathrm{2}}}  ]  [ \binder{ \gamma_{{\mathrm{1}}} } \coloneqq  \mathord{\boldsymbol{!} }   ]   \hasType   \bmttnt{A} [ \binder{ \gamma_{{\mathrm{1}}} } \coloneqq  \mathord{\boldsymbol{!} }   ] ,
  \]
  as required.
\end{proof}

\section{Full Contents for \Zcref{sec:s4andltl}}
\subsection{Semantic Comparison to $\CS4$}
We designed \bml as a generalization of $\logicS4$, and we formally confirm this fact in this section.
We interpret the S4 modal operator $\Box$ as a bounded modal operator $\Box^{\succeq \exclam }$, with respect to empty resource. For formal comparison, we consider a restricted language of \bml and a language of $\logicS4$:
\begin{align*}
   \mathcal{L}_{\exclam}  \ni \bmttnt{A}, \bmttnt{B} & \Coloneqq \bmttnt{p} \mid \bmttnt{A}  \mathbin{\rightarrow}  \bmttnt{B} \mid  \Box^{\mathord{\succeq}   \mathord{\boldsymbol{!} }  }  \bmttnt{A} \text{;} &
   \mathcal{L}_{\Box}  \ni A, B & \Coloneqq \bmttnt{p} \mid A  \mathbin{\rightarrow}  B \mid  \Box A \text{.}
\end{align*}
And we define $ \lvert  \mathord{-}  \rvert  :  \mathcal{L}_{\exclam}  \to  \mathcal{L}_{\Box} $ and $ ({  \mathord{-}  })^{\mathord{\succeq} \exclam}  :  \mathcal{L}_{\Box}  \to  \mathcal{L}_{\exclam} $ as translations that swap $\Box^{\succeq  \exclam }$ and $\Box$.

\begin{definition}
  We define two functions $ \lvert  \mathord{-}  \rvert  :  \mathcal{L}_{\exclam}  \to  \mathcal{L}_{\Box} $ and $ ({  \mathord{-}  })^{\mathord{\succeq} \exclam}  :  \mathcal{L}_{\Box}  \to  \mathcal{L}_{\exclam} $ as follows:
  \begin{align*}
     \lvert \bmttnt{p} \rvert  &= \bmttnt{p}                         &  ({ \bmttnt{p} })^{\mathord{\succeq} \exclam}  &= \bmttnt{p} \\
     \lvert \bmttnt{A}  \mathbin{\rightarrow}  \bmttnt{B} \rvert  &=  \lvert \bmttnt{A} \rvert   \mathbin{\rightarrow}   \lvert \bmttnt{B} \rvert  &  ({ A  \mathbin{\rightarrow}  B })^{\mathord{\succeq} \exclam}  &=  ({ A })^{\mathord{\succeq} \exclam}   \mathbin{\rightarrow}   ({ B })^{\mathord{\succeq} \exclam}  \\
     \lvert  \Box^{\mathord{\succeq}   \mathord{\boldsymbol{!} }  }  \bmttnt{A}  \rvert  &=  \Box  \lvert \bmttnt{A} \rvert    &  ({  \Box A  })^{\mathord{\succeq} \exclam}  &=  \Box^{\mathord{\succeq}   \mathord{\boldsymbol{!} }  }   ({ A })^{\mathord{\succeq} \exclam}  
  \end{align*}
\end{definition}

We examine the correspondence between \bml and $\logicS4$ semantically through these translations.

\Textcite{alechina+2001categorical} introduced
a birelational model for \emph{$\CS4$},
a \emph{constructive} variant of $\logicS4$:

\begin{definition}[$\CS4$-Model~\cite{alechina+2001categorical}]
  A\/ \emph{$\CS4$-model}\footnote{%
    For simplicity, we omit \emph{fallible worlds} from the definition
    because $\bot$ is not considered in this paper, but
    our model can be extended to having them as well.%
  }
  is a quadruple\/ $\tuple{W, \mathord{⪯}, R, V}$, where
  \begin{itemize}
    \item
      $\tuple{W, \mathord{⪯}}$ is a nonempty preordered set,
    \item
      $R$ is a preorder on~$W$ with condition:
      \begin{itemize}
        \item
          \emph{Left-persistency:}
          $\parens{R \mathrel{;} \mathord{⪯}} \subseteq \parens{\mathord{⪯} \mathrel{;} R}$,
          and
      \end{itemize}
    \item
      $V$ assigns each atom~$p$ to an upward-closed subset of~$W$.
  \end{itemize}
\end{definition}
Then we define satisfaction of a modal operator $ M , \bmttnt{w}  \vDash_{\CS4}   \Box A  $ as $\forall v. w ⪯ v \Rightarrow \forall u. v \mathrel{R} u \Rightarrow  M , \bmttnt{u}  \vDash_{\CS4}  A $.
We show that \bml models and $\CS4$ models can be translated into each other.
From $\CS4$ models to \bml models, we take steps as described below, which we call \emph{one-point-model construction}:
\begin{lemma*}[Stabilization]\label{claim:stabilization}
  Given a\/ $\CS4$-model $M =  \langle W ,  \preceq ,   R  ,  V \rangle $. Define\/ $ \sqsubseteq $ as\/ $ \mathrel{(\mathord{  \mathrel{\mathord{ \preceq } \mathbin{;} \mathord{  R  } }  }) } $. Then $ M ^{*}  =  \langle W ,  \preceq ,  \sqsubseteq ,  V \rangle $ is a stable\/ $\CS4$-model.
\end{lemma*}
\begin{proof}
  The transitivity of $\mathord{⊑}$ follows from the left-persistency of~$R$,
  and the stability follows from reflexivity of $R$.
\end{proof}

\begin{lemma}\label{claim:stabilization:equivalence}
  Given a\/ $\CS4$-model~$M$.
  For any $A ∈ ℒ_{□}$, the following are equivalent:
  \begin{enumerate}
    \item $M, w ⊨_{\CS4} A$;
    \item $M^{*}, w ⊨_{\CS4} A$.
  \end{enumerate}
\end{lemma}

\begin{proof}
  By induction on $A$. Here we check the case $A ≡ □B$:
  \begin{alignat*}{3}
          &M, w ⊨_{\CS4} □B \\
    \iff\null& ∀ v ⪰ w. ∀ u ∈ R(v).\null&&\parens[\big]{M, &&u ⊨_{\CS4} B} \\
    \iff\null& ∀ v ⪰ w. ∀ u ⊒ v.        &&\parens[\big]{M, &&u ⊨_{\CS4} B}
    \tag{\textasteriskcentered}\label{equation:stabilization} \\
    \iff\null& ∀ v ⪰ w. ∀ u ⊒ v.        &&\parens[\big]{M^{*},\null &&u ⊨_{\CS4} B}
    \tag{by IH} \\
    \iff\null& M^{*}, w ⊨_{\CS4} □B
  \end{alignat*}
  where the left-to-right direction of \eqref{equation:stabilization}
  follows from transitivity of $\mathord{⪯}$, and
  the converse follows from reflexivity of $\mathord{⪯}$.
\end{proof}

\begin{lemma*}[Root-Extension]\label{claim:root-extension}
  Given a stable\/ $\CS4$-model $M =  \langle W ,  \preceq ,  \sqsubseteq ,  V \rangle $. Define $ M _{!}  =  \langle  W _{  \exclam  }  ,   \preceq _{  \exclam  }  ,   \sqsubseteq _{  \exclam  }  ,   V _{  \exclam  }  \rangle $ as follows:
  \begin{itemize}
  \item $ W _{  \exclam  }  = W \amalg \{ \exclam \}$;
  \item $\bmttnt{w} \,  \preceq _{  \exclam  }  \, \bmttnt{v} \iff \bmttnt{w} =  \exclam  \Or \bmttnt{w} \, \preceq \, \bmttnt{v}$;
  \item $\bmttnt{w} \,  \sqsubseteq _{  \exclam  }  \, \bmttnt{v} \iff \bmttnt{w} =  \exclam  \Or \bmttnt{w} \, \sqsubseteq \, \bmttnt{v}$;
  \item $\bmttnt{w} \, \in \,  V _{  \exclam  }   \bmttsym{(}  \bmttnt{p}  \bmttsym{)} \iff V  \bmttsym{(}  \bmttnt{p}  \bmttsym{)} = W$ if $\bmttnt{w} =  \exclam $, and $\bmttnt{w} \, \in \, V  \bmttsym{(}  \bmttnt{p}  \bmttsym{)}$ otherwise.
  \end{itemize}
  Then $ M _{!} $ is a stable\/ $\CS4$-model with a root\/ $ \exclam $, namely, a\/ \bml-structure.
\end{lemma*}

\begin{proof}
  Straightforward.
\end{proof}

\begin{lemma}\label{claim:root-extension:equivalence}
  Given a stable\/ $\CS4$-model~$M$.
  For any $A ∈ ℒ_{□}$, the following are equivalent:
  \begin{enumerate}
    \item $M, w ⊨_{\CS4} A$;
    \item $M_{\tp}, w ⊨_{\CS4} A$.
  \end{enumerate}
\end{lemma}

\begin{proof}
  By induction on $A$.
\end{proof}

\begin{lemma*}[One-Point Model]\label{claim:one-point-model}
  Given a\/ \bml-structure $M$. Define $ { M }_{*} $ as
  \[  \langle \bmttsym{\{}  \bmttsym{*}  \bmttsym{\}} ,   \{\langle \bmttsym{*} , \bmttsym{*} \rangle\}  ,   \{ \bmttsym{*} \mapsto M  \}  \rangle \text{.} \]
  Then $ { M }_{*} $ is a\/ \bml-model.
\end{lemma*}

\begin{proof}
  Obvious.
\end{proof}

\begin{lemma}\label{claim:one-point-model:equivalence}
  Given a $\theLogic$-structure~$M$.
  Define a $∗$-assignment~$\tp$ for the one-point model~$M_{∗}$ as $\tp ↦ \tp$.
  Then, for any $A ∈ ℒ_{□}$, the following are equivalent:
  \begin{enumerate}
    \item $M, w ⊨_{\CS4} A$;
    \item $M_{∗}, ∗, w ⊩^{\tp} \toB|A|$.
  \end{enumerate}
\end{lemma}

\begin{proof}
  By induction on $A$. There are three cases:
  \begin{case}[$A ≡ α$]
    \begin{align*}
               & M, w ⊨_{\CS4} α \\
      \iff\null& w ∈ V(α) \\
      \iff\null& M_{∗}, ∗, w ⊩^{\tp} \toB|α|
    \end{align*}
  \end{case}
  \begin{case}[$A ≡ B \to C$]
    \begin{alignat*}{6}
      &\mathrlap{M, w ⊨_{\CS4} B \to C} \\
      \iff\null& ∀ v ⪰ w. \parens[\big]{
        M, &&v ⊨_{\CS4}\null &&B &&\implies
        M, &&v ⊨_{\CS4}\null &&C
      }
      \\
      \iff\null& ∀ v ⪰ w. \parens[\big]{
        M_{∗}, ∗,\null &&v ⊩^{\tp} &&\toB|B| &&\implies
        M_{∗}, ∗,\null &&v ⊩^{\tp} &&\toB|C|
      }
      \tag{by IH}
      \\
      \iff&
      \mathrlap{M_{∗}, ∗, w ⊩^{\tp} \toB|(B \to C)|}
    \end{alignat*}
  \end{case}
  \begin{case}[$A ≡ {□ B}$]
    \begin{alignat*}{4}
      &M, w ⊨_{\CS4} □B
      \\
      \iff\null& ∀ v ⪰ w. ∀ u ⊒ v.\null &&\parens[\big]{M, &&u ⊨_{\CS4}\null &&B}
      \\
      \iff\null& ∀ u ⊒ w. &&\parens[\big]{M, &&u ⊨_{\CS4} &&B}
      \tag{$⊑$ is stable}
      \\
      \iff\null& ∀ u ⊒ w. &&\parens[\big]{M_{∗}, ∗,\null &&u ⊩^{\tp} &&\toB|B|}
      \tag{by IH}
      \\
      \iff\null& \mathrlap{M_{∗}, ∗, w ⊩^{\tp} □^{⪰ \tp} \toB|B|}
      \tag*{\qedhere}
    \end{alignat*}
  \end{case}
\end{proof}

By combining these steps, we can construct a \bml-model from a $\CS4$-model. We can confirm that such a \bml-model behaves equivalently to the original $\CS4$-model:

\begin{theorem*}\label{claim:CS4-to-BML}
  Given a\/ $\CS4$-model $M =  \langle W ,  \preceq ,   R  ,  V \rangle $. Define $\mathfrak{M}$ as $ { \bmttsym{(}   M ^{*}   \bmttsym{)} }_{!*} $ and a\/ $\bmttsym{*}$-assignment $ \exclam $ for\/ $\mathfrak{M}$ as\/ $ \exclam  \mapsto  \exclam $. Then\/ $\mathfrak{M}$ is a\/ \bml-model, and for any $A \in  \mathcal{L}_{\Box} $, the following are equivalent:
  \begin{itemize}
  \item $ M , \bmttnt{w}  \vDash_{\CS4}  A $;
  \item $ \mathfrak{M} ,  \bmttsym{*} ,  \bmttnt{w}    \Vdash  ^{  \exclam  }   ({ A })^{\mathord{\succeq} \exclam}  $.
  \end{itemize}
\end{theorem*}
\begin{proof}
  It follows that $𝔐$ is a $\theLogic$-model from
  \Zcref{claim:stabilization,claim:root-extension,claim:one-point-model}, and
  for any $A ∈ ℒ_{□}$, we have
  \begin{alignat*}{4}
              & M, &&w ⊨_{\CS4}\null && A \\
     \iff\null& M^{*},        \null && w ⊨_{\CS4}\null &  & A
     \tag{by \Zcref{claim:stabilization:equivalence}}       \\
     \iff\null& (M^{*})_{\tp},\null && w ⊨_{\CS4}      &  & A
     \tag{by \Zcref{claim:root-extension:equivalence}}   \\
     \iff\null& 𝔐, ∗,         \null && w ⊩^{\tp}       &  & \toB|A|
     \tag{by \Zcref{claim:one-point-model:equivalence}}
     \\
     \tag*{\qedhere}
  \end{alignat*}
\end{proof}

Conversely, we construct a $\CS4$-model from a \bml-model via \emph{flattening}.

\begin{definition}[Flattening]\label{definition:flattening}
  Let\/ $\mathfrak{M} =  \langle W ,   \preccurlyeq  ,   \lbrace  M _{ \bmttnt{w} }  \rbrace_{ \bmttnt{w} \in W }  \rangle $ be a\/ \bml-model, with each $ M _{ \bmttnt{w} } $ as\/ $ \langle  D _{ \bmttnt{w} }  ,   \preceq _{ \bmttnt{w} }  ,   \sqsubseteq _{ \bmttnt{w} }  ,   V _{ \bmttnt{w} }  ,    \exclam  _{ \bmttnt{w} }  \rangle $. Then a\/ $\CS4$-model $ \mathfrak{M} _{+}  =  \langle  W _{+}  ,   \preceq _{+}  ,    R  _{+}  ,   V _{+}  \rangle $ is defined as follows:
  \begin{itemize}
  \item $ W _{+}  = \sum_{\bmttnt{w} \, \in \, W} D _{ \bmttnt{w} } $;
  \item $ \langle \bmttnt{w} , \bmttnt{d} \rangle  \,  \preceq _{+}  \,  \langle \bmttnt{w'} , \bmttnt{d'} \rangle  \iff \bmttnt{w} \,  \preccurlyeq  \, \bmttnt{w'}$ and $\bmttnt{d} \,  \preceq _{ \bmttnt{w'} }  \, \bmttnt{d'}$;
  \item $ \langle \bmttnt{w} , \bmttnt{d} \rangle  \,   R  _{+}  \,  \langle \bmttnt{w'} , \bmttnt{d'} \rangle  \iff \bmttnt{w} = \bmttnt{w'}$ and $\bmttnt{d} \,  \sqsubseteq _{ \bmttnt{w'} }  \, \bmttnt{d'}$;
  \item $ \langle \bmttnt{w} , \bmttnt{d} \rangle  \, \in \,  V _{+}   \bmttsym{(}  \bmttnt{p}  \bmttsym{)} \iff \bmttnt{d} \, \in \,  V _{ \bmttnt{w} }   \bmttsym{(}  \bmttnt{p}  \bmttsym{)}$.
  \end{itemize}
\end{definition}

\begin{lemma*}\label{claim:flattening}
  $ \mathfrak{M} _{+} $ is a\/ $\CS4$-model.
\end{lemma*}

\begin{proof}
  Most of the conditions are straightforward.
  The transitivity of $\mathord{⪯_+}$ follows from
  that of $\mathord{⪯_u}$:
  if $⟨w, d⟩ ⪯_+ ⟨v, e⟩$ and
  $⟨v, e⟩ ⪯_+ ⟨u, f⟩$, then
  $⟨w, d⟩ ⪯_+ ⟨u, f⟩$:
  \[
    \includegraphics[
      page=1,
      alt={%
        The intuitionistic relation of the flattened model
        satisfies transitivity.%
      },
    ]{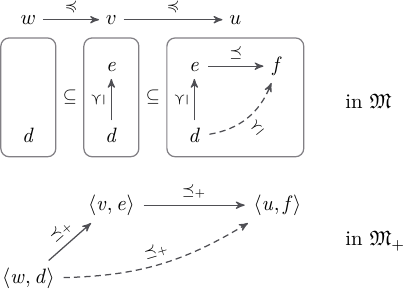}
  \]
  Left-persistency follows from right-stability:
  if $⟨w, d⟩ \mathrel{R_+} ⟨w, e⟩ ⪯_+ ⟨v, e'⟩$, then
  $⟨w, d⟩ ⪯_+ ⟨v, d⟩ \mathrel{R_+} ⟨v, e'⟩$:
  \[
    \includegraphics[
      page=2,
      alt={%
        The modal relation of the flattened model
        satisfies left-persistency.%
      },
    ]{flattened-models.pdf}
    \qedhere
  \]
\end{proof}

\begin{theorem*}\label{claim:flattening:equivalence}
  Given a\/ \bml-model\/ $\mathfrak{M}$. For any $\bmttnt{A} \in  \mathcal{L}_{\exclam} $, the following are equivalent:
  \begin{itemize}
  \item $ \mathfrak{M} ,  \bmttnt{w} ,  \bmttnt{d}    \Vdash  ^{ \rho }  \bmttnt{A} $;
  \item $  \mathfrak{M} _{+}  ,  \langle \bmttnt{w} , \bmttnt{d} \rangle   \vDash_{\CS4}   \lvert \bmttnt{A} \rvert  $.
  \end{itemize}
\end{theorem*}

\begin{proof}
  By induction on $A$. There are three cases:
  \begin{case}[$A ≡ α$]
    \begin{align*}
      &\hphantom{\null\iff\null}
      𝔐, w, d ⊩^{ρ} α \\
      &\iff ∀ v ≽ w. d ∈ V_v(α) \\
      &\iff d ∈ V_w(α) \tag{$V_w(α)$ is increasing} \\
      &\iff ⟨w, d⟩ ∈ V_+(α) \\
      &\iff 𝔐_+, ⟨w, d⟩ ⊨_{\CS4} \toF|α|
    \end{align*}
  \end{case}
  \begin{case}[$A ≡ B \to C$]\setlength{\multlinegap}{0.5em}
    \begin{align*}
      &\hphantom{\null\iff\null}
      𝔐, w, d ⊩^{ρ} B \to C
      \\
      &\iff
      ∀ v ≽ w. ∀ e ⪰_v d. \left\{
        \begin{multlined}
          𝔐, v, e ⊩^{ρ} B \\[-\jot]
          \implies
          𝔐, v, e ⊩^{ρ} C
        \end{multlined}
      \right.
      \\
      &\iff
      ∀ v ≽ w. ∀ e ⪰_v d. \left\{
        \begin{multlined}
          𝔐_+, ⟨v, e⟩ ⊨_{\CS4} \toF|B| \\
          \mathopen{}\implies
          𝔐_+, ⟨v, e⟩ ⊨_{\CS4} \toF|C|
        \end{multlined}
      \right.
      \tag{IH}
      \\
      &\iff
      ∀ ⟨v, e⟩ ⪰_+ ⟨w, d⟩. \left\{
        \begin{multlined}
          𝔐_+, ⟨v, e⟩ ⊨_{\CS4} \toF|B| \\
          \mathopen{}\implies
          𝔐_+, ⟨v, e⟩ ⊨_{\CS4} \toF|C|
        \end{multlined}
      \right.
      \\
      &\iff
      𝔐_+, ⟨w, d⟩ ⊨_{\CS4} \toF|B \to C|
    \end{align*}
  \end{case}
  \begin{case}[$A ≡ {□^{⪰ \tp} B}$]
    \begin{align*}
      &\hphantom{\null\iff\null}
      𝔐, w, d ⊩^{ρ} □^{⪰ \tp} B
      \\
      &\iff
      ∀ v ≽ w. ∀ e ⊒_v d.
        \parens*{𝔐, v, e ⊩^{ρ} B}
      \\
      &\iff
      ∀ v ≽ w. ∀ e ⊒_v d.
        \parens[\Big]{𝔐_+, ⟨v, e⟩ ⊨_{\CS4} \toF|B|}
      \tag{IH}
      \\
      \begin{split}
        &\iff
        ∀ ⟨v, d'⟩ ⪰_+ ⟨w, d⟩. \\[-\jot]
        &\qquad\qquad
        ∀ ⟨v, e⟩ ∈ \mathrel{R_+}(⟨v, d'⟩).
        \parens[\Big]{𝔐_+, ⟨v, e⟩ ⊨_{\CS4} \toF|B|}
      \end{split}
      \tag{\textdagger}\label{equation:flattening:box}
      \\
      &\iff
      𝔐_+, ⟨w, d⟩ ⊨_{\CS4} \toF|□^{⪰ \tp} B|
    \end{align*}
    where
    the left-to-right direction of \eqref{equation:flattening:box} follows
    from left-stability:
    \[
      \includegraphics[
        page=3,
        alt={%
          The left-stability of each \bml-structure ensures
          the equivalence of satisfaction.%
        },
      ]{flattened-models.pdf}
      \qedhere
    \]
  \end{case}
\end{proof}

These theorems indicate that the two constructions, namely $ { \bmttsym{(}   M ^{*}   \bmttsym{)} }_{!*} $ and $ \mathfrak{M} _{+} $, are pseudo-inverse operations that preserve satisfaction, leading to the following characterization:

\begin{theorem*}\label{claim:semantic-isomorphism-appendix}
  The $ \mathcal{L}_{\Box} $-fragment of\/ $\CS4$ is isomorphic to the $ \mathcal{L}_{\exclam} $-fragment of\/ \bml up to logical equivalence.
\end{theorem*}

\begin{proof}
  Follows from \Zcref{claim:CS4-to-BML,claim:flattening:equivalence}, where
  $\toF|\mathord{-}|$ and $\toB|(\mathord{-})|$ are the isomorphisms.
\end{proof}

This theorem corresponds to \Zcref{claim:semantic-isomorphism} in the main text.

\subsection{Comparison to Kripke-style S4 Modal Lambda Calculus}
\Textcite{journals/jacm/DaviesP01} provided a modal lambda calculus that corresponds to their Kripke-style natural-deduction proof system. For convenience, we call it \lambox in this paper. \Zcref{fig:lamboxsynxtax} provides the definition of \lambox. By forgetting proof terms, we obtain a natural-deduction proof system.

\begin{figure}[bpt]
 \begin{tabular}{clll}
  \textbf{Variables} & $\bmttmv{x}, \bmttmv{y}$ && \\
  \textbf{Types}     & $A, B$ & $\Coloneqq$ & $\bmttnt{p} \mid A  \mathbin{\rightarrow}  B \mid  \Box A  $ \\
  \textbf{Terms}     & $ {M }, {N }$ & $\Coloneqq$ & $\bmttmv{x} \mid  \lambda  \bmttmv{x} ^{ A } .  {M }  \mid  {M }   {N } $ \\
  && $\mid$ & $\mathbf{box} \, \bmttsym{\{}  {M }  \bmttsym{\}} \mid  \mathbf{unbox} _{ \bmttnt{k} }\{ {M } \}  $ \\
   \textbf{Context} & $ {\Gamma }, {\Delta }$ & $\Coloneqq$ & $ \varepsilon  \mid {\Gamma }  \bmttsym{,}  \bmttmv{x}  \hasType  A$ \\
   \textbf{Context Stack} & ${\Psi}$ & $\Coloneqq$ & ${\Gamma } \mid  {\Psi} \mathbin{;} {\Gamma } $  \\
 \end{tabular}
 \vspace{1em}

 \begin{mathpar}
   \begin{prooftree}
     \caption{Var}\label{rule:lambox-type-var}
     \hypo{\bmttmv{x}  \hasType  A \, \in \, {\Gamma }}
     \infer1{  {\Psi} \mathbin{;} {\Gamma }  \vdash_{\logicS4}  \bmttmv{x} \hasType A }
   \end{prooftree}
   \and
   \begin{prooftree}
     \caption{$ \mathbin{\rightarrow} $-I}\label{rule:lambox-type-to-i}
     \hypo{  {\Psi} \mathbin{;} {\Gamma }  \bmttsym{,}  \bmttmv{x}  \hasType  A  \vdash_{\logicS4}  {M } \hasType B }
     \infer1{  {\Psi} \mathbin{;} {\Gamma }  \vdash_{\logicS4}   \lambda  \bmttmv{x} ^{ A } .  {M }  \hasType A  \mathbin{\rightarrow}  B }
   \end{prooftree}
   \and
   \begin{prooftree}
     \caption{$ \mathbin{\rightarrow} $-E}\label{rule:lambox-type-to-e}
     \hypo{  {\Psi} \mathbin{;} {\Gamma }  \vdash_{\logicS4}  {M } \hasType A  \mathbin{\rightarrow}  B }
     \hypo{  {\Psi} \mathbin{;} {\Gamma }  \vdash_{\logicS4}  {N } \hasType A }
     \infer2{  {\Psi} \mathbin{;} {\Gamma }  \vdash_{\logicS4}   {M }   {N }  \hasType B }
   \end{prooftree}
   \\
   \begin{prooftree}
     \caption{$ \Box $-I}\label{rule:lambox-type-nec-i}
     \hypo{  {\Psi} \mathbin{;}  \varepsilon   \vdash_{\logicS4}  {M } \hasType A }
     \infer1{ {\Psi} \vdash_{\logicS4}  \mathbf{box} \, \bmttsym{\{}  {M }  \bmttsym{\}} \hasType  \Box A  }
   \end{prooftree}
   \and
   \begin{prooftree}
     \caption{$ \Box $-E}\label{rule:lambox-type-nec-e}
     \hypo{ {\Psi} \vdash_{\logicS4}  {M } \hasType  \Box A  }
     \infer1{  {\Psi} \mathbin{;}   {\Delta } _{ 1 }  \mathbin{;}\dotsb\mathbin{;}  {\Delta } _{ \bmttnt{k} }    \vdash_{\logicS4}   \mathbf{unbox} _{ \bmttnt{k} }\{ {M } \}  \hasType A }
   \end{prooftree}
 \end{mathpar}

 \caption{Syntax and Typing Rules of \lambox}
 \label{fig:lamboxsynxtax}
\end{figure}

$\lambda^\Box$ can be regarded as a restricted version of our lambda calculus, in which quoted code is always closed. This means that a box type $ \Box A $ corresponds to a bounded modal type with an initial classifier $ \Box^{\mathord{\succeq}   \mathord{\boldsymbol{!} }  }  \bmttnt{A} $. The whole definition of the translation is provided in~\Zcref{fig:lamboxtranslation}.

\begin{figure}[hbpt]
 \begin{align*}
  \intertext{\textbf{Terms} $ ( {M } )_{ \overrightarrow{\gamma} ,  \overrightarrow{ \tau }  }^{\mathord{\succeq} \exclam} $}
   ( \bmttmv{x} )_{ \overrightarrow{\gamma} ,  \overrightarrow{ \tau }  }^{\mathord{\succeq} \exclam}  &= \bmttmv{x} \\
   (  \lambda  \bmttmv{x} ^{ A } .  {M }  )_{ \overrightarrow{\gamma}  \bmttsym{,}  \delta_{{\mathrm{1}}} ,  \overrightarrow{ \tau }  }^{\mathord{\succeq} \exclam}  &=  \lambda \binder{ \bmttmv{x} }\has@{ \delta_{{\mathrm{2}}} }  ({ A })^{\mathord{\succeq} \exclam}  \ldotp  ( {M } )_{ \overrightarrow{\gamma}  \bmttsym{,}  \delta_{{\mathrm{2}}} ,  \overrightarrow{ \tau }  }^{\mathord{\succeq} \exclam}   \\
  & \text{where $\delta_{{\mathrm{2}}}$ is fresh} \\
   (  {M }   {N }  )_{ \overrightarrow{\gamma} ,  \overrightarrow{ \tau }  }^{\mathord{\succeq} \exclam}  &=    ( {M } )_{ \overrightarrow{\gamma} ,  \overrightarrow{ \tau }  }^{\mathord{\succeq} \exclam}     ( {N } )_{ \overrightarrow{\gamma} ,  \overrightarrow{ \tau }  }^{\mathord{\succeq} \exclam}  \\
   ( \mathbf{box} \, \bmttsym{\{}  {M }  \bmttsym{\}} )_{ \overrightarrow{\gamma} , \bmttnt{T} }^{\mathord{\succeq} \exclam}  &=  \mathbf{quo} (\binder{ \tau })\lbrace^{\binder{ \delta } \within  \mathord{\boldsymbol{!} }  }   ( {M } )_{ \bmttsym{(}  \overrightarrow{\gamma}  \bmttsym{,}  \delta  \bmttsym{)} , \bmttsym{(}  \bmttnt{T}  \bmttsym{,}  \tau  \bmttsym{)} }^{\mathord{\succeq} \exclam}  \rbrace \\
  & \qquad\text{where $\delta$ and $\tau$ are fresh}\\
   (  \mathbf{unbox} _{ \bmttnt{k} }\{ {M } \}  )_{ \bmttsym{(}  \overrightarrow{\gamma}  \bmttsym{,}    \delta _{ 0 }  ,\dotsc,  \delta _{ \bmttnt{k} }    \bmttsym{)} , \bmttsym{(}  \bmttnt{T}  \bmttsym{,}    \tau _{ 1 }  ,\dotsc,  \tau _{ \bmttnt{k} }    \bmttsym{)} }^{\mathord{\succeq} \exclam}  &=  \mathbf{unq} _{   \tau _{ 1 }  ,\dotsc,  \tau _{ \bmttnt{k} }   }\lbrace^{  \delta _{ 0 }  }  ( {M } )_{ \bmttsym{(}  \overrightarrow{\gamma}  \bmttsym{,}   \delta _{ 0 }   \bmttsym{)} , \bmttnt{T} }^{\mathord{\succeq} \exclam}  \rbrace 
 \end{align*}

 \medskip
 \textbf{Contexts} $ {\Psi}  \rightsquigarrow  \Gamma  /  \overrightarrow{\gamma} ,  \bmttnt{T} $
 \begin{mathpar}
   \begin{prooftree}
     \hypo{}
     \infer1{  \varepsilon   \rightsquigarrow   \varepsilon   /    \mathord{\boldsymbol{!} }   ,   \varepsilon  }
   \end{prooftree}
   \and
   \begin{prooftree}
     \hypo{  {\Psi} \mathbin{;} {\Gamma }   \rightsquigarrow  \Gamma  /  \bmttsym{(}  \overrightarrow{\gamma}  \bmttsym{,}  \delta_{{\mathrm{1}}}  \bmttsym{)} ,  \bmttnt{T} }
     \hypo{\delta_{{\mathrm{2}}}\ \text{is fresh}}
     \infer2{  {\Psi} \mathbin{;} {\Gamma }  \bmttsym{,}  \bmttmv{x}  \hasType  A   \rightsquigarrow  \Gamma  \bmttsym{,}   \binder{ \bmttmv{x} }\has@{ \delta_{{\mathrm{2}}} }  ({ A })^{\mathord{\succeq} \exclam}    /  \bmttsym{(}  \overrightarrow{\gamma}  \bmttsym{,}  \delta_{{\mathrm{2}}}  \bmttsym{)} ,  \bmttnt{T} }
   \end{prooftree}
   \\
   \begin{prooftree}
     \hypo{ {\Psi}  \rightsquigarrow  \Gamma  /  \overrightarrow{\gamma} ,  \bmttnt{T} }
      \hypo{\text{$\delta$ and $\tau$ are fresh} }
      \infer2{  {\Psi} \mathbin{;}  \varepsilon    \rightsquigarrow  \Gamma  \bmttsym{,}   \tau : \mathord{\blacktriangleright} ^{\binder{ \delta } \within  \mathord{\boldsymbol{!} }  }   /  \bmttsym{(}  \overrightarrow{\gamma}  \bmttsym{,}  \delta  \bmttsym{)} ,  \bmttsym{(}  \bmttnt{T}  \bmttsym{,}  \tau  \bmttsym{)} }
    \end{prooftree}
    \and
    \begin{prooftree}
      \hypo{  {\Psi} \mathbin{;}   {\Delta } _{ 1 }  \mathbin{;}\dotsb\mathbin{;}  {\Delta } _{ \bmttnt{k} }     \rightsquigarrow  \Gamma  /  \bmttsym{(}  \overrightarrow{\gamma}  \bmttsym{,}    \delta _{ 0 }  ,\dotsc,  \delta _{ \bmttnt{k} }    \bmttsym{)} ,  \bmttsym{(}  \bmttnt{T}  \bmttsym{,}    \tau _{ 1 }  ,\dotsc,  \tau _{ \bmttnt{k} }    \bmttsym{)} }
      \infer1{ {\Psi}  \rightsquigarrow  \Gamma  \bmttsym{,}   \mathord{\blacktriangleleft} _{   \tau _{ 1 }  ,\dotsc,  \tau _{ \bmttnt{k} }   }^{  \delta _{ 0 }  }   /  \bmttsym{(}  \overrightarrow{\gamma}  \bmttsym{,}   \delta _{ 0 }   \bmttsym{)} ,  \bmttnt{T} }
    \end{prooftree}
 \end{mathpar}

 \caption{Translation from \lambox to our lambda calculus}
 \label{fig:lamboxtranslation}
\end{figure}

The term translation $ ( {M } )_{ \overrightarrow{\gamma} , \bmttnt{T} }^{\mathord{\succeq} \exclam} $ carries a sequence of classifiers $\overrightarrow{\gamma}$, which represents positions for each past stage, and a modal transition witness $\bmttnt{T}$, which represents each stage transition. 
The context translation judgment
$ {\Psi}  \rightsquigarrow  \Gamma  /  \overrightarrow{\gamma} ,  \bmttnt{T} $
states that ${\Psi}$ is translated to $\Gamma$, with $\overrightarrow{\gamma}$ recording the positions of past stages and $\bmttnt{T}$ recording the stage transitions between them.

\begin{theorem}\label{claim:embed-kripke-style-cs4-full}
 If\/ $ {\Psi} \vdash_{\logicS4}  {M } \hasType A $ and\/ $ {\Psi}  \rightsquigarrow  \Gamma  /  \overrightarrow{\gamma} ,  \bmttnt{T} $, then\/ $\Gamma  \vdash   ( {M } )_{ \overrightarrow{\gamma} , \bmttnt{T} }^{\mathord{\succeq} \exclam}   \hasType   ({ A })^{\mathord{\succeq} \exclam} $ holds.
\end{theorem}

\begin{proof}
 By induction on the derivation of $ {\Psi} \vdash_{\logicS4}  {M } \hasType A $. We demonstrate the case of \ref{rule:lambox-type-nec-e}.

 \textit{Case} \ref{rule:lambox-type-nec-e}: We have a derivation
 \begin{center}
   \begin{prooftree}
     \hypo{ {\Psi} \vdash_{\logicS4}  {M }' \hasType  \Box A  }
     \infer1[\ref{rule:lambox-type-nec-e}]{  {\Psi} \mathbin{;}   {\Gamma } _{ 1 }  \mathbin{;}\dotsb\mathbin{;}  {\Gamma } _{ \bmttnt{k} }    \vdash_{\logicS4}   \mathbf{unbox} _{ \bmttnt{k} }\{ {M }' \}  \hasType A }
   \end{prooftree}
 \end{center}
 Decomposing $\overrightarrow{\gamma}$ to $\overrightarrow{\gamma}'  \bmttsym{,}    \delta _{ 0 }  ,\dotsc,  \delta _{ \bmttnt{k} }  $ and $\bmttnt{T}$ to $\bmttnt{T'}  \bmttsym{,}    \tau _{ 1 }  ,\dotsc,  \tau _{ \bmttnt{k} }  $, we derive $ {\Psi}  \rightsquigarrow  \Delta  \bmttsym{,}   \mathord{\blacktriangleleft} _{   \tau _{ 1 }  ,\dotsc,  \tau _{ \bmttnt{k} }   }^{  \delta _{ 0 }  }   /  \bmttsym{(}  \overrightarrow{\gamma}'  \bmttsym{,}   \delta _{ 0 }   \bmttsym{)} ,  \bmttnt{T'} $. Then we can apply the induction hypothesis, and get $\Delta  \bmttsym{,}   \mathord{\blacktriangleleft} _{   \tau _{ 1 }  ,\dotsc,  \tau _{ \bmttnt{k} }   }^{  \delta _{ 0 }  }   \vdash   ( {M }' )_{ \bmttsym{(}  \overrightarrow{\gamma}'  \bmttsym{,}   \delta _{ 0 }   \bmttsym{)} , \bmttnt{T'} }^{\mathord{\succeq} \exclam}   \hasType   \Box^{\mathord{\succeq}   \mathord{\boldsymbol{!} }  }   ({ A })^{\mathord{\succeq} \exclam}  $. We apply \ref{rule:type-bm-e} to derive $\Delta  \vdash   \mathbf{unq} _{   \tau _{ 1 }  ,\dotsc,  \tau _{ \bmttnt{k} }   }\lbrace^{  \delta _{ 0 }  }  ( {M }' )_{ \bmttsym{(}  \overrightarrow{\gamma}'  \bmttsym{,}   \delta _{ 0 }   \bmttsym{)} , \bmttnt{T'} }^{\mathord{\succeq} \exclam}  \rbrace   \hasType   ({ A })^{\mathord{\succeq} \exclam} $, which is what we want.
\end{proof}

We also confirm that the translation is injective. We define forgetful translation $ \lvert  \mathord{-}  \rvert_{\logicS4} $ as follows ($ \token{len}( \bmttnt{T} ) $ represents the number of atomic witnesses in $\bmttnt{T}$):
\begin{align*}
   \lvert \bmttnt{p} \rvert_{\logicS4}  &= \bmttnt{p} \\
   \lvert \bmttnt{A}  \mathbin{\rightarrow}  \bmttnt{B} \rvert_{\logicS4}  &=  \lvert \bmttnt{A} \rvert_{\logicS4}   \mathbin{\rightarrow}   \lvert \bmttnt{B} \rvert_{\logicS4}  \\
   \lvert  \Box^{\mathord{\succeq}  \gamma }  \bmttnt{A}  \rvert_{\logicS4}  &=  \Box  \lvert \bmttnt{A} \rvert_{\logicS4}   \\
   \lvert  \forall  \gamma_{{\mathrm{1}}} \within  \gamma_{{\mathrm{2}}} . \bmttnt{A}  \rvert_{\logicS4}  &=  \lvert \bmttnt{A} \rvert_{\logicS4} \\
  & \\
   \lvert \bmttmv{x} \rvert_{\logicS4}  &= \bmttmv{x} \\
   \lvert  \lambda \binder{ \bmttmv{x} }\has@{ \gamma } \bmttnt{A} \ldotp \bmttnt{M}  \rvert_{\logicS4}  &=  \lambda  \bmttmv{x} ^{  \lvert \bmttnt{A} \rvert_{\logicS4}  } .   \lvert \bmttnt{M} \rvert_{\logicS4}   \\
   \lvert  \bmttnt{M}   \bmttnt{N}  \rvert_{\logicS4}  &=   \lvert \bmttnt{M} \rvert_{\logicS4}     \lvert \bmttnt{N} \rvert_{\logicS4}   \\
   \lvert  \mathbf{quo} (\binder{ \tau })\lbrace^{\binder{ \gamma_{{\mathrm{1}}} } \within \gamma_{{\mathrm{2}}} }  \bmttnt{M} \rbrace  \rvert_{\logicS4}  &= \mathbf{box} \, \bmttsym{\{}   \lvert \bmttnt{M} \rvert_{\logicS4}   \bmttsym{\}} \\
   \lvert  \mathbf{unq} _{ \bmttnt{T} }\lbrace^{ \gamma } \bmttnt{M} \rbrace  \rvert_{\logicS4}  &=  \mathbf{unbox} _{  \token{len}( \bmttnt{T} )  }\{  \lvert \bmttnt{M} \rvert_{\logicS4}  \}  \\
   \lvert  \lambda \binder{ \gamma_{{\mathrm{1}}} }\within  \gamma_{{\mathrm{2}}} . \bmttnt{M}  \rvert_{\logicS4}  &=  \lvert \bmttnt{M} \rvert_{\logicS4} 
\end{align*}

\begin{theorem}\label{claim:embed-kripke-style-cs4-forget}
  \begin{enumerate}
    \item $ \lvert  ({ A })^{\mathord{\succeq} \exclam}  \rvert_{\logicS4}  = \bmttnt{A}$
    \item $ \lvert  ( {M } )_{ \overrightarrow{\gamma} , \bmttnt{T} }^{\mathord{\succeq} \exclam}  \rvert_{\logicS4}  = {M }$
  \end{enumerate}
\end{theorem}

\begin{proof}
  By induction on the structure of $A$ and ${M }$, respectively.
\end{proof}

Then, the $\logicS4$ part of \Zcref{claim:embed-s4-and-ltl} follows from \Zcref{claim:embed-kripke-style-cs4-full} and \Zcref{claim:embed-kripke-style-cs4-forget}.

\subsection{Comparison to \lamcirc}
\begin{figure}[bt]
 \begin{tabular}{clcl}
  \textbf{Variables}     & $\bmttmv{x}, \bmttmv{y}$ && \\
  \textbf{Staging Level} & $\bmttnt{k}$ & $\in$ & $\mathbb{N}$ \\
  \textbf{Types}         & $A^{\circ}, B^{\circ}$ & $\Coloneqq$ & $\bmttnt{p} \mid A^{\circ}  \mathbin{\rightarrow}  B^{\circ} \mid  \mdsmwhtcircle A^{\circ}  $ \\
  \textbf{Terms}         & $ {M^{\circ} }, {N^{\circ} }$ & $\Coloneqq$ & $\bmttmv{x} \mid  \lambda  \bmttmv{x} ^{ A^{\circ} } .  {M^{\circ} }  \mid  {M^{\circ} } \  {N^{\circ} } $ \\
  & & $\mid$ & $\mathbf{next} \, \bmttsym{\{}  {M^{\circ} }  \bmttsym{\}} \mid \mathbf{prev} \, \bmttsym{\{}  {M^{\circ} }  \bmttsym{\}} $ \\
  \textbf{Context}       & $ \Gamma^{\circ}, \Delta^{\circ}$ & $\Coloneqq$ & $ \varepsilon  \mid \Gamma^{\circ}  \bmttsym{,}   \bmttmv{x} :^{ \bmttnt{k} } A^{\circ}  $ \\
 \end{tabular}
 \vspace{1em}

 \raggedright
 \fbox{\mbox{$ \Gamma^{\circ} \vdash_{ \bmttnt{k} }  {M^{\circ} } \hasType A^{\circ} $}}
 \vspace{0.5em}

 \begin{mathpar}
  \begin{prooftree}
    \caption{Var}\label{rule:lamcirc-t-var}
    \hypo{  \bmttmv{x} :^{ \bmttnt{k} } A^{\circ}  \in \Gamma^{\circ} }
    \infer1{ \Gamma^{\circ} \vdash_{ \bmttnt{k} }  \bmttmv{x} \hasType A^{\circ} }
  \end{prooftree}
  \and
  \begin{prooftree}
    \caption{$ \mathbin{\rightarrow} $-I}\label{rule:lamcirc-t-func}
    \hypo{ \Gamma^{\circ}  \bmttsym{,}   \bmttmv{x} :^{ \bmttnt{k} } A^{\circ}  \vdash_{ \bmttnt{k} }  {M^{\circ} } \hasType B^{\circ} }
    \infer1{ \Gamma^{\circ} \vdash_{ \bmttnt{k} }   \lambda  \bmttmv{x} ^{ A^{\circ} } .  {M^{\circ} }  \hasType A^{\circ}  \mathbin{\rightarrow}  B^{\circ} }
  \end{prooftree}
  \and
  \begin{prooftree}
    \caption{$ \mathbin{\rightarrow} $-E}\label{rule:lamcirc-t-app}
    \hypo{ \Gamma^{\circ} \vdash_{ \bmttnt{k} }  {M^{\circ} } \hasType A^{\circ}  \mathbin{\rightarrow}  B^{\circ} }
    \hypo{ \Gamma^{\circ} \vdash_{ \bmttnt{k} }  {N^{\circ} } \hasType A^{\circ} }
    \infer2{ \Gamma^{\circ} \vdash_{ \bmttnt{k} }   {M^{\circ} } \  {N^{\circ} }  \hasType B^{\circ} }
  \end{prooftree}
  \\
  \begin{prooftree}
    \caption{$ \mdsmwhtcircle $-I}\label{rule:lamcirc-t-next}
    \hypo{ \Gamma^{\circ} \vdash_{ \bmttnt{k}  \bmttsym{+}  1 }  {M^{\circ} } \hasType A^{\circ} }
    \infer1{ \Gamma^{\circ} \vdash_{ \bmttnt{k} }  \mathbf{next} \, \bmttsym{\{}  {M^{\circ} }  \bmttsym{\}} \hasType  \mdsmwhtcircle A^{\circ}  }
  \end{prooftree}
  \and
  \begin{prooftree}
    \caption{$ \mdsmwhtcircle $-E}\label{rule:lamcirc-t-prev}
    \hypo{ \Gamma^{\circ} \vdash_{ \bmttnt{k} }  {M^{\circ} } \hasType  \mdsmwhtcircle A^{\circ}  }
    \infer1{ \Gamma^{\circ} \vdash_{ \bmttnt{k}  \bmttsym{+}  1 }  \mathbf{prev} \, \bmttsym{\{}  {M^{\circ} }  \bmttsym{\}} \hasType A^{\circ} }
  \end{prooftree}
 \end{mathpar}

 \raggedright
 \fbox{\mbox{${M^{\circ} }_{{\mathrm{1}}}  \Rightarrow_{\beta}  {M^{\circ} }_{{\mathrm{2}}}$}}
 \vspace{0.5em}

 \textbf{Axioms}
 \begin{align}
   \bmttsym{(}   \lambda  \bmttmv{x} ^{ A^{\circ} } .  {M^{\circ} }_{{\mathrm{1}}}   \bmttsym{)} \  {M^{\circ} }_{{\mathrm{2}}}  &  \Rightarrow_{\beta}  {M^{\circ} }_{{\mathrm{1}}}  \bmttsym{[}  \bmttmv{x}  \coloneqq  {M^{\circ} }_{{\mathrm{2}}}  \bmttsym{]} \tag*{$\beta$-$ \mathbin{\rightarrow} $}\label{rule:lamcirc-beta-imp}\\
  \mathbf{prev} \, \bmttsym{\{}  \mathbf{next} \, \bmttsym{\{}  {M^{\circ} }  \bmttsym{\}}  \bmttsym{\}} &  \Rightarrow_{\beta}  {M^{\circ} } \tag*{$\beta$-$ \mdsmwhtcircle $}\label{rule:lamcirc-beta-circ}
 \end{align}

\textbf{Compatibility Rules}

 \begin{mathpar}
  \begin{prooftree}
    \hypo{{M^{\circ} }_{{\mathrm{1}}}  \Rightarrow_{\beta}  {M^{\circ} }_{{\mathrm{2}}}}
    \infer1{ \lambda  \bmttmv{x} ^{ A^{\circ} } .  {M^{\circ} }_{{\mathrm{1}}}   \Rightarrow_{\beta}   \lambda  \bmttmv{x} ^{ A^{\circ} } .  {M^{\circ} }_{{\mathrm{2}}} }
  \end{prooftree}
  \and
  \begin{prooftree}
    \hypo{{M^{\circ} }_{{\mathrm{1}}}  \Rightarrow_{\beta}  {M^{\circ} }_{{\mathrm{2}}}}
    \infer1{ {M^{\circ} }_{{\mathrm{1}}} \  {N^{\circ} }   \Rightarrow_{\beta}   {M^{\circ} }_{{\mathrm{2}}} \  {N^{\circ} } }
  \end{prooftree}
  \and
  \begin{prooftree}
    \hypo{{M^{\circ} }_{{\mathrm{1}}}  \Rightarrow_{\beta}  {M^{\circ} }_{{\mathrm{2}}}}
    \infer1{ {N^{\circ} } \  {M^{\circ} }_{{\mathrm{1}}}   \Rightarrow_{\beta}   {N^{\circ} } \  {M^{\circ} }_{{\mathrm{2}}} }
  \end{prooftree}
  \\
  \begin{prooftree}
    \hypo{{M^{\circ} }_{{\mathrm{1}}}  \Rightarrow_{\beta}  {M^{\circ} }_{{\mathrm{2}}}}
    \infer1{\mathbf{next} \, \bmttsym{\{}  {M^{\circ} }_{{\mathrm{1}}}  \bmttsym{\}}  \Rightarrow_{\beta}  \mathbf{next} \, \bmttsym{\{}  {M^{\circ} }_{{\mathrm{2}}}  \bmttsym{\}}}
  \end{prooftree}
  \and
  \begin{prooftree}
    \hypo{{M^{\circ} }_{{\mathrm{1}}}  \Rightarrow_{\beta}  {M^{\circ} }_{{\mathrm{2}}}}
    \infer1{\mathbf{prev} \, \bmttsym{\{}  {M^{\circ} }_{{\mathrm{1}}}  \bmttsym{\}}  \Rightarrow_{\beta}  \mathbf{prev} \, \bmttsym{\{}  {M^{\circ} }_{{\mathrm{2}}}  \bmttsym{\}}}
  \end{prooftree}
 \end{mathpar}

 \removelastskip
 \caption{Detailed Definition of \lamcirc}
 \label{fig:lamcircsynxtaxfull}
\end{figure}

In this section, we show that \lamcirc, the linear-time temporal lambda calculus of Davies~\cite{journals/jacm/Davies17}, can be embedded into our calculus by a type-preserving translation. The translation follows the
construction of Murase et al.~\cite{conf/esop/MuraseNI23} for a two-level fragment, but supports the full multilevel calculus.

Typing judgments in \lamcirc simultaneously maintain a position for each stage. To emulate this behavior, our translation, defined in \Zcref{fig:translamcirc}, carries four parameters:
\begin{enumerate}
  \item Current level $\bmttnt{k}$,
  \item Upper-bound level $\bmttnt{l}$,
  \item Mapping between levels and classifiers $  \gamma _{ 0 }  ,\dotsc,  \gamma _{ \bmttnt{l} }  $, and
  \item Atomic modal transition witnesses $  \tau _{ 1 }  ,\dotsc,  \tau _{ \bmttnt{k} }  $.
\end{enumerate}
The mapping $  \gamma _{ 0 }  ,\dotsc,  \gamma _{ \bmttnt{l} }  $ assigns a classifier $ \gamma _{ \bmttnt{n} } $ to each level $\bmttnt{n}$ and represents the position maintained for that stage.
Given a mapping $\overrightarrow{\gamma}$, we write $ \overrightarrow{\gamma} [  \bmttnt{k} \mapsto \delta  ] $ to denote a new mapping that updates $\bmttnt{k}$th element of $\overrightarrow{\gamma}$ with $\delta$.

 The sequence of atomic modal transition witnesses $  \tau _{ 1 }  ,\dotsc,  \tau _{ \bmttnt{k} }  $ captures modal transitions between $  \gamma _{ 0 }  ,\dotsc,  \gamma _{ \bmttnt{k} }  $.
\begin{figure*}[hpbt]
  \textbf{Type}\enskip $ \llparenthesis A^{\circ} \rrparenthesis_{@ \bmttnt{k}  \le  \bmttnt{l} }^{   \gamma _{ \bmttnt{k}  \bmttsym{+}  1 }  ,\dotsc,  \gamma _{ \bmttnt{l} }   } $,\enskip $ \Parens{ A^{\circ} }_{@  \bmttnt{k}  \le  \bmttnt{l} }^{   \gamma _{ \bmttnt{k}  \bmttsym{+}  1 }  ,\dotsc,  \gamma _{ \bmttnt{l} }   } $
  \begin{align*}
   \llparenthesis \bmttnt{p} \rrparenthesis_{@ \bmttnt{k}  \le  \bmttnt{l} }^{ \overrightarrow{\gamma} }  = &\ \bmttnt{p}\\
   \llparenthesis A^{\circ}  \mathbin{\rightarrow}  B^{\circ} \rrparenthesis_{@ \bmttnt{k}  \le  \bmttnt{l} }^{ \overrightarrow{\gamma} }  = &\  \Parens{ A^{\circ} }_{@  \bmttnt{k}  \le  \bmttnt{l} }^{ \overrightarrow{\gamma} }   \mathbin{\rightarrow}   \llparenthesis B^{\circ} \rrparenthesis_{@ \bmttnt{k}  \le  \bmttnt{l} }^{ \overrightarrow{\gamma} }  \\
   \llparenthesis  \mdsmwhtcircle A^{\circ}  \rrparenthesis_{@ \bmttnt{k}  \le  \bmttnt{l} }^{   \gamma _{ \bmttnt{k}  \bmttsym{+}  1 }  ,\dotsc,  \gamma _{ \bmttnt{l} }   }  = &\  \Box^{\mathord{\succeq}   \gamma _{ \bmttnt{k}  \bmttsym{+}  1 }  }   \llparenthesis A^{\circ} \rrparenthesis_{@ \bmttnt{k}  \bmttsym{+}  1  \le  \bmttnt{l} }^{   \gamma _{ \bmttnt{k}  \bmttsym{+}  2 }  ,\dotsc,  \gamma _{ \bmttnt{l} }   }   \hspace{1em}\text{if $\bmttnt{k}<\bmttnt{l} $}\\
   \Parens{ A^{\circ} }_{@  \bmttnt{k}  \le  \bmttnt{l} }^{   \gamma _{ \bmttnt{k}  \bmttsym{+}  1 }  ,\dotsc,  \gamma _{ \bmttnt{l} }   }  = &\  \forall{ \binder{  \delta _{ \bmttnt{k}  \bmttsym{+}  1 }  }{\ge  \gamma _{ \bmttnt{k}  \bmttsym{+}  1 }  } }.\dotso\forall{ \binder{  \delta _{ \bmttnt{l} }  }{\ge  \gamma _{ \bmttnt{l} }  } }.   \llparenthesis A^{\circ} \rrparenthesis_{@ \bmttnt{k}  \le  \bmttnt{l} }^{   \delta _{ \bmttnt{k}  \bmttsym{+}  1 }  ,\dotsc,  \delta _{ \bmttnt{l} }   }  
  \end{align*}

  \textbf{Term}\enskip $ \llparenthesis {M^{\circ} } \rrparenthesis_{ @  \bmttnt{k}  \le  \bmttnt{l} }^{ \bmttsym{(}    \gamma _{ 0 }  ,\dotsc,  \gamma _{ \bmttnt{l} }    \bmttsym{)} ,  \bmttsym{(}    \tau _{ 1 }  ,\dotsc,  \tau _{ \bmttnt{k} }    \bmttsym{)} } $,\enskip $ \Parens{ {M^{\circ} } }_{@  \bmttnt{k}  \le  \bmttnt{l} }^{  \bmttsym{(}    \gamma _{ 0 }  ,\dotsc,  \gamma _{ \bmttnt{l} }    \bmttsym{)} ,  \bmttsym{(}    \tau _{ 1 }  ,\dotsc,  \tau _{ \bmttnt{k} }    \bmttsym{)}  } $
  \begin{align*}
       \llparenthesis \bmttmv{x} \rrparenthesis_{ @  \bmttnt{k}  \le  \bmttnt{l} }^{ \bmttsym{(}    \gamma _{ 0 }  ,\dotsc,  \gamma _{ \bmttnt{l} }    \bmttsym{)} ,  \bmttsym{(}    \tau _{ 1 }  ,\dotsc,  \tau _{ \bmttnt{k} }    \bmttsym{)} }  = &\  {  \bmttmv{x}  \gamma _{ \bmttnt{k}  \bmttsym{+}  1 }  \dotso \gamma  }_{ \bmttnt{l} }  \\
       \llparenthesis  \lambda  \bmttmv{x} ^{ A^{\circ} } .  {M^{\circ} }  \rrparenthesis_{ @  \bmttnt{k}  \le  \bmttnt{l} }^{ \bmttsym{(}    \gamma _{ 0 }  ,\dotsc,  \gamma _{ \bmttnt{l} }    \bmttsym{)} ,  \bmttsym{(}    \tau _{ 1 }  ,\dotsc,  \tau _{ \bmttnt{k} }    \bmttsym{)} }  = &\  \lambda \binder{ \bmttmv{x} }\has@{ \delta }  \Parens{ A^{\circ} }_{@  \bmttnt{k}  \le  \bmttnt{l} }^{   \gamma _{ \bmttnt{k}  \bmttsym{+}  1 }  ,\dotsc,  \gamma _{ \bmttnt{l} }   }  \ldotp  \llparenthesis {M^{\circ} } \rrparenthesis_{ @  \bmttnt{k}  \le  \bmttnt{l} }^{  \bmttsym{(}    \gamma _{ 0 }  ,\dotsc,  \gamma _{ \bmttnt{l} }    \bmttsym{)} [  \bmttnt{k} \mapsto \delta  ]  ,  \bmttsym{(}    \tau _{ 1 }  ,\dotsc,  \tau _{ \bmttnt{k} }    \bmttsym{)} }   \\
   \llparenthesis  {M^{\circ} } \  {N^{\circ} }  \rrparenthesis_{ @  \bmttnt{k}  \le  \bmttnt{l} }^{ \bmttsym{(}    \gamma _{ 0 }  ,\dotsc,  \gamma _{ \bmttnt{l} }    \bmttsym{)} ,  \bmttsym{(}    \tau _{ 1 }  ,\dotsc,  \tau _{ \bmttnt{k} }    \bmttsym{)} }  = &\   \llparenthesis {M^{\circ} } \rrparenthesis_{ @  \bmttnt{k}  \le  \bmttnt{l} }^{ \bmttsym{(}    \gamma _{ 0 }  ,\dotsc,  \gamma _{ \bmttnt{l} }    \bmttsym{)} ,  \bmttsym{(}    \tau _{ 1 }  ,\dotsc,  \tau _{ \bmttnt{k} }    \bmttsym{)} }     \Parens{ {N^{\circ} } }_{@  \bmttnt{k}  \le  \bmttnt{l} }^{  \bmttsym{(}    \gamma _{ 0 }  ,\dotsc,  \gamma _{ \bmttnt{l} }    \bmttsym{)} ,  \bmttsym{(}    \tau _{ 1 }  ,\dotsc,  \tau _{ \bmttnt{k} }    \bmttsym{)}  }  \\
       \llparenthesis \mathbf{next} \, \bmttsym{\{}  {M^{\circ} }  \bmttsym{\}} \rrparenthesis_{ @  \bmttnt{k}  \le  \bmttnt{l} }^{ \bmttsym{(}    \gamma _{ 0 }  ,\dotsc,  \gamma _{ \bmttnt{l} }    \bmttsym{)} ,  \bmttsym{(}    \tau _{ 1 }  ,\dotsc,  \tau _{ \bmttnt{k} }    \bmttsym{)} }  = &\  \mathbf{quo} (\binder{ \tau' })\lbrace^{\binder{ \delta } \within  \gamma _{ \bmttnt{k}  \bmttsym{+}  1 }  }   \llparenthesis {M^{\circ} } \rrparenthesis_{ @  \bmttnt{k}  \bmttsym{+}  1  \le  \bmttnt{l} }^{  \bmttsym{(}    \gamma _{ 0 }  ,\dotsc,  \gamma _{ \bmttnt{l} }    \bmttsym{)} [  \bmttnt{k}  \bmttsym{+}  1 \mapsto \delta  ]  ,  \bmttsym{(}    \tau _{ 1 }  ,\dotsc,  \tau _{ \bmttnt{k} }    \bmttsym{,}  \tau'  \bmttsym{)} }  \rbrace \\
      &  \quad\text{\ if $\bmttnt{k} < \bmttnt{l}$ and $\tau'$ is fresh} \\
       \llparenthesis \mathbf{prev} \, \bmttsym{\{}  {M^{\circ} }  \bmttsym{\}} \rrparenthesis_{ @  \bmttnt{k}  \le  \bmttnt{l} }^{ \bmttsym{(}    \gamma _{ 0 }  ,\dotsc,  \gamma _{ \bmttnt{l} }    \bmttsym{)} ,  \bmttsym{(}    \tau _{ 1 }  ,\dotsc,  \tau _{ \bmttnt{k} }    \bmttsym{)} }  = &\  \mathbf{unq} _{  \tau _{ \bmttnt{k} }  }\lbrace^{  \gamma _{ \bmttnt{k}  \bmttsym{-}  1 }  }  \llparenthesis {M^{\circ} } \rrparenthesis_{ @  \bmttnt{k}  \bmttsym{-}  1  \le  \bmttnt{l} }^{ \bmttsym{(}    \gamma _{ 0 }  ,\dotsc,  \gamma _{ \bmttnt{l} }    \bmttsym{)} ,  \bmttsym{(}    \tau _{ 1 }  ,\dotsc,  \tau _{ \bmttnt{k}  \bmttsym{-}  1 }    \bmttsym{)} }  \rbrace  \quad\text{if $\bmttnt{k} > 0$} \\
       \Parens{ {N^{\circ} } }_{@  \bmttnt{k}  \le  \bmttnt{l} }^{  \bmttsym{(}    \gamma _{ 0 }  ,\dotsc,  \gamma _{ \bmttnt{l} }    \bmttsym{)} ,  \bmttsym{(}    \tau _{ 1 }  ,\dotsc,  \tau _{ \bmttnt{k} }    \bmttsym{)}  }  = &\  \lambda{\binder{  \delta _{ \bmttnt{k}  \bmttsym{+}  1 }  } }{\ge  \gamma _{ \bmttnt{k}  \bmttsym{+}  1 }  }.\dotsc\lambda{\binder{  \delta _{ \bmttnt{l} }  } }{\ge  \gamma _{ \bmttnt{l} }  }.     \\
      & \qquad  \llparenthesis {N^{\circ} } \rrparenthesis_{ @  \bmttnt{k}  \le  \bmttnt{l} }^{ \bmttsym{(}    \gamma _{ 0 }  ,\dotsc,  \gamma _{ \bmttnt{k} }    \bmttsym{,}    \delta _{ \bmttnt{k}  \bmttsym{+}  1 }  ,\dotsc,  \delta _{ \bmttnt{l} }    \bmttsym{)} ,  \bmttsym{(}    \tau _{ 1 }  ,\dotsc,  \tau _{ \bmttnt{k} }    \bmttsym{)} } 
  \end{align*}

  \par\noindent
  \textbf{Context}\enskip $ \Gamma^{\circ} @ \bmttnt{k}  \rightsquigarrow_{\le  \bmttnt{l} }  \Delta  /  \bmttsym{(}    \gamma _{ 0 }  ,\dotsc,  \gamma _{ \bmttnt{l} }    \bmttsym{)} ,  \bmttsym{(}    \tau _{ 1 }  ,\dotsc,  \tau _{ \bmttnt{k} }    \bmttsym{)} $
  \begin{mathpar}
   \begin{prooftree}
     \caption{$\rightsquigarrow$-$ \varepsilon $}\label{rule:-->-empty}
     \hypo{}
     \infer1{  \varepsilon  @ 0  \rightsquigarrow_{\le  \bmttnt{l} }   \varepsilon   /   \overrightarrow{  \mathord{\boldsymbol{!} }  }  ,   \varepsilon  }
   \end{prooftree}
   \and
   \begin{prooftree}
     \caption{$\rightsquigarrow$-$ \mathord{\blacktriangleleft} $}\label{rule:-->-shut}
     \hypo{ \Gamma^{\circ} @ \bmttnt{k}  \rightsquigarrow_{\le  \bmttnt{l} }  \Delta  /  \overrightarrow{\gamma} ,  \bmttsym{(}    \tau _{ 1 }  ,\dotsc,  \tau _{ \bmttnt{k} }    \bmttsym{)} }
     \hypo{\bmttnt{k} > 0}
     \infer2{ \Gamma^{\circ} @ \bmttnt{k}  \bmttsym{-}  1  \rightsquigarrow_{\le  \bmttnt{l} }  \Delta  \bmttsym{,}   \mathord{\blacktriangleleft} _{  \tau _{ \bmttnt{k} }  }^{  \gamma _{ \bmttnt{k}  \bmttsym{-}  1 }  }   /  \overrightarrow{\gamma} ,  \bmttsym{(}    \tau _{ 1 }  ,\dotsc,  \tau _{ \bmttnt{k}  \bmttsym{-}  1 }    \bmttsym{)} }
   \end{prooftree}
   \and
   \begin{prooftree}
     \caption{$\rightsquigarrow$-V}\label{rule:-->-var}
     \hypo{ \Gamma^{\circ} @ \bmttnt{k}  \rightsquigarrow_{\le  \bmttnt{l} }  \Delta  /  \bmttsym{(}    \gamma _{ 0 }  ,\dotsc,  \gamma _{ \bmttnt{l} }    \bmttsym{)} ,  \bmttsym{(}    \tau _{ 1 }  ,\dotsc,  \tau _{ \bmttnt{k} }    \bmttsym{)} }
     \infer1{ \Gamma^{\circ}  \bmttsym{,}   \bmttmv{x} :^{ \bmttnt{k} } A^{\circ}  @ \bmttnt{k}  \rightsquigarrow_{\le  \bmttnt{l} }  \Delta  \bmttsym{,}   \binder{ \bmttmv{x} }\has@{ \delta }  \Parens{ A^{\circ} }_{@  \bmttnt{k}  \le  \bmttnt{l} }^{   \gamma _{ \bmttnt{k}  \bmttsym{+}  1 }  ,\dotsc,  \gamma _{ \bmttnt{l} }   }    /   \bmttsym{(}    \gamma _{ 0 }  ,\dotsc,  \gamma _{ \bmttnt{l} }    \bmttsym{)} [  \bmttnt{k} \mapsto \delta  ]  ,  \bmttsym{(}    \tau _{ 1 }  ,\dotsc,  \tau _{ \bmttnt{k} }    \bmttsym{)} }
   \end{prooftree}
   \\
   \begin{prooftree}
     \caption{$\rightsquigarrow$-$ \mathord{\blacktriangleright} $}\label{rule:-->-open}
     \hypo{ \Gamma^{\circ} @ \bmttnt{k}  \rightsquigarrow_{\le  \bmttnt{l} }  \Delta  /  \bmttsym{(}    \gamma _{ 0 }  ,\dotsc,  \gamma _{ \bmttnt{l} }    \bmttsym{)} ,  \bmttsym{(}    \tau _{ 1 }  ,\dotsc,  \tau _{ \bmttnt{k} }    \bmttsym{)} }
     \hypo{\bmttnt{k} < \bmttnt{l}}
     \hypo{\text{$\delta$ and $\tau'$ are fresh}}
     \infer3{ \Gamma^{\circ} @ \bmttnt{k}  \bmttsym{+}  1  \rightsquigarrow_{\le  \bmttnt{l} }  \Delta  \bmttsym{,}   \tau' : \mathord{\blacktriangleright} ^{\binder{ \delta } \within  \gamma _{ \bmttnt{k}  \bmttsym{+}  1 }  }   /   \bmttsym{(}    \gamma _{ 0 }  ,\dotsc,  \gamma _{ \bmttnt{l} }    \bmttsym{)} [  \bmttnt{k}  \bmttsym{+}  1 \mapsto \delta  ]  ,  \bmttsym{(}    \tau _{ 1 }  ,\dotsc,  \tau _{ \bmttnt{k} }    \bmttsym{,}  \tau'  \bmttsym{)} }
   \end{prooftree}
   \and
   \begin{prooftree}
     \caption{$\rightsquigarrow$-M}\label{rule:-->-mon}
     \hypo{ \Gamma^{\circ} @ \bmttnt{k}  \rightsquigarrow_{\le  \bmttnt{l} }  \Delta  /  \bmttsym{(}    \gamma _{ 0 }  ,\dotsc,  \gamma _{ \bmttnt{l} }    \bmttsym{)} ,  \bmttsym{(}    \tau _{ 1 }  ,\dotsc,  \tau _{ \bmttnt{k} }    \bmttsym{)} }
     \hypo{\text{$  \delta _{ \bmttnt{k}  \bmttsym{+}  1 }  ,\dotsc,  \delta _{ \bmttnt{l} }  $ are fresh }}
     \infer2{ \Gamma^{\circ} @ \bmttnt{k}  \rightsquigarrow_{\le  \bmttnt{l} }  \Delta  \bmttsym{,}   \binder{  \delta _{ \bmttnt{k}  \bmttsym{+}  1 }  }{ \within  \gamma _{ \bmttnt{k}  \bmttsym{+}  1 }  },\dotsc,\binder{  \delta _{ \bmttnt{l} }  }{ \within  \gamma _{ \bmttnt{l} }  }   /  \bmttsym{(}    \gamma _{ 0 }  ,\dotsc,  \gamma _{ \bmttnt{k} }    \bmttsym{,}    \delta _{ \bmttnt{k}  \bmttsym{+}  1 }  ,\dotsc,  \delta _{ \bmttnt{l} }    \bmttsym{)} ,  \bmttsym{(}    \tau _{ 1 }  ,\dotsc,  \tau _{ \bmttnt{k} }    \bmttsym{)} }
   \end{prooftree}
  \end{mathpar}
 \caption{Translation from \lamcirc (classifiers named $\delta$ or $ \delta _{ \bmttnt{k} } $ are chosen fresh)}
 \label{fig:translamcirc}
\end{figure*}

The basic idea is to annotate \lamcirc types and terms with
classifiers determined by the mapping. In particular, the translations of modal types, quotations, and unquotations use the corresponding classifiers.

However, this approach alone is not sufficient to define a sound translation. For example, consider the following \lamcirc term.
\[
  \lambda  \bmttmv{x} ^{  \mdsmwhtcircle A^{\circ}  } .  \mathbf{next} \, \bmttsym{\{}   \lambda  \bmttmv{y} ^{ B^{\circ} } .  \mathbf{prev} \, \bmttsym{\{}  \bmttmv{x}  \bmttsym{\}}   \bmttsym{\}} 
\]
When translating the type $ \mdsmwhtcircle A^{\circ} $, the translation infers its classifier from the position of the next stage. However, the classifiers inferred at the binding and use sites of x differ, making the translated types inconsistent. That is why we introduce polymorphic classifiers to address this gap, as shown in the definitions of $ \Parens{ A^{\circ} }_{@  \bmttnt{k}  \le  \bmttnt{l} }^{   \gamma _{ \bmttnt{k}  \bmttsym{+}  1 }  ,\dotsc,  \gamma _{ \bmttnt{l} }   } $, $ \Parens{ {M^{\circ} } }_{@  \bmttnt{k}  \le  \bmttnt{l} }^{  \bmttsym{(}    \gamma _{ 0 }  ,\dotsc,  \gamma _{ \bmttnt{l} }    \bmttsym{)} ,  \bmttsym{(}    \tau _{ 1 }  ,\dotsc,  \tau _{ \bmttnt{k} }    \bmttsym{)}  } $ and \ref{rule:-->-mon}. At each variable occurrence, these quantifiers are instantiated with the
classifiers corresponding to the point of use.
In this sense, polymorphic classifiers play an essential role in this translation. We can confirm that it preserves typing.

\begin{lemma}\label{lem:trans_lamcirc_position}
 If\/ $ \Gamma^{\circ} @ \bmttnt{k}  \rightsquigarrow_{\le  \bmttnt{l} }  \Delta  /  \bmttsym{(}    \gamma _{ 0 }  ,\dotsc,  \gamma _{ \bmttnt{l} }    \bmttsym{)} ,  \bmttsym{(}    \tau _{ 1 }  ,\dotsc,  \tau _{ \bmttnt{k} }    \bmttsym{)} $, then $ \mathrm{pos} ( \Delta )  =  \gamma _{ \bmttnt{k} } $.
\end{lemma}

\begin{proof}
 By induction on the derivation of $ \Gamma^{\circ} @ \bmttnt{k}  \rightsquigarrow_{\le  \bmttnt{l} }  \Delta  /  \bmttsym{(}    \gamma _{ 0 }  ,\dotsc,  \gamma _{ \bmttnt{l} }    \bmttsym{)} ,  \bmttsym{(}    \tau _{ 1 }  ,\dotsc,  \tau _{ \bmttnt{k} }    \bmttsym{)} $.
\end{proof}

\begin{lemma}\label{lem:trans_lamcirc_scope_transition}
 If the derivation $ \Gamma^{\circ}_{{\mathrm{1}}}  \bmttsym{,}  \Gamma^{\circ}_{{\mathrm{2}}} @ \bmttnt{k_{{\mathrm{1}}}}  \rightsquigarrow_{\le  \bmttnt{l} }  \Delta_{{\mathrm{1}}}  \bmttsym{,}  \Delta_{{\mathrm{2}}}  /  \bmttsym{(}    \gamma _{ 0 }  ,\dotsc,  \gamma _{ \bmttnt{l} }    \bmttsym{)} ,  \bmttsym{(}    \tau _{ 1 }  ,\dotsc,  \tau _{ \bmttnt{k_{{\mathrm{1}}}} }    \bmttsym{)} $ includes another derivation of\/ $ \Gamma^{\circ}_{{\mathrm{1}}} @ \bmttnt{k_{{\mathrm{2}}}}  \rightsquigarrow_{\le  \bmttnt{l} }  \Delta_{{\mathrm{1}}}  /  \bmttsym{(}    \delta _{ 0 }  ,\dotsc,  \delta _{ \bmttnt{l} }    \bmttsym{)} ,  \bmttsym{(}    \tau' _{ 1 }  ,\dotsc,  \tau' _{ \bmttnt{k_{{\mathrm{2}}}} }    \bmttsym{)} $, then $\Delta_{{\mathrm{1}}}  \bmttsym{,}  \Delta_{{\mathrm{2}}}  \vdash   \delta _{ \bmttnt{n} }   \preceq   \gamma _{ \bmttnt{n} } $ for all $\bmttnt{n}$ such that $0 \le \bmttnt{n} \le \bmttnt{l}$.
\end{lemma}

\begin{proof}
 By induction on the derivation of $ \Gamma^{\circ}_{{\mathrm{1}}}  \bmttsym{,}  \Gamma^{\circ}_{{\mathrm{2}}} @ \bmttnt{k_{{\mathrm{1}}}}  \rightsquigarrow_{\le  \bmttnt{l} }  \Delta_{{\mathrm{1}}}  \bmttsym{,}  \Delta_{{\mathrm{2}}}  /  \bmttsym{(}    \gamma _{ 0 }  ,\dotsc,  \gamma _{ \bmttnt{l} }    \bmttsym{)} ,  \bmttsym{(}    \tau _{ 1 }  ,\dotsc,  \tau _{ \bmttnt{k_{{\mathrm{1}}}} }    \bmttsym{)} $.
\end{proof}

\begin{lemma}\label{lem:trans_lamcirc_modal_transition}
 If\/ $ \Gamma^{\circ} @ \bmttnt{k}  \rightsquigarrow_{\le  \bmttnt{l} }  \Delta  /  \bmttsym{(}    \gamma _{ 0 }  ,\dotsc,  \gamma _{ \bmttnt{l} }    \bmttsym{)} ,  \bmttsym{(}    \tau _{ 1 }  ,\dotsc,  \tau _{ \bmttnt{k} }    \bmttsym{)} $, then $ \Delta \vdash  \tau _{ \bmttnt{n}  \bmttsym{+}  1 }  :  \gamma _{ \bmttnt{n} }  \sqsubseteq  \gamma _{ \bmttnt{n}  \bmttsym{+}  1 }  $ for any $\bmttnt{n}$ such that $0 \le \bmttnt{n} < \bmttnt{k}$.
\end{lemma}
\begin{proof}
 By induction on the derivation of $ \Gamma^{\circ} @ \bmttnt{k}  \rightsquigarrow_{\le  \bmttnt{l} }  \Delta  /  \bmttsym{(}    \gamma _{ 0 }  ,\dotsc,  \gamma _{ \bmttnt{l} }    \bmttsym{)} ,  \bmttsym{(}    \tau _{ 1 }  ,\dotsc,  \tau _{ \bmttnt{k} }    \bmttsym{)} $.
\end{proof}

\begin{lemma}\label{lem:trans_lamcirc_cls_subst_on_type}
 If $ \llparenthesis A^{\circ} \rrparenthesis_{@ \bmttnt{k}  \le  \bmttnt{l} }^{   \gamma _{ \bmttnt{k}  \bmttsym{+}  1 }  ,\dotsc,  \gamma _{ \bmttnt{l} }   } $ is defined,
 then $  \llparenthesis A^{\circ} \rrparenthesis_{@ \bmttnt{k}  \le  \bmttnt{l} }^{   \gamma _{ \bmttnt{k}  \bmttsym{+}  1 }  ,\dotsc,  \gamma _{ \bmttnt{l} }   }  [   \gamma _{ \bmttnt{k}  \bmttsym{+}  1 }  \coloneqq  \delta _{ \bmttnt{k}  \bmttsym{+}  1 }   ]\dotso[   \gamma _{ \bmttnt{l} }  \coloneqq  \delta _{ \bmttnt{l} }   ]  =  \llparenthesis A^{\circ} \rrparenthesis_{@ \bmttnt{k}  \le  \bmttnt{l} }^{   \delta _{ \bmttnt{k}  \bmttsym{+}  1 }  ,\dotsc,  \delta _{ \bmttnt{l} }   } $
\end{lemma}
\begin{proof}
 By induction on the structure of $A^{\circ}$.
\end{proof}

\begin{theorem*}\label{thm:embedlamcirc_preserve_typing}
 If\/ $ \Gamma^{\circ} \vdash_{ \bmttnt{k} }  {M^{\circ} } \hasType A^{\circ} $, $ \Gamma^{\circ} @ \bmttnt{k}  \rightsquigarrow_{\le  \bmttnt{l} }  \Delta  /  \bmttsym{(}    \gamma _{ 0 }  ,\dotsc,  \gamma _{ \bmttnt{l} }    \bmttsym{)} ,  \bmttsym{(}    \tau _{ 1 }  ,\dotsc,  \tau _{ \bmttnt{k} }    \bmttsym{)} $ and $\bmttnt{l}$ is sufficiently large, then $\Delta  \vdash   \llparenthesis {M^{\circ} } \rrparenthesis_{ @  \bmttnt{k}  \le  \bmttnt{l} }^{ \bmttsym{(}    \gamma _{ 0 }  ,\dotsc,  \gamma _{ \bmttnt{l} }    \bmttsym{)} ,  \bmttsym{(}    \tau _{ 1 }  ,\dotsc,  \tau _{ \bmttnt{k} }    \bmttsym{)} }   \hasType   \llparenthesis A^{\circ} \rrparenthesis_{@ \bmttnt{k}  \le  \bmttnt{l} }^{   \gamma _{ \bmttnt{k}  \bmttsym{+}  1 }  ,\dotsc,  \gamma _{ \bmttnt{l} }   } $.
\end{theorem*}
\begin{proof}
 By induction on the derivation of $ \Gamma^{\circ} \vdash_{ \bmttnt{k} }  {M^{\circ} } \hasType A^{\circ} $. We provide non-trivial cases for \ref{rule:lamcirc-t-var}, \ref{rule:lamcirc-t-app} and \ref{rule:lamcirc-t-prev}.

 \vspace{0.5em}
 \noindent\textit{Case} \ref{rule:lamcirc-t-var}: By inversion, we have
 \begin{center}
  \begin{prooftree}
    \infer0{ \Gamma^{\circ}_{{\mathrm{1}}}  \bmttsym{,}   \bmttmv{x} :^{ \bmttnt{k} } A^{\circ}   \bmttsym{,}  \Gamma^{\circ}_{{\mathrm{2}}} \vdash_{ \bmttnt{k} }  \bmttmv{x} \hasType A^{\circ} }
  \end{prooftree}
 \end{center}
 where $\Gamma^{\circ} = \Gamma^{\circ}_{{\mathrm{1}}}  \bmttsym{,}   \bmttmv{x} :^{ \bmttnt{k} } A^{\circ}   \bmttsym{,}  \Gamma^{\circ}_{{\mathrm{2}}}$ and ${M^{\circ} } = \bmttmv{x}$. In the derivation of $ \Gamma^{\circ} @ \bmttnt{k}  \rightsquigarrow_{\le  \bmttnt{l} }  \Delta  /  \bmttsym{(}    \gamma _{ 0 }  ,\dotsc,  \gamma _{ \bmttnt{l} }    \bmttsym{)} ,  \bmttsym{(}    \tau _{ 1 }  ,\dotsc,  \tau _{ \bmttnt{k} }    \bmttsym{)} $, we have the following derivation by \ref{rule:-->-var}, where $\Delta = \Delta_{{\mathrm{1}}}  \bmttsym{,}   \binder{ \bmttmv{x} }\has@{ \delta' }   \Parens{ A^{\circ} }_{@  \bmttnt{k}  \le  \bmttnt{l} }^{   \delta _{ \bmttnt{k}  \bmttsym{+}  1 }  ,\dotsc,  \delta _{ \bmttnt{l} }   }     \bmttsym{,}  \Delta_{{\mathrm{2}}}$ for some $\Delta_{{\mathrm{2}}}$.
 \begin{center}
  \begin{prooftree}
    \hypo{ \Gamma^{\circ}_{{\mathrm{1}}} @ \bmttnt{k}  \rightsquigarrow_{\le  \bmttnt{l} }  \Delta_{{\mathrm{1}}}  /  \bmttsym{(}    \delta _{ 0 }  ,\dotsc,  \delta _{ \bmttnt{l} }    \bmttsym{)} ,  \bmttsym{(}    \tau' _{ 1 }  ,\dotsc,  \tau' _{ \bmttnt{k} }    \bmttsym{)} }
    \infer1{ \Gamma^{\circ}_{{\mathrm{1}}}  \bmttsym{,}   \bmttmv{x} :^{ \bmttnt{k} } A^{\circ}  @ \bmttnt{k}  \rightsquigarrow_{\le  \bmttnt{l} }  \Delta_{{\mathrm{1}}}  \bmttsym{,}   \binder{ \bmttmv{x} }\has@{ \delta' }   \Parens{ A^{\circ} }_{@  \bmttnt{k}  \le  \bmttnt{l} }^{   \delta _{ \bmttnt{k}  \bmttsym{+}  1 }  ,\dotsc,  \delta _{ \bmttnt{l} }   }     /   \bmttsym{(}    \delta _{ 0 }  ,\dotsc,  \delta _{ \bmttnt{l} }    \bmttsym{)} [  \bmttnt{k} \mapsto \delta'  ]  ,  \bmttsym{(}    \tau' _{ 1 }  ,\dotsc,  \tau' _{ \bmttnt{k} }    \bmttsym{)} }
  \end{prooftree}
 \end{center}
 For such $\delta'$, we have $\Delta  \vdash  \delta'  \preceq   \gamma _{ \bmttnt{k} } $ from \Zcref{lem:trans_lamcirc_scope_transition}.
 From \Zcref{lem:trans_lamcirc_position}, we have $ \mathrm{pos} ( \Delta )  =  \gamma _{ \bmttnt{k} } $. Therefore, $\Delta  \vdash  \delta'  \preceq   \mathrm{pos} ( \Delta ) $ holds, and we can use \ref{rule:type-var} to derive $\Delta  \vdash  \bmttmv{x}  \hasType   \Parens{ A^{\circ} }_{@  \bmttnt{k}  \le  \bmttnt{l} }^{   \delta _{ \bmttnt{k}  \bmttsym{+}  1 }  ,\dotsc,  \delta _{ \bmttnt{l} }   } $. From the definition, $ \Parens{ A^{\circ} }_{@  \bmttnt{k}  \le  \bmttnt{l} }^{   \delta _{ \bmttnt{k}  \bmttsym{+}  1 }  ,\dotsc,  \delta _{ \bmttnt{l} }   }  =  \forall{ \binder{  \delta'' _{ \bmttnt{k}  \bmttsym{+}  1 }  }{\ge  \delta _{ \bmttnt{k}  \bmttsym{+}  1 }  } }.\dotso\forall{ \binder{  \delta'' _{ \bmttnt{l} }  }{\ge  \delta _{ \bmttnt{l} }  } }.   \llparenthesis A^{\circ} \rrparenthesis_{@ \bmttnt{k}  \le  \bmttnt{l} }^{   \delta'' _{ \bmttnt{k}  \bmttsym{+}  1 }  ,\dotsc,  \delta'' _{ \bmttnt{l} }   }  $. Also, $\Delta  \vdash   \delta _{ \bmttnt{n} }   \preceq   \gamma _{ \bmttnt{n} } $ for any $\bmttnt{n}$ such that $0 \le \bmttnt{n} \le \bmttnt{l}$ from \Zcref{lem:trans_lamcirc_scope_transition}. Hence, we can apply \ref{rule:type-polycls-e} as follows.
\begin{center}
 \begin{prooftree}
   \hypo{\Delta  \vdash  \bmttmv{x}  \hasType   \forall{ \binder{  \delta'' _{ \bmttnt{k}  \bmttsym{+}  1 }  }{\ge  \delta _{ \bmttnt{k}  \bmttsym{+}  1 }  } }.\dotso\forall{ \binder{  \delta'' _{ \bmttnt{l} }  }{\ge  \delta _{ \bmttnt{l} }  } }.   \llparenthesis A^{\circ} \rrparenthesis_{@ \bmttnt{k}  \le  \bmttnt{l} }^{   \delta'' _{ \bmttnt{k}  \bmttsym{+}  1 }  ,\dotsc,  \delta'' _{ \bmttnt{l} }   }  }
   \hypo{\text{$\Delta  \vdash   \delta _{ \bmttnt{n} }   \preceq   \gamma _{ \bmttnt{n} } $ for any $\bmttnt{n}$ s.t.\@ $\bmttnt{k}  \bmttsym{+}  1 \le \bmttnt{n} \le \bmttnt{l}$}}
   \infer2{\Delta  \vdash   {  \bmttmv{x}  \gamma _{ \bmttnt{k}  \bmttsym{+}  1 }  \dotso \gamma  }_{ \bmttnt{l} }   \hasType    \llparenthesis A^{\circ} \rrparenthesis_{@ \bmttnt{k}  \le  \bmttnt{l} }^{   \delta'' _{ \bmttnt{k}  \bmttsym{+}  1 }  ,\dotsc,  \delta'' _{ \bmttnt{l} }   }  [   \delta'' _{ \bmttnt{k}  \bmttsym{+}  1 }  \coloneqq  \gamma _{ \bmttnt{k}  \bmttsym{+}  1 }   ]\dotso[   \delta'' _{ \bmttnt{l} }  \coloneqq  \gamma _{ \bmttnt{l} }   ] }
 \end{prooftree}
\end{center}
We have $  \llparenthesis A^{\circ} \rrparenthesis_{@ \bmttnt{k}  \le  \bmttnt{l} }^{   \delta'' _{ \bmttnt{k}  \bmttsym{+}  1 }  ,\dotsc,  \delta'' _{ \bmttnt{l} }   }  [   \delta'' _{ \bmttnt{k}  \bmttsym{+}  1 }  \coloneqq  \gamma _{ \bmttnt{k}  \bmttsym{+}  1 }   ]\dotso[   \delta'' _{ \bmttnt{l} }  \coloneqq  \gamma _{ \bmttnt{l} }   ]  =  \llparenthesis A^{\circ} \rrparenthesis_{@ \bmttnt{k}  \le  \bmttnt{l} }^{   \gamma _{ \bmttnt{k}  \bmttsym{+}  1 }  ,\dotsc,  \gamma _{ \bmttnt{l} }   } $ from \ref{lem:trans_lamcirc_cls_subst_on_type}. Therefore this is what we want.

 \vspace{0.5em}
 \noindent\textit{Case} \ref{rule:lamcirc-t-app}: By inversion, we have
 \begin{center}
  \begin{prooftree}
    \hypo{ \Gamma^{\circ} \vdash_{ \bmttnt{k} }  {M^{\circ} }_{{\mathrm{1}}} \hasType B^{\circ}  \mathbin{\rightarrow}  A^{\circ} }
    \hypo{ \Gamma^{\circ} \vdash_{ \bmttnt{k} }  {M^{\circ} }_{{\mathrm{2}}} \hasType B^{\circ} }
    \infer2{ \Gamma^{\circ} \vdash_{ \bmttnt{k} }   {M^{\circ} }_{{\mathrm{1}}} \  {M^{\circ} }_{{\mathrm{2}}}  \hasType A^{\circ} }
  \end{prooftree}
 \end{center}
 where ${M^{\circ} } =  {M^{\circ} }_{{\mathrm{1}}} \  {M^{\circ} }_{{\mathrm{2}}} $. For the first assumption, we apply the induction hypothesis and get $\Delta  \vdash   \llparenthesis {M^{\circ} }_{{\mathrm{1}}} \rrparenthesis_{ @  \bmttnt{k}  \le  \bmttnt{l} }^{ \bmttsym{(}    \gamma _{ 0 }  ,\dotsc,  \gamma _{ \bmttnt{l} }    \bmttsym{)} ,  \bmttsym{(}    \tau _{ 1 }  ,\dotsc,  \tau _{ \bmttnt{k} }    \bmttsym{)} }   \hasType   \Parens{ B^{\circ} }_{@  \bmttnt{k}  \le  \bmttnt{l} }^{   \gamma _{ \bmttnt{k}  \bmttsym{+}  1 }  ,\dotsc,  \gamma _{ \bmttnt{l} }   }   \mathbin{\rightarrow}   \llparenthesis A^{\circ} \rrparenthesis_{@ \bmttnt{k}  \le  \bmttnt{l} }^{   \gamma _{ \bmttnt{k}  \bmttsym{+}  1 }  ,\dotsc,  \gamma _{ \bmttnt{l} }   } $. For the second assumption, we first derive $ \Gamma^{\circ} @ \bmttnt{k}  \rightsquigarrow_{\le  \bmttnt{l} }  \Delta  \bmttsym{,}   \binder{  \gamma' _{ \bmttnt{k}  \bmttsym{+}  1 }  }{ \within  \gamma _{ \bmttnt{k}  \bmttsym{+}  1 }  },\dotsc,\binder{  \gamma' _{ \bmttnt{l} }  }{ \within  \gamma _{ \bmttnt{l} }  }   /  \bmttsym{(}    \gamma _{ 0 }  ,\dotsc,  \gamma _{ \bmttnt{k} }    \bmttsym{,}    \gamma' _{ \bmttnt{k}  \bmttsym{+}  1 }  ,\dotsc,  \gamma' _{ \bmttnt{l} }    \bmttsym{)} ,  \bmttsym{(}    \tau _{ 1 }  ,\dotsc,  \tau _{ \bmttnt{k} }    \bmttsym{)} $ by \ref{rule:-->-mon}. Then we use the induction hypothesis with this context to get $\Delta  \bmttsym{,}   \binder{  \gamma' _{ \bmttnt{k}  \bmttsym{+}  1 }  }{ \within  \gamma _{ \bmttnt{k}  \bmttsym{+}  1 }  },\dotsc,\binder{  \gamma' _{ \bmttnt{l} }  }{ \within  \gamma _{ \bmttnt{l} }  }   \vdash   \llparenthesis {M^{\circ} }_{{\mathrm{2}}} \rrparenthesis_{ @  \bmttnt{k}  \le  \bmttnt{l} }^{ \bmttsym{(}    \gamma _{ 0 }  ,\dotsc,  \gamma _{ \bmttnt{k} }    \bmttsym{,}    \gamma' _{ \bmttnt{k}  \bmttsym{+}  1 }  ,\dotsc,  \gamma' _{ \bmttnt{l} }    \bmttsym{)} ,  \bmttsym{(}    \tau _{ 1 }  ,\dotsc,  \tau _{ \bmttnt{k} }    \bmttsym{)} }   \hasType   \llparenthesis B^{\circ} \rrparenthesis_{@ \bmttnt{k}  \le  \bmttnt{l} }^{   \gamma' _{ \bmttnt{k}  \bmttsym{+}  1 }  ,\dotsc,  \gamma' _{ \bmttnt{l} }   } $. We apply \ref{rule:type-polycls-i} multiple times to get $\Delta  \vdash   \Parens{ {M^{\circ} }_{{\mathrm{2}}} }_{@  \bmttnt{k}  \le  \bmttnt{l} }^{  \bmttsym{(}    \gamma _{ 0 }  ,\dotsc,  \gamma _{ \bmttnt{l} }    \bmttsym{)} ,  \bmttsym{(}    \tau _{ 1 }  ,\dotsc,  \tau _{ \bmttnt{k} }    \bmttsym{)}  }   \hasType   \Parens{ B^{\circ} }_{@  \bmttnt{k}  \le  \bmttnt{l} }^{   \gamma _{ \bmttnt{k}  \bmttsym{+}  1 }  ,\dotsc,  \gamma _{ \bmttnt{l} }   } $.
 Finally we apply \ref{rule:type-to-e} to obtain $\Delta  \vdash    \llparenthesis {M^{\circ} }_{{\mathrm{1}}} \rrparenthesis_{ @  \bmttnt{k}  \le  \bmttnt{l} }^{ \bmttsym{(}    \gamma _{ 0 }  ,\dotsc,  \gamma _{ \bmttnt{l} }    \bmttsym{)} ,  \bmttsym{(}    \tau _{ 1 }  ,\dotsc,  \tau _{ \bmttnt{k} }    \bmttsym{)} }     \Parens{ {M^{\circ} }_{{\mathrm{2}}} }_{@  \bmttnt{k}  \le  \bmttnt{l} }^{  \bmttsym{(}    \gamma _{ 0 }  ,\dotsc,  \gamma _{ \bmttnt{l} }    \bmttsym{)} ,  \bmttsym{(}    \tau _{ 1 }  ,\dotsc,  \tau _{ \bmttnt{k} }    \bmttsym{)}  }    \hasType   \llparenthesis A^{\circ} \rrparenthesis_{@ \bmttnt{k}  \le  \bmttnt{l} }^{   \gamma _{ \bmttnt{k}  \bmttsym{+}  1 }  ,\dotsc,  \gamma _{ \bmttnt{l} }   } $, which is what we want.

 \vspace{0.5em}
 \noindent\textit{Case} \ref{rule:lamcirc-t-prev}: By inversion, we have
 \begin{center}
  \begin{prooftree}
    \hypo{ \Gamma^{\circ} \vdash_{ \bmttnt{k}  \bmttsym{-}  1 }  {M^{\circ} }' \hasType  \mdsmwhtcircle A^{\circ}  }
    \infer1{ \Gamma^{\circ} \vdash_{ \bmttnt{k} }  \mathbf{prev} \, \bmttsym{\{}  {M^{\circ} }'  \bmttsym{\}} \hasType A^{\circ} }
  \end{prooftree}
 \end{center}
 where $\mathbf{prev} \, \bmttsym{\{}  {M^{\circ} }'  \bmttsym{\}} = {M^{\circ} }$. We use \ref{rule:-->-shut} to derive $ \Gamma^{\circ} @ \bmttnt{k}  \bmttsym{-}  1  \rightsquigarrow_{\le  \bmttnt{l} }  \Delta  \bmttsym{,}   \mathord{\blacktriangleleft} _{  \tau _{ \bmttnt{k} }  }^{  \gamma _{ \bmttnt{k}  \bmttsym{-}  1 }  }   /  \bmttsym{(}    \gamma _{ 0 }  ,\dotsc,  \gamma _{ \bmttnt{l} }    \bmttsym{)} ,  \bmttsym{(}    \tau _{ 1 }  ,\dotsc,  \tau _{ \bmttnt{k}  \bmttsym{-}  1 }    \bmttsym{)} $. Here, we can confirm that $\vdash  \Delta  \bmttsym{,}   \mathord{\blacktriangleleft} _{  \tau _{ \bmttnt{k} }  }^{  \gamma _{ \bmttnt{k}  \bmttsym{-}  1 }  }   \hasType \, \bmttkw{ctx}$ by \Zcref{lem:trans_lamcirc_modal_transition}. Then we apply the induction hypothesis to get $\Delta  \bmttsym{,}   \mathord{\blacktriangleleft} _{  \tau _{ \bmttnt{k} }  }^{  \gamma _{ \bmttnt{k}  \bmttsym{-}  1 }  }   \vdash   \llparenthesis {M^{\circ} }' \rrparenthesis_{ @  \bmttnt{k}  \bmttsym{-}  1  \le  \bmttnt{l} }^{ \bmttsym{(}    \gamma _{ 0 }  ,\dotsc,  \gamma _{ \bmttnt{l} }    \bmttsym{)} ,  \bmttsym{(}    \tau _{ 1 }  ,\dotsc,  \tau _{ \bmttnt{k}  \bmttsym{-}  1 }    \bmttsym{)} }   \hasType   \Box^{\mathord{\succeq}   \gamma _{ \bmttnt{k} }  }   \llparenthesis A^{\circ} \rrparenthesis_{@ \bmttnt{k}  \le  \bmttnt{l} }^{   \gamma _{ \bmttnt{k}  \bmttsym{+}  1 }  ,\dotsc,  \gamma _{ \bmttnt{l} }   }  $. Since $ \mathrm{pos} ( \Delta )  =  \gamma _{ \bmttnt{k} } $ by \Zcref{lem:trans_lamcirc_position}, we can apply \ref{rule:type-bm-e} to derive\\ $\Delta  \vdash   \mathbf{unq} _{  \tau _{ \bmttnt{k} }  }\lbrace^{  \gamma _{ \bmttnt{k}  \bmttsym{-}  1 }  }  \llparenthesis {M^{\circ} }' \rrparenthesis_{ @  \bmttnt{k}  \bmttsym{-}  1  \le  \bmttnt{l} }^{ \bmttsym{(}    \gamma _{ 0 }  ,\dotsc,  \gamma _{ \bmttnt{l} }    \bmttsym{)} ,  \bmttsym{(}    \tau _{ 1 }  ,\dotsc,  \tau _{ \bmttnt{k}  \bmttsym{-}  1 }    \bmttsym{)} }  \rbrace   \hasType   \llparenthesis A^{\circ} \rrparenthesis_{@ \bmttnt{k}  \le  \bmttnt{l} }^{   \gamma _{ \bmttnt{k}  \bmttsym{+}  1 }  ,\dotsc,  \gamma _{ \bmttnt{l} }   } $, which is what we want.

\end{proof}

We also confirm that the translation is injective. We define forgetful translation $ \lvert  \mathord{-}  \rvert_{\LTL} $ as follows:
\begin{align*}
   \lvert \bmttnt{p} \rvert_{\LTL}  &= \bmttnt{p} \\
   \lvert \bmttnt{A}  \mathbin{\rightarrow}  \bmttnt{B} \rvert_{\LTL}  &=  \lvert \bmttnt{A} \rvert_{\LTL}   \mathbin{\rightarrow}   \lvert \bmttnt{B} \rvert_{\LTL}  \\
   \lvert  \Box^{\mathord{\succeq}  \gamma }  \bmttnt{A}  \rvert_{\LTL}  &=  \mdsmwhtcircle  \lvert \bmttnt{A} \rvert_{\LTL}   \\
   \lvert  \forall  \gamma_{{\mathrm{1}}} \within  \gamma_{{\mathrm{2}}} . \bmttnt{A}  \rvert_{\LTL}  &=  \lvert \bmttnt{A} \rvert_{\LTL} 
  & \\
   \lvert \bmttmv{x} \rvert_{\LTL}  &= \bmttmv{x} \\
   \lvert  \lambda \binder{ \bmttmv{x} }\has@{ \gamma } \bmttnt{A} \ldotp \bmttnt{M}  \rvert_{\LTL}  &=  \lambda  \bmttmv{x} ^{  \lvert \bmttnt{A} \rvert_{\LTL}  } .   \lvert \bmttnt{M} \rvert_{\LTL}   \\
   \lvert  \bmttnt{M}   \bmttnt{N}  \rvert_{\LTL}  &=   \lvert \bmttnt{M} \rvert_{\LTL}  \   \lvert \bmttnt{N} \rvert_{\LTL}   \\
   \lvert  \mathbf{quo} (\binder{ \tau })\lbrace^{\binder{ \gamma_{{\mathrm{1}}} } \within \gamma_{{\mathrm{2}}} }  \bmttnt{M} \rbrace  \rvert_{\LTL}  &= \mathbf{next} \, \bmttsym{\{}   \lvert \bmttnt{M} \rvert_{\LTL}   \bmttsym{\}} \\
   \lvert  \mathbf{unq} _{ \bmttnt{T} }\lbrace^{ \gamma } \bmttnt{M} \rbrace  \rvert_{\LTL}  &= \mathbf{prev} \, \bmttsym{\{}   \lvert \bmttnt{M} \rvert_{\LTL}   \bmttsym{\}} \\
   \lvert  \lambda \binder{ \gamma_{{\mathrm{1}}} }\within  \gamma_{{\mathrm{2}}} . \bmttnt{M}  \rvert_{\LTL}  &=  \lvert \bmttnt{M} \rvert_{\LTL}  \\
   \lvert  \bmttnt{M} \gamma  \rvert_{\LTL}  &=  \lvert \bmttnt{M} \rvert_{\LTL} 
\end{align*}

\begin{theorem}\label{claim:embed-lamcirc-forget}
  \begin{enumerate}
    \item $ \lvert  \llparenthesis A^{\circ} \rrparenthesis_{@ \bmttnt{k}  \le  \bmttnt{l} }^{ \bmttsym{(}    \gamma _{ \bmttnt{k}  \bmttsym{+}  1 }  ,\dotsc,  \gamma _{ \bmttnt{l} }    \bmttsym{)} }  \rvert_{\LTL}  = A^{\circ}$
    \item $ \lvert  \llparenthesis {M^{\circ} } \rrparenthesis_{ @  \bmttnt{k}  \le  \bmttnt{l} }^{ \bmttsym{(}    \gamma _{ 0 }  ,\dotsc,  \gamma _{ \bmttnt{l} }    \bmttsym{)} ,  \bmttsym{(}    \tau _{ 1 }  ,\dotsc,  \tau _{ \bmttnt{k} }    \bmttsym{)} }  \rvert_{\LTL}  = {M^{\circ} }$
  \end{enumerate}
\end{theorem}

\begin{proof}
  By induction on the structure of $A^{\circ}$ and ${M^{\circ} }$, respectively.
\end{proof}

Then, the $\LTL$ part of \Zcref{claim:embed-s4-and-ltl} follows from \Zcref{thm:embedlamcirc_preserve_typing} and \Zcref{claim:embed-lamcirc-forget}.

\end{document}

%% file: bmtt.tex
\newcommand{\bmttnt}[1]{\mathit{#1}}
\newcommand{\bmttmv}[1]{\mathit{#1}}
\newcommand{\bmttkw}[1]{\mathbf{#1}}
\newcommand{\bmttsym}[1]{#1}
\newcommand{\bmttcom}[1]{\text{#1}}

\newcommand{\bmttrulehead}[3]{$#1$ & & $#2$ & & & \multicolumn{2}{l}{#3}}
\newcommand{\bmttprodline}[6]{& & $#1$ & $#2$ & $#3 #4$ & $#5$ & $#6$}
\newcommand{\bmttfirstprodline}[6]{\bmttprodline{#1}{#2}{#3}{#4}{#5}{#6}}

\newcommand{\bmttprodnewline}{\\}

\NewDocumentCommand\binder{m}{%
  \mathcolor{FireBrick}{#1}%
}
\newcommand\position[1]{%
  \mathcolor{RoyalBlue}{#1}%
}
\NewDocumentCommand\within{}{\mathrel{\mathord{:}\nobreak\mathord{\succeq} } }

\NewDocumentCommand\can{m}{#1\sp{\mathsf{c} } }
\NewDocumentCommand\exclam{}{\mathord{\boldsymbol{!} } }

\newcommand\token[1]{%
  \mathbf{#1}%
}

\newcommand{\bmttM}{
\bmttrulehead{\bmttnt{M}  ,\ \bmttnt{N}  ,\ v}{::=}{\bmttcom{terms}}\bmttprodnewline
\bmttfirstprodline{|}{\bmttmv{x}}{}{}{}{}\bmttprodnewline
\bmttprodline{|}{ \lambda \binder{ \bmttmv{x} }\has@{ \gamma } \bmttnt{A} \ldotp \bmttnt{M} }{}{}{}{}\bmttprodnewline
\bmttprodline{|}{ \bmttnt{M_{{\mathrm{1}}}}   \bmttnt{M_{{\mathrm{2}}}} }{}{}{}{}\bmttprodnewline
\bmttprodline{|}{ \mathbf{quo} (\binder{ \tau })\lbrace^{\binder{ \gamma } \within \bmttnt{cls} }  \bmttnt{M} \rbrace }{}{}{}{}\bmttprodnewline
\bmttprodline{|}{ \mathbf{unq} _{ \bmttnt{T} }\lbrace^{ \bmttnt{cls} } \bmttnt{M} \rbrace }{}{}{}{}\bmttprodnewline
\bmttprodline{|}{ \lambda \binder{ \gamma }\within  \bmttnt{cls} . \bmttnt{M} }{}{}{}{}\bmttprodnewline
\bmttprodline{|}{ \bmttnt{M} \bmttnt{cls} }{}{}{}{}\bmttprodnewline
\bmttprodline{|}{\bmttsym{(}  \bmttnt{M}  \bmttsym{)}}{}{}{}{}\bmttprodnewline
\bmttprodline{|}{ \bmttnt{M} }{}{}{}{}\bmttprodnewline
\bmttprodline{|}{ \bmttnt{M} [ \binder{ \gamma } \coloneqq \bmttnt{cls}  ] }{}{}{}{}\bmttprodnewline
\bmttprodline{|}{ \bmttnt{M} [ \binder{ \gamma } \coloneqq \bmttnt{cls} , \binder{ \bmttmv{x} }  \coloneqq   \bmttnt{N}  ] }{}{}{}{}\bmttprodnewline
\bmttprodline{|}{ \bmttnt{M} [ \binder{ \gamma } \coloneqq \bmttnt{cls} , \binder{ \tau }  \coloneqq   \bmttnt{T}  ] }{}{}{}{}\bmttprodnewline
\bmttprodline{|}{ \bmttnt{M} ^{ \bmttnt{T} } }{}{}{}{}\bmttprodnewline
\bmttprodline{|}{\mathcal{E}  \bmttsym{[}  \bmttnt{M}  \bmttsym{]}}{}{}{}{}\bmttprodnewline
\bmttprodline{|}{ \mathord{-} }{}{}{}{}\bmttprodnewline
\bmttprodline{|}{ { \bmttnt{M} }_{ \bmttnt{k} } }{}{}{}{}\bmttprodnewline
\bmttprodline{|}{ ( {M } )_{ \overrightarrow{\gamma} , \bmttnt{T} }^{\mathord{\succeq} \exclam} }{}{}{}{}\bmttprodnewline
\bmttprodline{|}{ \mathbf{let\ }\binder{ \bmttmv{x} }\has@{ \gamma } \bmttnt{A} \mathbf{\ =\ } \bmttnt{M} \mathbf{\ in\ } \bmttnt{N} }{}{}{}{}\bmttprodnewline
\bmttprodline{|}{ \cdot }{}{}{}{}\bmttprodnewline
\bmttprodline{|}{\bmttsym{10}}{}{}{}{}\bmttprodnewline
\bmttprodline{|}{ \lambda \binder{ \bmttmv{x} }^{ \gamma }\ldotp \bmttnt{M} }{}{}{}{}\bmttprodnewline
\bmttprodline{|}{ \llparenthesis {M^{\circ} } \rrparenthesis_{ @  \bmttnt{k}  \le  \bmttnt{l} }^{ \overrightarrow{\gamma} ,  \bmttnt{T} } } {\textsf{M}}{}{}{}\bmttprodnewline
\bmttprodline{|}{ \Parens{ {M^{\circ} } }_{@  \bmttnt{k}  \le  \bmttnt{l} }^{  \overrightarrow{\gamma} ,  \bmttnt{T}  } } {\textsf{M}}{}{}{}\bmttprodnewline
\bmttprodline{|}{ \lambda{\binder{ \bmttnt{cls_{{\mathrm{1}}}} } }{\ge \bmttnt{cls_{{\mathrm{2}}}} }.\dotsc\lambda{\binder{ \bmttnt{cls_{{\mathrm{3}}}} } }{\ge \bmttnt{cls_{{\mathrm{4}}}} }. \bmttnt{M} }{}{}{}{}\bmttprodnewline
\bmttprodline{|}{ \bmttnt{M} \bmttnt{cls_{{\mathrm{1}}}} \dotso \bmttnt{cls_{{\mathrm{2}}}} }{}{}{}{}\bmttprodnewline
\bmttprodline{|}{  }{}{}{}{}}

\newcommand{\bmttGElm}{
\bmttrulehead{\bmttnt{GElm}}{::=}{\bmttcom{elements of a typing context}}\bmttprodnewline
\bmttfirstprodline{|}{ \binder{ \bmttmv{x} }\has@{ \gamma } \bmttnt{A} }{}{}{}{}\bmttprodnewline
\bmttprodline{|}{ \tau : \mathord{\blacktriangleright} ^{\binder{ \gamma } \within \bmttnt{cls} } }{}{}{}{}\bmttprodnewline
\bmttprodline{|}{ \mathord{\blacktriangleleft} _{ \bmttnt{T} }^{ \bmttnt{cls} } }{}{}{}{}\bmttprodnewline
\bmttprodline{|}{ \binder{ \gamma } \within \bmttnt{cls} }{}{}{}{}\bmttprodnewline
\bmttprodline{|}{ \binder{ \bmttnt{cls_{{\mathrm{1}}}} }{ \within \bmttnt{cls_{{\mathrm{2}}}} },\dotsc,\binder{ \bmttnt{cls_{{\mathrm{3}}}} }{ \within \bmttnt{cls_{{\mathrm{4}}}} } } {\textsf{M}}{}{}{}}
